\documentclass[envcountsame]{elsarticle}

\usepackage{lineno,hyperref}

\usepackage{graphicx}

\usepackage{amsmath}

\usepackage{hyperref}
\usepackage{xcolor}

\usepackage{blindtext}
\usepackage{subcaption}
\graphicspath{{../}}
\DeclareGraphicsExtensions{.pdf,.eps}
\usepackage{times} % assumes new font selection scheme installed
\usepackage{amssymb}  % assumes amsmath package installed
\usepackage{amsfonts}
\usepackage{stmaryrd}
\usepackage{mathtools}
\usepackage{wrapfig} % \begin{wrapfigure}{alignment}{width} command.
\usepackage{array}
\usepackage{enumerate}
\usepackage{xspace}
\usepackage{ifthen}
\usepackage{color}
\usepackage{mathdots}
\usepackage{url}
\usepackage[newitem]{paralist}
\usepackage[para]{footmisc}
\usepackage{amsthm}  % assumes amsmath package installed
\usepackage{xspace}
\usepackage{bussproofs}
\usepackage{verbatim}
\allowdisplaybreaks

\newcommand{\Grob}{Gr\"{o}bner\xspace}
\newcommand{\KeYmaeraX}{KeYmaera~X\xspace}
\newcommand{\KX}{\KeYmaeraX}
\newcommand{\bbn}{\mathbb{N}}
\newcommand{\bbq}{\mathbb{Q}}
\newcommand{\bbr}{\mathbb{R}}
\newcommand{\bbc}{\mathbb{C}}
\newcommand{\bx}{\mathbf{x}}
\newcommand{\cl}[3][]{{\overline{#2}}^{{#3}_{#1}}}
\newcommand{\dimpt}[2]{\text{dim}_{\mathbf{#1}}(#2)}
\newcommand{\rx}{\mathbb{R}[\mathbf{x}]}
\newcommand{\cx}{\mathbb{C}[\mathbf{x}]}
\newcommand{\dkx}[1][K]{{#1}\{\mathbf{x}\}}
\newcommand{\ldf}[1]{\mathcal{L}_{\mathbf{F}}(#1)}
\newcommand{\ldfid}[1]{\mathcal{L}^*_{\mathbf{F}}(#1)}

\newcommand{\rgaexp}{{\tt RGA}_{o}}
\newcommand{\rrad}[1]{\sqrt[\mathbb{R}]{#1}}
\newcommand{\sat}[2][S]{{#2}:#1^\infty}
\newcommand{\tri}{{\tt Triangulate}\xspace}
\newcommand{\vi}[2]{\mathbf{I}_{#1}(#2)}
\newcommand{\vs}[2][]{\mathbf{V}_{#1}(#2)}
\newcommand{\vsr}[1]{\mathbf{V}_{\mathbb{R}}(#1)}
\newcommand{\vsc}[1]{\mathbf{V}_{\mathbb{C}}(#1)}
\newcommand{\xpf}{\mathbf{x'}-\mathbf{f}(\mathbf{x})}
\newcommand{\xpef}{\mathbf{x'}=\mathbf{f}(\mathbf{x})}
\newcommand{\xtpf}{\mathbf{\widetilde{x}'}-\mathbf{\widetilde{f}}(\mathbf{x})}

\newtheorem{theorem}{Theorem}

\newtheorem{corollary}[theorem]{Corollary}
\newtheorem{definition}[theorem]{Definition}
\newtheorem{example}[theorem]{Example}
\newtheorem{lemma}[theorem]{Lemma}
\newtheorem{notation}[theorem]{Notation}
\newtheorem{proposition}[theorem]{Proposition}

\theoremstyle{definition}
\newtheorem{remark}[theorem]{Remark}

\modulolinenumbers[5]

\makeatletter
\def\ps@pprintTitle{%
     \let\@oddhead\@empty
     \let\@evenhead\@empty
     \def\@oddfoot{\reset@font\hfil}
     \let\@evenfoot\@oddfoot
}
\makeatother
\begin{document}

\begin{frontmatter}

\title{Differential Elimination and Algebraic Invariants of Polynomial Dynamical Systems}
%\title{Elsevier \LaTeX\ template\tnoteref{mytitlenote}}


\author[mymainaddress,mysecondaryaddress]{William Simmons\corref{mycorrespondingauthor}}%\author[mymainaddress,mysecondaryaddress]{Elsevier Inc}
\cortext[mycorrespondingauthor]{Corresponding author}
\ead{william.simmons@twosixtech.com}

\author[mysecondaryaddress,mythirdaddress]{Andr\'e Platzer}%\author[mysecondaryaddress] {Global Customer Service\corref{mycorrespondingauthor}}
\ead{aplatzer@cs.cmu.edu}

\address[mymainaddress]{High Assurance Solutions group, Two Six Technologies, Arlington, USA}
\address[mysecondaryaddress]{Computer Science Department, Carnegie Mellon University, Pittsburgh, USA}
\address[mythirdaddress]{Department of Informatics, Karlsruhe Institute of Technology, Karlsruhe, Germany}

\begin{abstract}
Invariant sets are a key ingredient for verifying safety and other properties of cyber-physical systems that mix discrete and continuous dynamics. We adapt the elimination-theoretic Rosenfeld-\Grob algorithm to systematically obtain algebraic invariants of polynomial dynamical systems without using \Grob bases or quantifier elimination. We identify totally real varieties as an important class for efficient invariance checking. %at a lower computational cost.  %that leads to lower computational complexity than previous methods for checking invariance.
\end{abstract}

\end{frontmatter}
 
%\linenumbers
\section{Introduction}
%//describe what we will do, but needn't be as snappy as conclusion; present the scene

%[TODO:-importance of invariants: discuss safety and \KX (mention pegasus); several general related refs
%-elim thy and how exactness is beneficial (more amenable to formal verification, p.e.); mention Pegasus and need for a variety of methods, tradeoffs, cite related work; introductory remarks on the most related papers, but hold off on discussing them until related work
%-we provide an alternative elimination approach that has not been investigated before; give alternatives/new methods and remove need for old methods like GB and QE and potential benefits for complexity, proving things, etc. Worst case complexity that we could prove is bad, but no worse than the RGA competition and we have reason to think that much of the mess could be removed while still reaping the benefits of char set thy.]
 Computers are increasingly hitting the road, taking to the skies, and interacting with people and the physical environment in unprecedented ways.  Engineering concerns like realistic models and reliable sensors are critical, but just as important are the complex mathematical problems that lie at the heart of making \emph{cyber-physical systems} (CPS) safe \cite{Platzer10,Platzer18}. One of these central problems is computing \emph{invariant sets} \cite{goriely2001integrability} of continuous dynamical systems (Section \ref{lie}), where an invariant set is a collection of states from which any trajectory starting in the set will never leave. Given a system of ordinary differential equations (ODEs) and a set of unsafe states, we must identify initial states that will never lead to an unsafe state under the specified dynamics.  If an invariant does not intersect the unsafe set, then every state in the invariant set is a safe starting point. Beyond safety \cite{PlatzerT20}, invariant reasoning is important for issues like stability \cite{TanP21}, liveness \cite{tan2021axiomatic}, and control \cite{zhan2024reset}.
Moreover, recent work of Platzer and Qian \cite{platzer2024axiomatization} shows that verification of numerical ODE computation boils down to analysis of differential invariants. Due to its great practical significance, the problem of invariant generation for ODEs has received substantial interest \cite{DBLP:conf/hybrid/Rodriguez-CarbonellT05,DBLP:journals/fmsd/SankaranarayananSM08,DBLP:conf/hybrid/Tiwari08,DBLP:journals/fmsd/PlatzerC09,DBLP:conf/hybrid/Sankaranarayanan10,DBLP:conf/cav/SankaranarayananT11,DBLP:journals/sttt/TalyGT11,DBLP:conf/itp/Platzer12,DBLP:conf/etsc/WuZ12,DBLP:conf/tacas/GhorbalP14,DBLP:conf/vmcai/SogokonGJP16,DBLP:conf/sas/RouxVS16,Sankaranarayanan16,DBLP:conf/hybrid/KongBSJH17,DBLP:conf/sofsem/Boreale18,majumdar2020algebraic,DBLP:journals/jsc/GhorbalS22,zhao2023formal} and even has a dedicated tool \cite{DBLP:journals/fmsd/SogokonMTCP22}, but most methods have nontrivial heuristic search parts.

In this article we employ strategies from \emph{elimination theory} \cite{Wang01} to give new algorithms for systematically generating invariants and checking candidates for invariance. Elimination theory draws on algebra, geometry, and logic to give algorithmic procedures for understanding the polynomial systems that arise in many applications. In a narrow sense, elimination is the process of exposing explicit relationships and removing variables. The prototypical example of an elimination method is \emph{Gaussian elimination}, %\cite{stra03}, 
which converts a system of linear equations into a simpler system with the same solutions from which those solutions can be read off easily. However, linear algebraic equations are merely the tip of the iceberg. Crucially, many elimination tools are extremely general. They apply to nonlinear problems of diverse forms, allowing an unusual degree of reuse across theories and applications. In this paper, our problems require methods for symbolically analyzing systems of ODEs that have \emph{real algebraic} \cite{bochnak1998real} constraints; i.e., are defined by polynomial equations over the real numbers (Section \ref{lie}). In recent decades, researchers in \emph{differential algebra} \cite{MR0568864} have extended classical algebraic elimination methods to polynomial differential equations (Section \ref{diffalg}). %This involves topics such as ring theory, Gr\"{o}bner bases, characteristic sets and their differential generalizations \cite{todo}. 
Broadly speaking, \emph{differential elimination} \cite{LiY19} provides tools for finding \emph{all} relations implied by a polynomial differential system, regardless of the form. The elimination procedures we develop in Sections \ref{regsysresults} and \ref{rga} are intrinsically mathematical, but we show that the output corresponds directly to invariants of the input system, giving a novel, complementary view on systematically computing and checking ODE invariants. 
 
As always, there are trade-offs involved. Elimination methods are exact, general %(making them attractive targets for formal verification in safety-critical domains) 
and remarkably versatile, but used naively they can have high theoretical and practical costs \cite{Pardo95}. We emphasize alternatives within elimination theory that %potentially 
have better computational complexity than standard choices. %(See Remark \ref{difftriangrmk} and p. \pageref{trinotbad} for elaboration on the qualifier `potentially'.) 
For example, \emph{characteristic sets} \cite{mishra} carry less information than the more ubiquitous \emph{\Grob bases} \cite{clo1_4} but they have singly exponential complexity \cite{gallo1991efficient} in the number of variables instead of \Grob bases' doubly exponential growth \cite{mayr1982complexity}. We work extensively with \emph{regular systems} (Section \ref{regsyssec}), a weak form of characteristic sets with lower complexity. Similarly, real number solutions are the main goal in applications, but obtaining them is often significantly more expensive than working over the complex numbers \cite{Boreale20,becker1993computation}. We give a promising workaround that can \emph{equivalently} replace real algebraic computation with more efficient complex\footnote{Ironically, computation over complex-valued structures is often less complex than over real-valued structures.} alternatives in many situations (Sections \ref{tr} and \ref{check}). 

Computation is our main motivation for restricting to polynomial differential and algebraic equations. Such equations are general \cite{PlatzerT20}, already powerful for applications \cite{MeniniPT21} and undecidability typically ensues when we leave the polynomial setting \cite{richardson1969some}.  

\paragraph{Structure of the Paper} %\label{plan}

Section \ref{rev} reviews the algebraic geometry and differential algebra used to rigorously justify our results in Sections \ref{regsysresults} and \ref{rga}. If we do not have a reference at hand or the proof is short, we provide a proof without claiming originality.   %or at least familiar to specialists) 
%We include it for completeness and for the convenience of the reader. 
%!!mention appendix

Most of the material in Section \ref{rev} is standard and is self-contained assuming familiarity with basic linear algebra and multivariate calculus.  That said, the ideas are nontrivial and, depending on background, readers may need to review some or all of Section \ref{rev} in order to fully understand the new results and proofs in Sections \ref{regsysresults} and \ref{rga}. However, scanning Section \ref{rev} for notation and definitions %, along with the detailed examples in later sections, 
provides a good initial picture. % even without carefully reading Section \ref{rev}. 
Sections \ref{regsysresults} and \ref{rga} are the heart of the paper. In Section \ref{regsysresults} the main contributions are Theorems \ref{satinvar} and \ref{invarcor}. Theorem \ref{satinvar} identifies a novel sufficient condition for a system $(A=0,S\neq 0)$ of polynomial equations and inequations to implicitly determine an invariant. Theorem \ref{invarcor} gives multiple criteria, frequently met in practice, that allow $(A=0,S\neq 0)$ to define the associated invariant more explicitly. A detailed example (Section \ref{lorenzex1}) illustrates the necessary computations while postponing a full explanation until Section \ref{lorenzex2}.

Section \ref{rgaexp} introduces $\rgaexp$, our new variant of the Rosenfeld-\Grob algorithm (RGA) of Boulier et al. \cite{BoulierLOP95}. $\rgaexp$ takes in a system of ordinary differential equations $\xpef$ and polynomial equations $A=0$ and algorithmically extracts a structured, maximal invariant satisfying the equations. In particular, Theorem \ref{radlddecomp} shows that $\rgaexp$ produces output that meets the requirements of Theorem \ref{satinvar} for an invariant. Section \ref{lorenzex2} revisits the example of Section \ref{lorenzex1} in the context of RGA. In Section \ref{rgabds} we analyze the computational complexity of $\rgaexp$ and find an explicit upper bound (Theorem \ref{expRGA}) on the degrees of polynomials in the output. (Most prior results on RGA have either ignored complexity or focused on the order of intermediate differential equations. An exception is Corollary 11 of \cite{MR2556127}, which gives an Ackermannian bound on the degrees of polynomials returned by an RGA-like algorithm.) 
Our focus in Section \ref{check} is a simple and efficient---compared to standard methods from the literature---test for invariance (Theorem \ref{radinvartest}) based on \emph{totally real varieties}. These geometric objects provide a bridge that lets us draw conclusions about real number solutions using tools that are naturally suited to the complex numbers with their computational advantages.

Section \ref{related} compares and contrasts $\rgaexp$ with related approaches for generating and checking invariants. Section \ref{conc} summarizes our contributions and outlines promising future research questions. Finally, Section \ref{appendix} (following the references) is an appendix containing proofs not given in the body of the paper. (The corresponding theorem statements throughout the paper are marked by `(Appendix)'.)
 
%[TODO: what the main results are, where we're going and why]

\section{Mathematical Background}\label{rev}

%XTODO: where to look for algebraic or analytic prereqs not covered here]

\subsection{Ideals and Varieties}
We work extensively with multivariate polynomials over various \emph{fields} \cite{Lang:algebra} like the rational numbers $\mathbb{Q}$, the real numbers $\bbr$, and the complex numbers $\bbc$.  We often write $K$ or $L$ for an unspecified field.  (A field is essentially a set equipped with some ``addition" and ``multiplication" operations that satisfy the usual properties of arithmetic in $\mathbb{Q},\bbr$, or $\bbc$. In particular, we have the ability to divide by nonzero elements.) In Section \ref{diffalg} we consider \emph{differential fields}, which are fields having additional structure inspired by calculus. %If desired, the reader may mentally substitute one of these fields whenever an abstract field shows up.) 
We sometimes restrict to polynomials over the real or complex numbers even when results apply more generally. The fact that $\bbr$ has an ordering (given by the usual $<$ relation) is important, but we try to work algebraically as much as possible (in which case $\bbc$ is particularly convenient). Occasionally we use the Euclidean metric properties of $\bbr$ and $\bbc$ \cite{Rudin87}.

 Unless indicated otherwise, we use $n$ for the number of variables and write $\bx$ for the $n$-tuple $(x_1,x_2,\ldots,x_n)$; then the \emph{polynomial ring} $K[\bx]$ is the set of all polynomials in variables $x_1,x_2,\ldots,x_n$ over field $K$. We denote by $K^n$ the collection of all $n$-tuples of elements from $K$.
%[TODO: $(x_1,x_2,\ldots,x_n) = \bx$, $\rx$, $\cx$, fields, and definition of ideal; notation indicating which field]

\subsubsection{Ideals}
\begin{definition} \label{iddef}
Let $K$ be a field. If $I\subseteq K[\bx]$, $0\in I$, and $p+q\in I, gp\in I$ for all $p,q\in I$ and $g\in K[\bx]$, then we say $I$ is an \emph{ideal} of $K[\bx]$ and write $I\unlhd K[\bx]$. If $I\unlhd K[\bx]$ and $I\neq K[\bx]$, we call $I$ a \emph{proper ideal} and write $I \lhd K[\bx]$.

If $A\subseteq K[\bx]$, then the collection of all finite sums $\sum_{i} g_{i}p_i$, where $p_i \in A$ and $g_{i}$ is any polynomial in $K[\bx]$, is the \emph{ideal generated by $A$ in $K[\bx]$.} We write $(A)_K$ or just $(A)$ to denote this ideal of $K[\bx]$.
\end{definition}

\begin{comment}n
Our interest in ideals comes from the fact that they correspond to systems of polynomial equations. If $p\in I$, we associate $p$ with the equation $p=0$. The definition of an ideal mirrors the behavior of polynomial equation solutions: given $\mathbf{a}\in K^n$, if $p$ and $q$ are simultaneously 0 when evaluated at $\mathbf{a}$, then $p(\mathbf{a})+q(\mathbf{a})$ and $g(\mathbf{a})p(\mathbf{a})$ are also 0 for any $g\in K[\bx]$. 
\end{comment} 

Note that $(A)$ is the minimal ideal of $K[\bx]$ containing $A$. Elements of an ideal are linear combinations of generators of the ideal, although the ``coefficients" multiplied by the generators are arbitrary polynomials over $K$, not necessarily elements of $K$.

%[Example of an ideal and a (non-obviously) improper ideal]

Even though nonzero ideals in $K[\bx]$ are infinite, they can always be represented in a finite way:

\begin{theorem}[{Hilbert's basis theorem \cite[Thm. 4, p. 77]{clo1_4}}]\label{hbt}
Let $K$ be a field. Every ideal in the multivariate polynomial ring $K[\mathbf{x}]$ is finitely generated (has a finite generating set).
\end{theorem}

We generally rely on context to indicate the polynomial ring of which $I$ or $(A)$ is an ideal (namely, the ring mentioned whenever the ideal's name is introduced), but we use the following notation when changing fields:

%(?)revert to R and C if don't end up needing other field extensions
\begin{notation}
Let $K\subseteq L$ be fields and let $I\unlhd K[\bx]$. We write $I_L$ for the ideal generated by $I$ in $L[\bx]$.
\end{notation}

For instance, if $(A)_K\unlhd K[\bx]$, then $(A)_L$ is the ideal generated by $A$ (or equivalently by $(A)_K$) in $L[\bx]$.

\begin{comment}
\begin{lemma}\label{linfield}
Let $K\subseteq L$ be fields. If $p\in K[\mathbf{x}]$ and $p \in I_{L}$ for  $I \unlhd K[\mathbf{x}]$, then $p\in I$.
\end{lemma}
\begin{proof}
This is immediate from the definition of ideal membership and the linear algebra fact that \emph{linear} equations with coefficients from a field $K$ are solvable in $K$ if and only if they are solvable in the extension $L$.
\end{proof}
\end{comment}

\begin{comment}
\begin{lemma}
%\begin{thmE}\label{linfield2}[end,restate] pf-at-end attempt
Let $K\subseteq L$ be fields. If $p\in K[\mathbf{x}]$ and $p \in I_{L}$ for  $I \unlhd K[\mathbf{x}]$, then $p\in I$.
%\end{thmE} pf-at-end attempt
\end{lemma}
%\begin{proofE} pf-at-end attempt
\begin{proof}
This is immediate from the definition of ideal membership and the linear algebra fact that \emph{linear} equations with coefficients from a field $K$ are solvable in $K$ if and only if they are solvable in the extension $L$.
\end{proof}
%\end{proofE} pf-at-end attempt
\end{comment}

Two types of ideals play a particularly fundamental role in relating geometry (specifically, solution sets of polynomial equations) to algebra.

\begin{definition}
A \emph{prime} ideal $P \lhd K[\bx]$ is a proper ideal such that for all $p,q\in K[\bx]$ for which $pq\in P$, either $p\in P$ or $q \in P$.

A \emph{radical} ideal $I\unlhd K[\bx]$ has the property that for all $p\in K[\bx]$ such that $p^N \in I$ for some natural number $N$, already $p\in I$.  The \emph{radical} $\sqrt{I}$ of an arbitrary ideal $I$ is the intersection of all radical ideals containing $I$.\end{definition}

The radical $\sqrt{I}$ is a radical ideal and consists precisely of those polynomials $p$ such that $p^N\in I$ for some $N$. Ideal $I$ is radical if and only if $I=\sqrt{I}$. Prime ideals are radical, but not necessarily vice versa (for instance, $(xy)\unlhd \bbr[x,y]$ is radical but not prime).

\subsubsection{Algebraic Varieties and Constructible Sets} Sets definable by polynomial equations, as well as the corresponding ideals, are central players in our work:

\begin{definition}
Let $K$ be a field, $A\subseteq K[\bx]$, and $X\subseteq K^n$. The \emph{zero set} of $A$, denoted $\vs[K]{A}$, is the set of all $\mathbf{a}\in K^n$ such that $f(\mathbf{a})=0$ for all $f\in A$. If $X=\mathbf{V}_K(B)$ for some $B\subseteq K[\bx]$, we call $X$ a \emph{variety}, \emph{algebraic set}, or \emph{Zariski closed set} over $K$. We refer to subsets of $X$ that are themselves zero sets of polynomials in $K[\bx]$ as \emph{subvarieties}, \emph{algebraic subsets}, or \emph{Zariski closed subsets} over $K$.

Even if $X$ is not algebraic, we define the \emph{vanishing ideal} of $X$, denoted $\vi{K}{X}$, to be the collection of all $g\in K[\bx]$ such that $g(\mathbf{b})=0$ for all $\mathbf{b}\in X$. (Note that $\vi{K}{X}$ is a radical ideal.)  If $K$ is understood from context, we may write $\vs{A}$ or $\vi{}{X}$ for the zero set and vanishing ideal, respectively. If $A$ contains only one polynomial $p$, we write $\vs{p}$ instead of $\vs{A}$.%instead of $\vs[K]{\{p\}}$ or $\vs[K]{(p)}$.
\end{definition}

\begin{example}
Trivial examples of algebraic sets are the empty set and all of $K^n$. A more interesting example is $\mathbf{V}_K(\text{det}(\bx))\subseteq K^{n^2}$, where $\text{det}(\mathbf{a})$ is the determinant of the $n\times n$ matrix corresponding to $\mathbf{a}\in K^{n^2}$. The points of $\mathbf{V}_K(\text{det}(\bx))$ represent singular matrices.
\end{example}

%[Ex: $((y^2+1)x)$]

%We cite several easy but important consequences of the definitions.

%\begin{lemma}[{\cite[?]{clo1_4}}] \label{rev}
%Let $K$ be a field, $I,J\unlhd K[\bx]$, and $X,Y \subseteq K^n$. Then $\vs{I}\subseteq \vs{J}$ if and only if $I\supseteq J$, and $\vi{}{X}\subseteq \vi {}{Y}$ if and only if $X\supseteq Y$. //too strong! The reversals depend (I'm virtually certain) on the sets in K^n being varieties, and on using rad/real rad and NSS/RNSS. Replaced with agdict.
%\end{lemma}

\begin{lemma}[{\cite[Thm. 15, p. 196]{clo1_4}}] \label{capcup}
Let $K$ be a field and let $I,J\unlhd K[\bx]$. Then $\mathbf{V}_{K}(I\cap J)= \mathbf{V}_K(I)\cup \mathbf{V}_K(J)$.
\end{lemma}

\begin{lemma}[{\cite[Thm. 7(ii), p. 183]{clo1_4};} { Table \ref{galois}}] \label{viva}
Let $K$ be a field and $A\subseteq K[\bx]$. Then $\vs{\vi{}{\vs{A}}}= \vs{A}$.
\end{lemma}

\begin{remark}\label{galoisrmk}
Lemma \ref{viva} suggests an inverse relationship between $\mathbf{V}$ and $\mathbf{I}$. A deeper companion result that fills out this intuition is Hilbert's Nullstellensatz (Theorem \ref{nss}). We will review several similar results involving various fields ($\bbr,\bbc,$ and differential fields (Section \ref{diffalg})), so for convenience we list all of them in Table \ref{galois} on p. \pageref{galois}. (The table is referenced prior to the statement of each such theorem.) These theorems are instances of \emph{Galois connections} \cite{mac2013categories, ore1944galois}.
\end{remark}

\phantomsection
\label{galois}
\begin{table}[h!]
\caption{Summary of Galois connection theorems}

\centering
\begin{tabular}{| p{\textwidth} |} \hline
\smallskip
\emph{Any field $K$:}
\begin{itemize}
\item (Lemma \ref{viva}) Let $A\subseteq K[\bx]$. Then $\vs{\vi{}{\vs{A}}}= \vs{A}$.
\item (Lemma \ref{irredprime}) An algebraic variety $X$ over $K$ is irreducible if and only if its vanishing ideal $\vi{K}{X}$ is a prime ideal.
\end{itemize}
$\bbc$ \emph{(or any algebraically closed field):}
\begin{itemize}
\item (Theorem \ref{nss}, Hilbert's Nullstellensatz) Polynomial $p\in \mathbb{C}[\mathbf{x}]$ vanishes at every point in the complex zero set of  $I\unlhd \mathbb{C}[\mathbf{x}]$ if and only if $p\in\sqrt{I}$.
\item (Corollary \ref{agdict}, complex algebra-geometry dictionary) Let $A,B\subseteq \cx$. Then $\vsc{A}\subseteq \vsc{B}$ if and only if $\sqrt{(A)}\supseteq \sqrt{(B)}$.
\item (Theorem \ref{algsatrad}, Hilbert's Nichtnullstellensatz) Let $A\subseteq \cx$, and let $0\notin S\subseteq \cx$ be finite. A polynomial $p\in\cx$ vanishes at every complex solution of $(A=0,S\neq 0)$ if and only if $p\in\sqrt{\sat{(A)}}$.

\end{itemize}
$\bbr$ \emph{(or any real closed field \cite{bochnak1998real}):}
\begin{itemize}
\item (Theorem \ref{rnss}, real Nullstellensatz) Polynomial $p\in \mathbb{R}[\mathbf{x}]$ vanishes at every point in the real zero set of $I\unlhd \mathbb{R}[\mathbf{x}]$ if and only if $p\in\rrad{I}$.
\item (Corollary \ref{agdict}, real algebra-geometry dictionary) Let $A,B\subseteq \rx$. Then $\vsr{A}\subseteq \vsr{B}$ if and only if $\rrad{(A)}\supseteq \rrad{(B)}$.

\end{itemize}
\emph{Any differential field K (characteristic 0):}
\begin{itemize}
\item (Theorem \ref{dnss}, differential Nullstellensatz) Given $A\subseteq \dkx$, a differential polynomial $p\in\dkx$ vanishes at every point of $\vs[K,\delta]{A}$ if and only if $p\in\sqrt{[A]}$.
\item (Corollary \ref{dagdict}, differential algebra-geometry dictionary) Let $A,B\subseteq \dkx$. Then $\vs[K,\delta]{A}\subseteq \vs[K, \delta]{B}$  if and only if $\sqrt{[A]}\supseteq \sqrt{[B]}$.
\item (Theorem \ref{diffsatrad}, differential Nichtnullstellensatz) Let $A\subseteq \dkx$, and let $0\notin S\subseteq \dkx$ be finite. A differential polynomial $p\in\dkx$ vanishes at every solution of $(A=0,S\neq 0)$ in every differential extension field of $K$ if and only if $p\in\sqrt{\sat{[A]}}$.
\end{itemize}
\\ \hline \end{tabular}
\label{galois}
\end{table}

Complements of algebraic sets (i.e., points \emph{not} satisfying certain polynomial equations) are also important for our results in Sections \ref{regsysresults} and \ref{rga}.
\begin{definition}
Let $K$ be a field. $K$-\emph{constructible sets} are Boolean combinations (i.e., complements, finite intersections, and finite unions) of Zariski closed sets over $K$.  We write \emph{constructible sets} if $K$ is understood.
\end{definition}

A typical example of a constructible set that might not be an algebraic set is the set difference $X\setminus Y$ of two algebraic sets $X,Y$. However, finite unions and arbitrary intersections of algebraic sets are still algebraic \cite[p. 24]{shaf3ed}. Said differently, \emph{Zariski open sets} (complements of Zariski closed sets) form a \emph{topology} on $K^n$.  %Hilbert's basis theorem (Theorem \ref{hbt}) implies that

\subsubsection{Zariski Closure} 

Even if a constructible set is not algebraic, it can be augmented to one that is equivalent from the perspective of polynomial \emph{equations}.  
\begin{definition}
Let $K$ be a field. Given a set $X\subseteq K^n$, the intersection of all Zariski closed sets over $K$ that contain $X$ is the \emph{Zariski closure} of $X$ over $K$. We denote the Zariski closure by $\cl{X}{K}$, or simply $\cl{X}{}$ if $K$ is understood.
\end{definition}

Being an intersection of Zariski closed sets, the Zariski closure is itself a Zariski closed set.  As suggested above, the Zariski closure is the best \emph{algebraic} overapproximation (using coefficients from $K$) to the original set.  In particular, two sets having the same Zariski closure are indistinguishable using polynomial equations. (For a more precise statement, see Lemma \ref{vanishclos}.) 

Over $\bbr$ or $\bbc$, Zariski open sets form a topology that is coarser than the usual Euclidean topology. That is, Zariski open sets are Euclidean open, but not necessarily vice versa (for instance, the open interval $(0,1)\subseteq \bbr^1$ is Euclidean open but not Zariski open because its complement is not the set of roots of a univariate polynomial). However, the following is true:

\begin{lemma}[{\cite[Cor. 4.20]{michalek2021invitation}}]\label{euclidclos}
%10-20-21 pretty careful check

Real (respectively, complex) Zariski-closed sets are closed in the real (respectively, complex) Euclidean topology, and the real (respectively, complex) Zariski closure of a real (respectively, complex) constructible set is the real (respectively, complex) Euclidean closure of the set.
\end{lemma}

\begin{comment}
\begin{proof}
See \cite[Cor. 4.20]{michalek2021invitation}. The cited proof shows that the complex Euclidean closure of the image of a complex variety under a polynomial map equals the complex Zariski closure of the image, but the constructible set version (real or complex) reduces to that result. (Overview of the idea: Use the identity map. Since the complex Euclidean and Zariski closures are equal by \cite[Cor. 4.20]{michalek2021invitation},  it now suffices to show i) the real Zariski closure is the restriction of the complex Zariski closure to $\bbr ^n$, and ii) likewise for the real versus the complex Euclidean topologies. Claim i) is Lemma \ref{restrictclos}, which is not dependent on the current lemma. Claim ii) is immediate using the Euclidean metric on $\bbr^n$ and $\bbc^n$.) %10-20-21 324; , p. 69 michalekinvit cit
\end{proof}
\end{comment}

To distinguish between the Zariski and the Euclidean closures of a set $X$ over $\bbr$ or $\bbc$, we write $\overline{X}^{\bbr\text{-euc}}$ or $\overline{X}^{\bbc\text{-euc}}$ for the Euclidean closures.

\begin{definition}
Let $K$ be a field and $X,Y \subseteq K^n$ where $Y$ is Zariski closed over $K$. We say that $X$ is $K$-\emph{Zariski dense} (or just Zariski dense) in $Y$ if $Y=\cl{X}{K}$.
\end{definition}

\begin{comment}n
Here the difference between the Zariski and Euclidean topologies is evident: for instance, $\bbr$ is $\bbc$-Zariski dense in $\bbc$ (because the only univariate polynomial with infinitely many roots is $0$) but $\bbr$ is not dense in $\bbc$ with respect to the Euclidean topology since points of $\bbc \setminus \bbr$ are not limit points of $\bbr$. This does not violate Lemma \ref{euclidclos} because $\bbr$ is not a $\bbc$-constructible set.
\end{comment}

%This intuition applies to the sets as wholes, not to subsets. We see this with $\cl{\bbr}{\bbc}=\bbc$ (which holds because the only univariate polynomial with infinitely many roots is $0$) even though, for instance, $\vs[\bbr]{x^2+1}=\emptyset \neq \{i,-i\}=\vs[\bbc]{x^2+1}$.

The next lemma concerns \emph{irreducible} algebraic sets, i.e., algebraic sets over a field $K$ that are not the union of two smaller algebraic subsets over $K$. 
\begin{lemma}[{\cite[p. 35]{shaf3ed};}{ Table \ref{galois}}]\label{irredprime}
An algebraic variety $X$ over $K$ is irreducible if and only if its vanishing ideal $\vi{K}{X}$ is a prime ideal. %11-15-21 339
\end{lemma}
 A typical example is an algebraic set defined by a single polynomial that is irreducible in $K[\bx]$.  Some authors require an algebraic set to be irreducible to qualify for the name ``variety", but we do not impose that restriction. %However, we must be careful with arbitrary sets $A$ of polynomials defining a Zariski closed set $X$ if $(A)\neq \vi{K}{X}$. For instance, $X:=\vs[\bbr]{(x^2+1)x}= \{0\}$ is irreducible over $\bbr$ but $(A):= ((x^2+1)x)_{\bbr}$ is not prime. %Theorems \ref{} and \ref{}

\begin{comment}
\begin{lemma}\label{irredclos}
%11-1-21 careful check
Let $K$ be a field and let $X,Y\subseteq K^n$ be Zariski closed sets. If $X$ is irreducible over $K$ and $X\setminus Y$ is nonempty, then $\overline{X\setminus Y}^K= X$.
\end{lemma}

\begin{proof}
If $\overline{X\setminus Y}^K$ were not all of $X$, then $\overline{X\setminus Y}^K\cup (X\cap Y)$ would be a decomposition of $X$ that contradicts irreducibility.
\end{proof}
\end{comment}

\begin{lemma}[{Appendix}]\label{vanishclos}
%10-20-21 careful check
Let $K$ be a field. 
\begin{enumerate}
\item If $X\subseteq K^n$ and $p\in K[\mathbf{x}]$, then $p$ vanishes at every point of $X$ if and only if $p$ vanishes at every point of the Zariski closure $\overline{X}^K$ of $X$; i.e., $\vi{K}{X} = \vi{K}{\overline{X}^K}$.
\item Given $X,Y\subseteq K^n$, we have $\overline{X}^K=\overline{Y}^K$ if and only $\vi{K}{X}=\vi{K}{Y}$.
\end{enumerate}
\end{lemma}

\begin{comment}
\begin{proof}
1. In both cases $\mathbf{V}_{K}(p)$ is a Zariski-closed set containing $X$.

\noindent 2. If $\overline{X}^K=\overline{Y}^K$, then $\vi{K}{X}=\vi{K}{\overline{X}^K}=\vi{K}{\overline{Y}^K}=\vi{K}{Y}$ by 1. Conversely, suppose $\vi{K}{X}=\vi{K}{Y}$. Then similarly $\vi{K}{\overline{X}^K}=\vi{K}{\overline{Y}^K}$ by 1. Since $\overline{X}^K$,$\overline{Y}^K$ are Zariski closed, we have $\overline{X}^K=\vs[K]{\vi{K}{\overline{X}^K}}=\vs[K]{\vi{K}{\overline{Y}^K}}=\overline{Y}^K$ by Lemma \ref{viva}.
\end{proof}
\end{comment}

\begin{definition}
Let $K\subseteq L$ be fields and let $Y\subseteq L^n$. We call $Y\cap K^n$ the $K$-\emph{points} of $Y$ or the \emph{restriction} of $Y$ to $K^n$. We denote this set by $Y(K)$. 
\end{definition}
\begin{lemma}[{Appendix}]\label{restrictclos}
%10-20-21 321 careful check
Let $X\subseteq \mathbb{R}^n$. Then the real Zariski closure $\overline{X}^{\bbr}$ equals $\overline{X}^{\mathbb{C}}\cap\mathbb{R}^n$, the restriction of the complex Zariski closure to the reals.
\end{lemma}

\begin{comment}
\begin{proof}
$\subseteq$: We must show that $\overline{X}^{\mathbb{C}}\cap\mathbb{R}^n$ contains $X$ and is real Zariski closed. Containment of $X$ is clear. For real Zariski closedness, note that  $\overline{X}^{\mathbb{C}}\cap\mathbb{R}^n$ is the zero set of the real and imaginary parts of the defining polynomials of $\overline{X}^{\mathbb{C}}$. (That is, suppose $\overline{X}^{\mathbb{C}}=\vs[\bbc]{A}$ for some $A\subseteq \cx$. For all $p\in A$, distribute over the real and imaginary parts of each monomial's coefficient to write $p=p_1 +ip_2$ where $p_1,p_2\in \rx$. We only care about real solutions $\mathbf{a}\in \bbr^n$ here, so $p(\mathbf{a})=0$ if and only if $p_1(\mathbf{a})=p_2(\mathbf{a})=0$.)

\noindent $\supseteq$: We must show that $\overline{X}^{\mathbb{C}}\cap\mathbb{R}^n$ is contained in every real Zariski closed set $Y$ that contains $X$. Observe that $\vs[\bbc]{\vi{\bbr}{Y}}$ is a complex Zariski closed set containing $Y$ and thus $X$. Then $\overline{X}^{\mathbb{C}}\subseteq \vs[\bbc]{\vi{\bbr}{Y}}$ and so $\overline{X}^{\mathbb{C}}\cap\mathbb{R}^n\subseteq \vs[\bbc]{\vi{\bbr}{Y}}\cap \bbr^n =  \vs[\bbr]{\vi{\bbr}{Y}} =Y$, with the last equality following from Lemma \ref{viva}.  %The other containment is immediate from the definitions.
\end{proof}
\end{comment} 

\subsubsection{The Nullstellensatz: Connecting geometry and algebra}
The following well-known result is fundamental for going back and forth between varieties and ideals %(that is, between geometry and algebra)
over $\bbc$:

\begin{theorem}[{Hilbert's Nullstellensatz  \cite[Thm. 2, p. 179]{clo1_4};}{ Table \ref{galois}}] \label{nss}
Polynomial $p\in \mathbb{C}[\mathbf{x}]$ vanishes at every point in the complex zero set of  $I\unlhd \mathbb{C}[\mathbf{x}]$ if and only if $p\in\sqrt{I}$. Equivalently, $\vi{\bbc}{\vs[\bbc]{I}}=\sqrt{I}$ for $I \unlhd \cx$.
\end{theorem}

Real algebraic geometry requires a more refined (but computationally less tractable) notion than the radical, the \emph{real radical}: 

\begin{definition}
Let $I \unlhd \mathbb{R}[\mathbf{x}]$. The \emph{real radical} $\sqrt[\mathbb{R}]{I}$ of $I$ is the set of all polynomials $p\in \rx$ such that for some natural number $m$ and polynomials $g_1,g_2,\ldots, g_s \in \rx$, the sum $p^{2m}+g_1^2 + g_2^2+\cdots +g_s^2$ belongs to $I$. %(This is a \emph{sum-of-squares}, which we henceforth abbreviate by \emph{SOS}.) 
\end{definition}
\begin{remark}\label{sos}
An expression of the form $g_1^2 + g_2^2+\cdots +g_s^2$ is a \emph{sum-of-squares}. Sums-of-squares are important for real algebraic geometry because $g(\bx)=(g_1(\bx))^2 + (g_2(\bx))^2+\cdots +(g_s(\bx))^2$ is zero at $\mathbf{a}\in \bbr^n$ if and only if $g_i(\mathbf{a})=0$ for all $1\leq i\leq s$.
\end{remark}
The following straightforward properties follow immediately from the definitions and Remark \ref{sos}. We may use them without comment but record them here for completeness.
\begin{proposition}\label{radprop}
\phantom{ghost}
\begin{enumerate}
\item If $I\unlhd \rx$, then $\rrad{I}$ is a radical ideal that contains $\sqrt{I}$ (possibly strictly) .
\item If $I\unlhd \cx$ (respectively, $\rx$), then $\sqrt{\sqrt{I}}=\sqrt{I}$ (respectively, $\rrad{\rrad{I}}=\rrad{I}$  and $\rrad{\sqrt{I}}=\rrad{I}$).
\item If $I,J\unlhd \,\cx$ (respectively, $\rx$) and $I\subseteq J \subseteq \sqrt{I}$ (respectively, $\rrad{I}$), then

\begin{enumerate} 
\item  $\sqrt{J}=\sqrt{I}$ (respectively, $\rrad{J}=\rrad{I}$) and
\item $\vsc{I}=\vsc{J}=\vsc{\sqrt{I}}$ (respectively, $\vsr{I}=\vsr{J}=\vsr{\rrad{I}}$).
\end{enumerate}
\end{enumerate}
\end{proposition}

A much deeper result known as the \emph{real Nullstellensatz} tells us that $\sqrt[\mathbb{R}]{I}$ is the largest collection of polynomials over $\bbr$ that has the same \emph{real} zero set as $I$:

\begin{theorem}[{Real Nullstellensatz \cite[Cor. 4.1.8]{bochnak1998real};}{ Table \ref{galois}}]\label{rnss}
Polynomial $p\in \mathbb{R}[\mathbf{x}]$ vanishes at every point in the real zero set of $I\unlhd \mathbb{R}[\mathbf{x}]$ if and only if $p\in\rrad{I}$. Equivalently, $\vi{\bbr}{\vs[\bbr]{I}}=\rrad{I} $ for $I\unlhd \rx$. % p. 86

\end{theorem}

With the Nullstellensatz and real Nullstellensatz, we may give an ``algebra-geometry dictionary" connecting ideals and algebraic sets :
\begin{corollary}[{Algebra-geometry dictionary \cite[Thm.~7, p.~183]{clo1_4};}{ Table \ref{galois}}]\label{agdict}
Let $A,B\subseteq \cx$ (respectively, $\rx$). Then $\vsc{A}\subseteq \vsc{B}$ (respectively, $\vsr{A}\subseteq \vsr{B}$) if and only if $\sqrt{(A)}\supseteq \sqrt{(B)}$ (respectively, $\rrad{(A)}\supseteq \rrad{(B)}$).
\end{corollary}

The cited result only deals with the complex case, but the argument is identical given the real Nullstellensatz in place of the Nullstellensatz.

\subsubsection{Inequations and Saturation}
Just as the definition of an ideal corresponds to the behavior of equations, we need an operation on ideals that reflects the presence of \emph{inequations}. \begin{comment}n Note that $1\neq 0$ and if two polynomials $g,h$ are not 0 at some point, then $gh$ is not zero at the point either. These simple observations motivate the \emph{saturation} of an ideal by a \emph{multiplicative set}: n\end{comment}

\begin{definition}
Let $K$ be a field. Let $S$ be a subset of $K[\bx]$ that does not contain $0$. By $S^\infty$ we denote the \emph{multiplicative set generated by} $S$ (i.e., the set containing $1$ and every finite product of elements from $S$. In general, a subset of $K[\bx]$ that does not contain $0$, contains $1$, and is closed under multiplication is called a \emph{multiplicative set}.). If $I \unlhd K[\bx]$, the \emph{saturation} of $I$ by $S$ is the set of all $p\in K[\bx]$ such that for some $s\in S^\infty$ we have $sp\in I$. We write $I:S^\infty$ to denote the saturation of $I$ by $S$. \end{definition}

%\begin{remark} \label{satrmk}
We think of the elements of $I$ as equations and the elements of $S$ as inequations; i.e., for all $p\in I$ and $g\in S$ we include $p(\bx)=0,g(\bx)\neq 0$ in the system of simultaneous equations and inequations. Lemma \ref{satgeom} (2) makes this interpretation precise. 

Note that the saturation $\sat{I}$ is an ideal containing $I$ and equals the entire polynomial ring if and only if $S^\infty$ contains some element of $I$. Also, if $S$ is already a multiplicative set, then $S=S^\infty$. When we write $I:S^\infty$ we automatically assume that $I\unlhd K[\bx]$ even if we do not explicitly state this. If $S=\emptyset$, then $\sat{I}=\sat[\{1\}]{I}=I$.
%\end{remark}

Usually the saturation adds to or ``saturates" $I$ with new elements, but sometimes the ideal is unchanged: 
\begin{lemma}\label{satprime}
If $I$ is a prime ideal containing no element of $S$, then $\sat{I}=I$. 
\end{lemma}
\begin{proof}
$\subseteq$: By definition we have $sp\in I$ for some $s\in S^\infty$ if $p\in \sat{I}$, but for $I$ prime either $s$ or $p$ must then belong to $I$. It must be $p\in I$ because $s\notin I$ lest $1\in I$ or some product of elements of $S$ belong to $I$. This is impossible for $I$ prime since $I\cap S = \emptyset$ by assumption. 

\smallskip

\noindent $\supseteq$: Immediate.
\end{proof}

Our main interest is the case where $S=\{s_1, s_2, \ldots,s_m \}$ is finite. In this case it is equivalent to saturate by the single element defined by the product of the $s_j$. Let $\mathrm{\Pi} S := \prod_j s_j = s_1s_2\cdots s_m$. %note that $\mathrm{\Pi} S=0$ if and only if some $s_j=0$. We write 
We write $\sat[(\mathrm{\Pi}  S)]{I}$ to denote the saturation of $I$ by the multiplicative set $\{1, \mathrm{\Pi}  S, (\mathrm{\Pi}  S)^2, (\mathrm{\Pi} S)^3, \ldots\}$ that is generated by the single element $\mathrm{\Pi}  S$ (note that $\mathrm{\Pi}  S$ is not 0 because $0\notin S$). %Our motivation for doing this is that we are interested in 

The solution set $\vs{I}\setminus \vs{\mathrm{\Pi} S}$ (points that make all elements of $I$, but \emph{no} elements of $S$, vanish) is closely related to $\sat{I}$. The following lemma will help prove our central results in Section \ref{regsysresults} by making the connection between saturation ideals, equations, and inequations (and hence constructible sets):
 %(see Lemma \ref{satgeom}, which subsequently .
 %and not $\vs{I}\setminus \vs{S}$ (which only removes points from $\vs{I}$ that cause \emph{all} elements of $S$ to vanish).

%\begin{lemma}\label{satfin}
%If $0\notin S\subseteq K[\bx]$ is finite, then $\sat{I} =\sat[(\mathrm{\Pi}  S)]{I}$.
%\end{lemma}

\begin{lemma}[{Appendix}]\label{satgeom}
%andre reviewed and I did, too.
Let $I \,\unlhd \, \mathbb{C}[\mathbf{x}]$ and let $0\notin S\subseteq \bbc[\bx]$ be finite. Then
\begin{enumerate}
\item $\sat{I} =\sat[(\mathrm{\Pi}  S)]{I}$ and 
 \item $\mathbf{V}_{\mathbb{C}}(I:S^{\infty})=\vsc{\sat[(\mathrm{\Pi}  S)]{I}}=\overline{\mathbf{V}_{\mathbb{C}}(I)\setminus \mathbf{V}_{\mathbb{C}}(\mathrm{\Pi}  S)}^{\mathbb{C}}$.
 \end{enumerate}
\end{lemma}

Lemma \ref{satgeom} (1) holds for any field. In Lemma \ref{satgeom} (2) we work over $\bbc$ because the proof uses the Nullstellensatz. The algebraic version of Lemma \ref{satgeom} (2) is:

%Technically the cited proof saturates the ideal $I$ by another ideal instead of a multiplicative set (requiring a slightly different definition of saturation). However, our case is equivalent; instead of the multiplicative set $S$, we can substitute the ideal  generated by all finite products of nonconstant elements from $S$. %10-12-21 314 careful check

\begin{theorem}[{Hilbert's Nichtnullstellensatz;} {Table \ref{galois}}; Proof in appendix]\label{algsatrad}
Let $A\subseteq \cx$, and let $0\notin S\subseteq \cx$ be finite. A polynomial $p\in\cx$ vanishes at every complex solution of $(A=0,S\neq 0)$ if and only if $p\in\sqrt{\sat{(A)}}$.
\end{theorem}

The next result will help prove correctness (Theorem \ref{triangcorr}) of an important elimination algorithm in Section \ref{rga}.

\begin{theorem}[{Splitting, algebraic case \cite[Cor. 5]{rga}}]\label{algsplitrad}
Let $A\subseteq \cx$, and let $0\notin S\subseteq \cx$ be finite. If $h\in \cx \setminus \{0\}$, then \[\sqrt{\sat{(A)}}=\sqrt{\sat{(A,h)}}\cap \sqrt{\sat[(S\cup\{h\})]{(A)}}.\]
\end{theorem}
\begin{proof}
The claim follows from Theorem \ref{algsatrad} and the fact that every point either makes $h$ vanish or not.
\end{proof} %, p. 103
In Section \ref{regsysresults} we will need the following fact about zero sets of saturations and extending ideals of $\rx$ to $\cx$:
\begin{lemma}\label{restrictsat}
%10-16-21 318 careful check
For $I\unlhd \rx$ and multiplicative set $S\subseteq \rx$, $\mathbf{V}_{\mathbb{C}}(I_{\mathbb{C}}: S^{\infty})\cap \mathbb{R}^n =  \mathbf{V}_{\mathbb{C}}(I: S^{\infty})\cap \mathbb{R}^n$.
\end{lemma}

\begin{proof}
$\subseteq$: Clear because $I: S^{\infty}\subseteq I_{\mathbb{C}}: S^{\infty}$. 

\smallskip

\noindent $\supseteq$: Assume $\mathbf{a}\in  \mathbf{V}_{\mathbb{C}}(I: S^{\infty})\cap \mathbb{R}^n$ and $p\in  I_{\mathbb{C}}: S^{\infty}$. Note that the real and imaginary parts of $p$ both belong to $I: S^{\infty}$. Hence both vanish at $\mathbf{a}$ and $p(\mathbf{a})=0$ as needed.
\end{proof}

\subsection{Totally Real Varieties: Keeping complex things real} \label{tr} %while being complex}
%11-22-21 ~$

As suggested earlier, real radicals are usually considered computationally intractable \cite{becker1993computation, Boreale20}. A recurring theme of our work in Sections \ref{regsysresults} and \ref{rga} is that we can sometimes circumvent the complications of real algebraic geometry by working over $\bbc$. For instance, we would like to use, without loss of precision, complex radical ideals instead of unwieldy real radical ideals. The following is an important condition that, if satisfied, makes this possible.

\begin{definition}\label{trdef}
Let $X \subseteq \bbc^n$ be a complex variety that is defined over $\bbr$. (That is, $X=\vsc{A}$ for some $A\subseteq \rx$.) We say $X$ is \emph{totally real} if the real points of $X$ are $\bbc$-Zariski dense in $X$; i.e., $\overline{X(\bbr)}^{\bbc}= \cl{X\cap \bbr^n}{\bbc}=X$.
\end{definition}
The key intuition about a totally real variety $X$ is that it has ``enough real points" for the real variety $X(\bbr)$ to closely resemble $X$. More precisely, the real points are not contained in a proper complex subvariety of $X$ and in that sense are algebraically indistinguishable from the strictly complex points. Proposition \ref{realcxdim} below characterizes this phenomenon in terms of dimension. %More precisely, totally real varieties have the same real dimension as complex dimension (see Definition \ref{krull} and %and \ref{dimdef}, 
%Remarks \ref{realdimrmk},\ref{cxrealdim} %and Lemma \ref{trsmooth} 
%below). 
If $X$ is totally real, we can often transfer simpler proofs or more efficient algorithms for $X$ to the real variety $X(\bbr)$ \cite{sottile19}. % In this section and any time we discuss totally real varieties, our field is either $\bbr$ or $\bbc$.
%We do not use the following fact in the sequel, but it illustrates that being totally real allows us to replace real radicals with radicals:%whose real points determine the possible complex points:

The following fact illustrates that being totally real allows us to replace real radicals with complex radical ideals:

\begin{proposition}[{\cite[Prop. 1.3]{sander1996aspects}}]\label{trrad}
$(A)\unlhd \rx$ is real radical (i.e.,$ \rrad{(A)}=(A)$) if and only if $(A)$ is radical and $\vs[\bbc]{A}$ is totally real.
\end{proposition}
We need an alternative characterization that is more amenable to computation. To state it, we require the notions of \emph{irreducible component}, \emph{dimension of a variety}, and \emph{nonsingular point}.
\begin{definition}
Let $X\subseteq \bbc^n$ be a complex variety (not necessarily defined over $\bbr$). An \emph{irreducible component} $Y$ of $X$ is a maximal irreducible complex subvariety of $X$ (i.e., $Y$ is not strictly contained in any irreducible complex subvariety of $X$).
\end{definition}
\begin{example}
The complex variety $\vs[\bbc]{x^2+y^2}=\vs[\bbc]{(x+iy)(x-iy)}$ is not irreducible, but it has two irreducible components defined by $x+iy=0$ and $x-iy=0$ (these are lines in $\bbc^2$).
\end{example}
\begin{definition}\label{krull}
Let $X\subseteq \bbc^n$ be a nonempty complex variety, and let $\mathbf{a}\in X$. We define the (complex) \emph{dimension of $X$ at $\mathbf{a}$} (written $\text{dim}_{\mathbf{a}}(X)$) to be $d$ if i) there are distinct irreducible complex subvarieties $X_1,\ldots, X_{d}$ of $X$ such that $\{\mathbf{a}\} \subsetneq X_1\subsetneq \cdots \subsetneq X_{d}$ and ii) this is the longest such sequence of subvarieties in $X$. If there is no such $X_1$, we define  $\text{dim}_{\mathbf{a}}(X)$ to be 0. The \emph{dimension of} $X$ is the maximal dimension of $X$ at any of its points.
\end{definition}

By the correspondence between prime ideals and irreducible varieties, this definition is equivalent to the maximal length of a chain of prime ideals in the \emph{coordinate ring} of the variety (the quotient of $\cx$ by the ideal of all polynomials in $\cx$ that vanish at every complex point of the variety; this ideal is radical) \cite[p. 25]{shaf3ed}. This version is often referred to as the \emph{Krull dimension}.

\begin{example}
If $X=\{\mathbf{a}_1, \ldots, \mathbf{a}_m\}$ is finite, where $\mathbf{a}_i\in \bbc^n$, then $\text{dim}_{\mathbf{a}_i}(X)=0$. If $X=\vs[\bbc]{p}$ for some non-constant polynomial $p\in \cx$, then $\dimpt{\mathbf{a}}{X}= n-1$ at all points of $X$. The dimension of $\bbc^n$ at any point is $n.$ 

If a variety is reducible, then different components can have different dimensions. %(which is why we defined dimension at a point; dimension is a local property). 
For instance, $X=\vsc{x(y-z), y(y-z)}\subseteq \bbc^3$ is a union of a line ($x=y=0$, dimension 1) and a plane ($y-z=0$, dimension 2). At the intersection point $(0,0,0)$, $X$ has dimension 2. 
\end{example}

The same definitions are used, \emph{mutatis mutandis}, for the  \emph{real dimension} of a real variety \cite[Def. 2.8.1]{bochnak1998real}; we just replace $\bbc^n$ with $\bbr^n$ and ``complex'' with ``real'' in Definition \ref{krull}. For the prime ideal/Krull dimension version, we use the \emph{real coordinate ring} (the quotient of $\rx$ by all polynomials over $\bbr$ that vanish at the real points of the variety; this ideal is real radical and not just radical). The real dimension could be smaller than the complex if the variety is not totally real (see Example \ref{trrmk}), but not otherwise:

\begin{proposition}[{\cite[Thm. 12.6.1 (3)]{marshall2008positive}}{\cite[Thm. 2.4]{HarrisHS20}}] \label{realcxdim}
Let $X \subseteq \bbc^n$ be an irreducible complex variety defined over $\bbr$.  Then $X$ is totally real if and only if the real dimension of $X$ equals the complex dimension.
\end{proposition}

Whether or not this happens is determined by \emph{nonsingular points}.

\begin{definition} \label{dimdef}
Let $X=\vs[\bbc]{A}$ where $A=\{p_1,p_2,\ldots, p_s\}\subseteq \bbc[\bx]$ and $(A)$ is a radical ideal; recall $|\bx|=n$. The \emph{Jacobian matrix} corresponding to $A$ is the $s\times n$-matrix $J_A:= \left[\frac{\partial p_i}{\partial x_j}\right]$ whose $ij$-th entry is the formal partial derivative of polynomial $p_i$ with respect to variable $x_j$. We say $\mathbf{a}\in X$ is a \emph{nonsingular} point of the variety $X$ if the matrix $J_A(\mathbf{a})$ (that is, $J_A$ with $\mathbf{a}$ substituted for $\bx$) has rank equal to $n-\dimpt{\mathbf{a}}{X}$. (In other words, there is a maximal choice of $n-\dimpt{\mathbf{a}}{X}$ rows of $J_A(\mathbf{a})$ that form a linearly independent set of $\bbc$-vectors in $\bbc^n$.) %It can be shown that $n-\dimpt{\mathbf{a}}{X}$ is the maximum possible rank.) 
Otherwise we say $\mathbf{a}$ is a \emph{singular} point of $X$.
\end{definition}

If $X=\vs[\bbc]{A_1}=\vs[\bbc]{A_2}$ and both $(A_1)$ and $(A_2)$ are radical, then nonsingularity of $\mathbf{a}\in X$ is independent of the choice of $A_1$ or $A_2$ in the sense that $J_{A_1}(\mathbf{a})$ and  $J_{A_2}(\mathbf{a})$ have the same rank. This follows from the Nullstellensatz (which implies in our case that $(A_1)=(A_2)=\vi{\bbc}{X})$ and \cite[Thm. 5.1, p. 32]{Hartshorne}.

\begin{example}
Suppose $X=\vs[\bbc]{p}$ for some square-free (has no repeated factors) $p\in\cx$; then $(p)$ is radical. Moreover, matrix $J_{\{p\}}(\mathbf{a})$ has the form $\left[\frac{\partial p}{\partial x_1}(\mathbf{a}),\ldots, \frac{\partial p}{\partial x_n}(\mathbf{a})\right]$ and either has rank 0 (if all $\frac{\partial p}{\partial x_i}$ vanish at $\mathbf{a}$; then $\mathbf{a}$ is singular) or rank $1=n-(n-1)=n-\dimpt{\mathbf{a}}{X}$ (in which case $\mathbf{a}$ is a nonsingular point of $X$).
\end{example}

%Definition \ref{dimdef} applies to real varieties with the same changes as before: 1. $\bbc$ is replaced everywhere with $\bbr$ and 2. $(A)$ must be a \emph{real} radical ideal and not just a radical ideal \cite[Def. 3.3.4]{bochnak1998real}).

While the definition is technical, the intuitive picture is simpler: at a nonsingular point the variety looks ``smooth", while at a singular point we find  cusps/sharp corners, self-intersections, etc. For instance, points in the intersection of two irreducible components are always singular \cite[Thm. 2.9]{shaf3ed}. % [TODO: graphs?]
Importantly, smooth points give us a characterization of totally real varieties that we can calculate with: %, p. 103

\begin{proposition}[{\cite[Thm. 12.6.1 (4)]{marshall2008positive}}{\cite[p. 736]{blekherman2021sums}}] \label{trsmooth}
Let $X \subseteq \bbc^n$ be a complex variety defined over $\bbr$. Then $X$ is totally real if and only if every irreducible component of $X$ contains a nonsingular real point of $X$. 
\end{proposition}

In light of Proposition \ref{realcxdim}, the equivalence given by Proposition \ref{trsmooth} is not surprising because the real dimension of a variety (defined over $\bbr$) at a smooth real point is the same as the complex dimension \cite[p. 185]{ContiT95}.  %Along with Lemma \ref{trsmooth}, \cite[Thm. 12.6.1]{marshall2008positive} establishes that totally real varieties have the same real and complex dimensions (``enough real points"). %algebraically well behaved, so we might expect that near a nonsingular real point $\mathbf{a}$ the complex points satisfy the same polynomial equations over $\bbr$ as $\mathbf{a}$ does. 

Decomposition into irreducible components \cite{Wang92a} and finding smooth points \cite{HarrisHS20} are both well-studied (albeit computationally intensive) problems in algorithmic algebraic geometry, so we can decide if any given complex variety is totally real. 

\begin{example}\label{trrmk}
As a simple example, the complex variety defined by $x^2+y^2=0$ is not totally real because the lone real point $(0,0)$ is singular, being the intersection of two complex lines $x+iy=0,x-iy=0$ (the irreducible components of $\vsc{x^2+y^2}$). Alternatively, the partial derivatives $2x$ and $2y$ vanish simultaneously at $(0,0)$; note that $x^2+y^2$ is square-free so $(x^2+y^2)$ is radical. Nevertheless, the equations $x=0,y=0$ define the same real points and the corresponding singleton set $\{(0,0)\}$ \emph{is} a totally real complex variety (see Proposition \ref{alltr}). In terms of dimension, all this happens because $\vsr{x^2+y^2}$ has real dimension 0 but $\vsc{x^2+y^2}$ has complex dimension 1. We give further examples in Section \ref{regsysresults}.

In general, totally real varieties contain non-real points as well; the definition only requires that the complex solutions to the defining equations form the smallest algebraic set containing all the real solutions. 
\end{example}

The next result shows more generally than Example \ref{trrmk} that being totally real depends on the complex variety and not only on the set of real points.

\begin{proposition}[{Appendix}]\label{alltr}
For any $A\subseteq \rx$, there exists a finite set $B\subseteq \rx$ such that  $\mathbf{V}_{\mathbb{C}}(B)$ is totally real and $\mathbf{V}_{\bbr}(A)=\mathbf{V}_{\bbr}(B)$.
\end{proposition}
%10-19-21 pretty careful check
%10-21-21 careful check
%//another way of stating it: BrakeEt16 num alg geom paper; apparently folklore/common knowlege

%//omit if never end up using

Following up on Example \ref{trrmk}, an illustration of Proposition \ref{alltr} is $A=\{x^2+y^2\}$ and $B=\{x,y\}$.% $\vs[\bbc]{A}$ is not totally real but $\vs[\bbc]{B}$ is, even though these complex varieties have the same real points. In other words,

Proposition \ref{alltr} shows that any real variety has a representation (i.e., some choice of defining polynomial equations) such that the corresponding complex variety is totally real. Theoretically, this suggests that we can always assume our real varieties are totally real as complex varieties. Computationally, though, this is a nontrivial assumption because it may require computing generators of the real radical to transform the representation if the original equations have ``too many complex solutions''/define a complex variety that is not totally real.

However, in our experience most systems of polynomial equations over $\bbr$ that arise from applications already define a totally real variety. Indeed, the only counterexamples we are aware of have the form $\vsc{p}$ where $p\in \rx$ and $p(\mathbf{a})\geq 0$ for all $\mathbf{a}\in \bbr^n$ (for instance, sums of squares like in Example \ref{trrmk}). The property of being totally real is common and even appears to be the typical situation. Evidence of this is the following theorem, sometimes called the ``sign-changing criterion" in the literature: \begin{theorem}[{\cite[Thm. 12.7.1]{marshall2008positive}}]\label{signchange}
Let $p\in \rx$ be irreducible. Then $\vsc{p}$ is totally real if and only if there exist $\mathbf{a},\mathbf{b}\in \bbr^n$ such that $p(\mathbf{a})>0$ and $p(\mathbf{b})<0$.
\end{theorem}

In Section \ref{check} we discuss the intriguing relationship between this result and efficiently verifying algebraic invariants (Definition \ref{invardef}) of polynomial dynamical systems.

We extend the notion of totally real varieties to constructible sets. (We are not aware of this definition appearing in the literature, but it is only a minor generalization of the concept for varieties.) This will enable us to use results about saturation ideals over $\bbc$ (in particular, Lemma \ref{satgeom} (2)) when handling inequations over $\bbr$.
\begin{definition}\label{trcons}
Let $A,B\subseteq \rx$. We say that the constructible set $\mathbf{V}_{\mathbb{C}}(A)\setminus  \mathbf{V}_{\mathbb{C}}(B)$ is a \emph{totally real constructible set} if $\overline{\mathbf{V}_{\mathbb{C}}(A)\setminus  \mathbf{V}_{\mathbb{C}}(B)}^{\mathbb{C}}=\overline{\mathbf{V}_{\bbr}(A)\setminus  \mathbf{V}_{\bbr}(B)}^{\mathbb{C}}$ ; i.e., the real constructible set $\mathbf{V}_{\bbr}(A)\setminus  \mathbf{V}_{\bbr}(B)$ is Zariski dense in  $\overline{\mathbf{V}_{\mathbb{C}}(A)\setminus  \mathbf{V}_{\mathbb{C}}(B)}^{\bbc}$.
\end{definition}
The main way to obtain a totally real constructible set is to take the set difference of a totally real variety with an arbitrary variety:
\begin{lemma}[{Appendix}]\label{trtrcons}
Let $A,B \subseteq \rx$. If $\vsc{A}$ is totally real, then $\vsc{A}\setminus \vsc{B}$ is a totally real constructible set. \end{lemma}
% 7-5-22 careful mental
%10-21-21 323-4 pretty careful

Totally real constructible sets play an important role in Theorem \ref{invarcor} (Section \ref{regsysresults}), where they give a simpler characterization of one of our main results (Theorem \ref{satinvar}).

\subsection{Lie Derivatives and Algebraic Invariants of Polynomial Vector Fields} \label{lie}

A system of ODEs 

\[\mathbf{x'} =(x_1',\ldots,x_n') = (f_1(x_1,\ldots, x_n),\ldots,f_n(x_1,\ldots, x_n))= \mathbf{f}(\bx)\]

\vspace{.2cm}

\noindent defines a \emph{vector field} $\mathbf{F}:\bbr^n \rightarrow \bbr^n$ \cite{hubbard2015vector} that describes the motion of a hypothetical object moving according to the given differential equations. (In particular, $\mathbf{F}(\mathbf{a})=\mathbf{f}(\mathbf{a})=(f_1(\mathbf{a}),\ldots,f_n(\mathbf{a}))$ is the velocity vector at $\mathbf{a}\in \bbr^n$.) Abusing terminology, we also refer to the corresponding system $\mathbf{x'}=\mathbf{f}(\bx)$ as a vector field.
In this paper we assume that the $f_i$ are multivariate polynomials over $\bbr$, so $\xpef$ satisfies the Picard-Lindel\"{o}f existence and uniqueness theorem for ODEs at all points of $\bbr^n$ \cite[Thm. 10.VI]{walter98}.  We write $\xpf$ for the set $\{x_1'-f_1(\bx),\ldots, x_n'-f_n(\bx)\}$. %of differential polynomials. %we refer to subsets and elements of $\xpf$. 
%We only consider systems that have unique solutions for all initial conditions.

The rate of change of a function $p:\mathbb{R}^n\rightarrow \mathbb{R}$ along vector $\mathbf{F}(\mathbf{x})$ is given by the \emph{Lie derivative} $\mathcal{L}_{\mathbf{F}}(p)$ of $p$ with respect to $\mathbf{F}$ at $\mathbf{x}$: %i.e., the directional derivative along F

\[\mathcal{L}_{\mathbf{F}}(p)(\bx):= \nabla p \boldsymbol{\cdot} \mathbf{F}(\bx)= %\sum_{i=1}^n \left( \frac{\partial p}{\partial x_i}(\mathbf{x})\right )x_i' = 
\sum_{i=1}^n \left( \frac{\partial p}{\partial x_i}(\mathbf{x})\right )f_i(\mathbf{x}).\]

\noindent  We repeat the procedure on  $\mathcal{L}_{\mathbf{F}}(p)$ to get \emph{higher-order} Lie derivatives $\mathcal{L}_{\mathbf{F}}^{(2)}(p)$,  $\mathcal{L}_{\mathbf{F}}^{(3)}(p),\ldots$ (take $\mathcal{L}_{\mathbf{F}}^{(0)}(p)$ to be $p$ and  $\mathcal{L}_{\mathbf{F}}^{(n+1)}(p)$ to be $\mathcal{L}_{\mathbf{F}}(\mathcal{L}_{\mathbf{F}}^{(n)}(p))$).

To emphasize the particular ODEs, we sometimes write $\mathcal{L}_{\xpef}(p)$ instead of $\ldf{p}$. If the vector field $\mathbf{F}$ is understood, we simply write $\dot{p}$ for the Lie derivative $\mathcal{L}_{\mathbf{F}}(p)$. As with $\mathbf{F}$, we only consider $p\in \rx$. 

The following is straightforward:
\begin{lemma} \label{ldderiv}
Lie derivatives with respect to a vector field $\mathbf{F}$ obey the sum and product rules: $\mathcal{L}_{\mathbf{F}}{(p+q)} = \mathcal{L}_{\mathbf{F}}(p)+\mathcal{L}_{\mathbf{F}}(q)$ and $\mathcal{L}_{\mathbf{F}}{(qp)} = \mathcal{L}_{\mathbf{F}}(q)p+q\mathcal{L}_{\mathbf{F}}(p)$. 
\end{lemma}

Lie derivatives have a close relationship with \emph{invariant sets} of a vector field $\mathbf{F}$.

\begin{definition}\label{invardef}
Let  $\mathbf{F}$ be a vector field on $\bbr^n$ defined by a system of ODEs $\xpef$. A subset $X\subseteq \mathbb{R}^n$ is \emph{invariant} with respect to $\mathbf{F}$ if for every $\mathbf{x_0}\in X$ and solution $\varphi_{\bx_0}: D \rightarrow \bbr^n$ to the initial value problem $(\xpef, \bx(0)=\bx_\mathbf{0})$, we have $\varphi_{\bx_{\mathbf{0}}}(t)\in X$ for all $t\geq 0$ such that $\varphi_{\bx_0}$ is defined. ($D$ could be a proper open subset of $\bbr$ if $\varphi_{\bx_{\mathbf{0}}}$ does not exist for all time.)
\end{definition}
%the \emph{flow line} \cite{todo} passing through $\mathbf{x}$ is contained in $X$; i.e.,

Intuitively, if $X$ is invariant, then an object starting at any point of $X$ will remain in $X$ as the object follows the dynamics described by $\mathbf{F}$. We restrict our focus to invariant sets that are real algebraic varieties.  %An algebraic invariant set is precisely the common real solution set of all orders of Lie derivative $p, \dot{p},\ddot{p},\ldots$ for some polynomial $p$ \cite[Thm. 3]{GhorbalP14}. % \emph{Hilbert's basis theorem} \cite[p. 76]{clo1} thus implies that we only need Lie derivatives up to a finite order \ref{complexdiff} to define a given invariant set. 
The following well-known result (Theorem \ref{fullinvarcrit}) characterizes such invariant sets. Its statement uses Lie derivatives of all orders as well as the concept of invariant ideals. An ideal $I\unlhd \rx$ is an \emph{invariant ideal} with respect to $\mathbf{F}$ if for all $p\in I$, the Lie derivative $\mathcal{L}_{\mathbf{F}}(p)$ belongs to $I$ (i.e., $I$ is closed under Lie differentiation, written $ \ldf{I}=\{\ldf{p} \mid p\in I\} \subseteq I$).
\begin{notation}
Given a subset $A\subseteq \rx$ and polynomial vector field $\mathbf{F}$, we write $(\mathcal{L}^*_{\mathbf{F}}(A))$ (or simply $(\mathcal{L}^*(A))$ if the vector field is understood) to denote the ideal $(\mathcal{L}_{\mathbf{F}}^{(k)}(p) \mid p\in A, k\in \bbn)\, \unlhd \rx$ generated by the collection of higher Lie derivatives $\mathcal{L}_{\mathbf{F}}^{(k)}(p)$ for all $k\geq 0$ and $p\in A$. (Note that $(\mathcal{L}^*_{\mathbf{F}}(A))$ is distinct from $\ldf{A}$, which is the set of first Lie derivatives of the elements of $A$.)
\end{notation}

Observe that $(\mathcal{L}^*_{\mathbf{F}}(A))$ is an invariant ideal of $\mathbf{F}$ by construction because the given generating set is closed under Lie differentiation with respect to $\mathbf{F}$. 

%[X TODO: a little intuition and example of LD]

\begin{theorem}[{Characterization of algebraic invariants }{\cite[Lemma 5]{Boreale20}}{\cite[Lemma 2.1]{christopher2007inverse}}; Proof in appendix]\label{fullinvarcrit}
Let $\mathbf{F}$ be a polynomial vector field and let $X\subseteq \mathbb{R}^n$ be real Zariski closed. The following are equivalent:
\begin{enumerate}
\item $X$ is an algebraic invariant set of $\mathbf{F}$.
\item For all $p_1,\ldots,p_m \in \rx$ such that $X=\mathbf{V}_{\bbr}(p_1,\ldots,p_m)$, we have $\mathcal{L}_{\mathbf{F}}^{(k)}(p_i)(\mathbf{a})=0$ for all $k\geq 0$, $1\leq i \leq m$, and $\mathbf{a} \in X$.
\item There exists $I\unlhd \rx$ such that $X=\mathbf{V}_{\bbr}(I)$ and $I$ is an invariant ideal with respect to $\mathbf{F}$ (i.e., $\ldf{I}\subseteq I$).
\end{enumerate}
\end{theorem}

\begin{corollary}[{Appendix}]\label{membinvar}
Let $\mathbf{F}$ be a polynomial vector field and let $X=\vsr{p_1,\ldots, p_m}$. If $\ldf{p_i}\in (p_1,\ldots,p_m)$ for all $1\leq i\leq m$, then $X$ is an algebraic invariant set of $\mathbf{F}$.
\end{corollary}

The useful Lemma \ref{rradinvar} and Proposition \ref{maxinvar} below are probably folklore but we know of no specific references.  
%\begin{remark}
%Phrasing and notation vary in the literature. The statement in \cite[Lemma 2.1]{ChristopherEt07} is slightly different but equivalent to ours and the proof found there is the cleanest one we know.  
%\end{remark}

%\begin{lemma}[{\cite[Lemma 5]{Boreale20}}]\label{invarcrit}
%Let $\mathbf{F}$ be a polynomial vector field. A set $X\subseteq \mathbb{R}^n$ is an algebraic invariant set of $\mathbf{F}$ if and only if there exists $I\unlhd \rx$ such that $X=\mathbf{V}_{\bbr}(I)$ and $I$ is an invariant ideal with respect to $\mathbf{F}$.  %the vanishing ideal $I(X)$ is closed under the Lie derivative with respect to $\mathbf{x'}=\mathbf{f}(\mathbf{x})$. 
%\end{lemma} //subsmed in \ref{fullinvarcrit}
\begin{lemma}\label{rradinvar}
If $I\unlhd \rx$ is an invariant ideal of $\mathbf{F}$, then $\sqrt[\bbr]{I}$ is also an invariant ideal.
\end{lemma}
\begin{proof}
We use all three equivalent statements in Theorem \ref{fullinvarcrit}. Since $I$ is an invariant ideal of $\mathbf{F}$, we know $X:=\mathbf{V}_{\bbr}(I)$ is an invariant set by $(3)\Rightarrow (1)$ in Theorem \ref{fullinvarcrit}. We have  $X=\vsr{\rrad{I}}$ by Proposition \ref{radprop} (3b). By Hilbert's basis theorem (Theorem \ref{hbt}), there exist $p_1,\ldots, p_m \in \rx$ that generate $\rrad{I}$. Since $X$ is invariant and $X=\vsr{p_1,\ldots, p_m}$, $(1)\Rightarrow (2)$ implies that $\mathcal{L}_{\mathbf{F}}^{(k)}(p_i)(\mathbf{a})=0$ for all $k\geq 0$, $1\leq i \leq m$, and $\mathbf{a} \in X$. By the real Nullstellensatz (Theorem \ref{rnss}), each order of Lie derivative $\mathcal{L}_{\mathbf{F}}^{(k)}(p_i)$ belongs to $\sqrt[\bbr]{(p_1,\ldots, p_m)}=\sqrt[\bbr]{\sqrt[\bbr]{I}}$ (since the  $p_i$ generate  $\rrad{I}$) $=\sqrt[\bbr]{I}$ (Proposition \ref{radprop} (2)). Hence $\ldf{\rrad{I}}\subseteq \rrad{I}$ as needed.
%//actually only need 1st order LDs in the rrad (base case of (1)-> (2)), but OK
\end{proof}
%\begin{remark}
%It is not hard to show that 
%$\sqrt{I}$ is also invariant if $I$ is, but we do not use this fact.
%\end{remark}

Finally, we show that $\mathbf{V}_{\bbr}(\mathcal{L}^*_{\mathbf{F}}(A))$ is the ``largest" invariant contained in $\vsr{A}$: %compatible with the equations $p_i(\bx)=0$ for all $p_i\in A$:

\begin{proposition}\label{maxinvar}
 Let $A\subseteq \rx$ and let $\mathbf{F}$ be a polynomial vector field. The real variety $\mathbf{V}_{\bbr}(\mathcal{L}^*_{\mathbf{F}}(A))$ is invariant with respect to  $\mathbf{F}$, is contained in $\mathbf{V}_{\bbr}(A)$, and contains every invariant set $X$ of $\mathbf{F}$ such that $X\subseteq \mathbf{V}_{\bbr}(A)$.
\end{proposition}

\begin{proof}
The ideal $(\mathcal{L}^*_{\mathbf{F}}(A))$ is an invariant ideal, so $\mathbf{V}_{\bbr}(\mathcal{L}^*_{\mathbf{F}}(A))$ is an invariant set by Theorem \ref{fullinvarcrit}. We have $\mathbf{V}_{\bbr}(\mathcal{L}^*_{\mathbf{F}}(A))\subseteq \vsr{A}$ by Corollary \ref{agdict} because $A\subseteq \mathcal{L}^*_{\mathbf{F}}(A)$. For the last claim, by Theorem \ref{fullinvarcrit} there is an invariant ideal $I$ such that $X=\vsr{I}$. By Corollary \ref{agdict}, $\rrad{I}\supseteq \rrad{(A)}$ since $\vsr{I}\subseteq \vsr{A}$. Note that $\rrad{I}$ is invariant by Lemma \ref{rradinvar} and that $\rrad{I}\supseteq A$ since $\rrad{(A)}\supseteq A$. Because $(\mathcal{L}^*_{\mathbf{F}}(A))$ is contained in every invariant ideal containing $A$, we have $(\mathcal{L}^*_{\mathbf{F}}(A))\subseteq \rrad{I}$ and thus by Corollary \ref{agdict} $\vsr{\mathcal{L}^*_{\mathbf{F}}(A)}\supseteq \vsr{\rrad{I}}=\vsr{I}=X$.
\end{proof}

%See Example \ref{lorenzex1} for analysis of a concrete algebraic invariant of the important \emph{Lorenz system} \cite{lorenz1963deterministic}.

\subsection{Differential Polynomial Rings and Differential Ideals}\label{diffalg}
Differential algebra extends commutative algebra and algebraic geometry to settings that involve differentiation. This makes the theory a natural choice for exploring algebraic invariants of polynomial dynamical systems. In particular, differential elimination \cite{LiY19} is the algorithmic engine that powers our main results, Theorem \ref{satinvar} in Section \ref{regsysresults} and Theorem \ref{radlddecomp} in Section \ref{rga}. %part of differential algebra, provides tools for finding \emph{all} relations implied by a polynomial differential system, regardless of the form. Moreover, the procedures we propose in Section \ref{rga} have a concrete relationship to invariants of dynamical systems (Theorem \ref{radlddecomp}).

Most of the following definitions and basic lemmas are found in standard references \cite{Ritt, KolchinDAAG}.
\begin{definition}
An (ordinary) \emph{differential field} is a field $K$ equipped with a single \emph{derivation} operator $'$ that satisfies $(a+b)'= a'+b'$ and $(ab)' =a'b +ab'$ for all $a,b\in K$. 
\end{definition}

Examples are standard fields with the trivial zero derivation, the fraction field $\bbq(x)$ of $\bbq[x]$ with derivation given by the quotient rule, and the field of meromorphic functions (quotients of complex analytic functions) with the usual complex derivative \cite{marker2005model}. Many of the definitions and results in Sections \ref{diffalg},\ref{rank},and \ref{regsyssec} apply to more general differential fields, but in this paper we restrict to ordinary differential fields of characteristic 0 (repeated addition of $1$ never yields 0; equivalently, the field contains the rational numbers $\mathbb{Q}$ as a subfield). Henceforth we simply refer to differential fields, with these restrictions being understood even if not stated explicitly.  We generally leave the differential field unspecified and assume that it contains all the solutions that interest us (e.g., real analytic functions solving systems of ODEs with rational number coefficients). %Differential fields also naturally appear in the definition of differential polynomials:

%for specific purposes we only need real and complex numbers. 
\begin{definition}
Let $K$ be a differential field and let $x$ be an indeterminate. The \emph{differential polynomial ring} $K\{x\}$ with \emph{differential indeterminate} $x$ and coefficients $K$ is the polynomial ring $K[x= x^{(0)},x'=x^{(1)},\ldots, x^{(k)},\ldots]$ having infinitely many algebraic indeterminates (named in a suggestive way). The derivation $'$ on $K$ extends to all of $K\{x\}$ by defining $(x^{(k)})'= x^{(k+1)}$, %$(x^{(i)}x^{(j)})' = (x^{(i)})'(x^{(j)}) +   (x^{(i)})(x^{(j)})'$. 
$(p+q)'=p' +q'$, and $(pq)'=p'q + pq'$ for $p,q\in K\{x\}$. We call $x^{(k)}$ the $k$-\emph{th derivative} of $x$. The largest such $k$ appearing for any variable in a differential polynomial is the \emph{order} of the polynomial. If  $k\geq 1$ we call $x^{(k)}$ a \emph{proper} derivative of $x$. For uniformity we call $x=x^{(0)}$ a derivative (just not a proper one.) Differential polynomial rings having several differential indeterminates $(x_1,x_2,\ldots, x_n)=\bx$ are defined analogously and are denoted by $K\{\bx\}$. If an element $p \in K\{\bx\}$ has no proper derivatives of $\bx$ (i.e., $p\in K[\bx]$), we say $p$ is a \emph{nondifferential} polynomial.  %It is easy to see that the claimed extension of $'$ from $K$ to $K\{x\} = K[x,x',\ldots, x^{(k)},\ldots]$ is well defined. 
\end{definition}

Intuitively, differential polynomials are standard polynomials except for the presence of variables representing ``derivatives" of the (finitely many) differential indeterminates. For instance, $x(y')^3 +(x'')^2 -3$ is an order 2 differential polynomial of total degree 4 (from the term $x(y')^3$) in two differential variables $x,y$. 

%\begin{remark}\label{diffpolyrmk}
The derivatives are formal and do not necessarily represent limits of difference quotients like in calculus; we only require that they obey the sum and product rules. However, just as we can substitute elements of a field for the variables of a nondifferential polynomial, we can substitute differentiable functions into a differential polynomial and treat $'$ as the usual analytic derivative. For example, $p:=x''+x\in \bbc \{\bx\}$ is a differential polynomial and $\sin : \bbc \rightarrow \bbc$ is an element of the differential field of complex meromorphic functions (which contains the elements of $\bbc$ considered as constant functions). Substituting $\sin$ for $x$ and interpreting $'$ as the usual complex derivative, we find that $p(\sin)= \sin '' + \sin = -\sin + \sin =0$, the zero function (which is the zero element in this differential field). 

Note that $p'(\sin)$ is also 0 : $p'= x''' +x'$ and $p'(\sin)= \sin''' + (\sin)' = -\cos + \cos= 0$. More generally, if $\mathbf{a}\in K^n$ for differential field $K$ and $q\in \dkx$, then $q'(\mathbf{a})= (q(\mathbf{a}))'$. For instance, $q(\mathbf{a})=0$ implies $q'(\mathbf{a})=0$. Commutativity of differentiation and substitution is analogous to commutativity of polynomial addition and substitution; e.g., $(r+s)(\mathbf{a})=r(\mathbf{a})+s(\mathbf{a})$.
%i.e., $(\mathbf{x}(t))'$ is the standard derivative of a complex-valued function of a real variable.)
%\end{remark}

\subsubsection{Differential Ideals}

Just as ideals give us an algebraic way to approach polynomial equations, \emph{differential ideals} correspond to systems of polynomial differential equations.

\begin{definition}\label{diffid}
Let $K$ be a differential field. A \emph{differential ideal} is an ideal %(Definition \ref{iddef}) 
$I \unlhd \dkx$  that is also closed under differentiation: $p\in I$ implies $p'\in I$. A \emph{radical differential ideal} is a differential ideal that is also a radical ideal in $\dkx$. 

Let $A\subseteq K\{\bx\}$. We write $[A]_K$ or just $[A]$ to denote the \emph{differential ideal generated by} $A$ in $K\{\bx\}$: the collection of all finite sums $\sum_{i,j} g_{i,j}p_i^{(j)}$ where $p_i \in A$, $p_i^{(j)}$ is the $j$-th derivative of $p_i$, and $g_{i,j}\in K\{\bx\}$ is any differential polynomial. 

We write $(A)_K$ or $(A)$ to denote the ideal generated by $A$ in $K\{\bx\}$ viewed as a polynomial ring. (That is, $(A)$ consists of sums $\sum_{i} g_{i}p_i$ where $p_i \in A$ and $g_{i}\in K\{\bx\}$ is any differential polynomial.) If there is risk of confusion, we specify whether $(A)$ is the ideal generated by $A$ in $\dkx$ or in $K[\bx]$, but context usually makes it clear (for instance, this could only be an issue if $A\subseteq K[\bx]$).

 %the (nondifferential) polynomial ring $K[x_1,\ldots, x_1^{(k_1)}, x_2, \ldots, x_2^{(k_2)},\ldots, x_n, \ldots ,x_n^{(k_n)}]$ where $x_i^{(k_i)}$ is the highest derivative of $x_i$ that appears in any member of $A$. 

\end{definition}

The ideal $(A)$ is defined similarly to $[A]$, except that in forming $(A)$ we do not allow differentiation of the elements of $A$. Hence $(A)\subseteq [A]$ but generally $(A)\subsetneq [A]$. Note that $[A]$ is the minimal differential ideal containing $A$.  %we ignore the derivation operator $'$ and treat $x', y''$, etc., as new variables with no relationship to $x,y$, or each other. Unlike $[A]$, the only derivatives that appear in $(A)$ are those already present in $A$. %and do not differentiate further. //11-23-21 This is J^V_G in Ritt02 pp32-3; since we want the ideal generated by this in the diff poly for Rosenfeld's lemma, let's not bother with the purely algebraic version.

We remark that invariant ideals with respect to a polynomial vector field $\mathbf{F}$ are essentially differential ideals in the sense of Definition \ref{diffid} (the only difference being that $\rx$ is a polynomial ring and not a  differential polynomial ring). By Lemma \ref{ldderiv} (sum and product rule for the Lie derivative), the  operator $\mathcal{L}_{\mathbf{F}}: \rx \rightarrow \rx$ is a derivation on the polynomial ring $\rx$. Since $\ldf{I}\subseteq I$ for an invariant ideal $I$, such an ideal is a differential ideal with respect to the derivation $\mathcal{L}_{\mathbf{F}}$. See Lemma \ref{subslie} for the converse relating a given differential ideal to an invariant ideal.

Radical differential ideals are theoretically and practically more tractable than general differential ideals. We need the following straightforward properties: %The following simple property of radical ideals holds more generally in arbitrary \emph{commutative rings}, but we only need it for polynomial and differential polynomial rings. 
 \begin{lemma}\label{radrestr}
Let $K$ be a differential field and let $I \unlhd K\{\bx\}$ be a differential ideal.
\begin{enumerate}
\item The radical $\sqrt{I}$ of differential ideal $I$ in $K\{\bx\}$ is also a differential ideal.
\item  $\sqrt{I \cap K[\bx]}=\sqrt{I}\cap K[\bx]$ as ideals in $K[\bx]$.
\end{enumerate}
 \end{lemma}
 
\begin{proof}
\begin{enumerate}
\item \cite[Lemma 1.15]{marker2005model}.
\item Immediate from the definition of the radical of an ideal. Note that $\sqrt{I \cap K[\bx]}$ is the radical of $I \cap K[\bx]$ taken in $K[\bx]$ while the $\sqrt{I}$ is the radical of $I$ taken in $K\{\bx\}$.
\end{enumerate}

\end{proof}

\begin{remark}
By Lemma \ref{radrestr} (1), the radical ideal $\sqrt{[A]}$ is a differential ideal and hence the smallest radical differential ideal containing $A$. This radical differential ideal is often denoted by $\{A\}$ in the differential algebra literature. %  called a \emph{perfect differential ideal} and authors . 
Because we work so often with finite sets containing differential polynomials, we do \emph{not} follow this usage and by $\{p_1,\ldots,p_r\}$ we mean the set of differential polynomials with elements $p_i$, not the radical differential ideal $\sqrt{[p_1,\ldots,p_r]}$.
\end{remark}
\subsubsection{Differential Nullstellensatz and Differential Saturation Ideals}
On the ``geometric'' side, we are interested in zero sets of collections of differential polynomials. In the differential setting, there is no specific field playing the role of $\bbc$ that contains solutions to all polynomial differential equations with coefficients in the field. (\emph{Differentially closed fields} \cite{robinson1959concept}, named in analogy to algebraically closed fields like $\bbc$, do exist but we do not explicitly need them in this paper. Moreover, there are no known natural examples.) Instead we use the following terminology and notation.
\begin{definition}\label{diffvar}
Let $K$ be a differential field with $A,B\subseteq \dkx$. Let  $L$ be any differential field extending $K$ and let $\mathbf{a}\in L^n$ satisfy $p(\mathbf{a})=0$ for all $p\in A$. Then we say that $\mathbf{a}$ is a \emph{point} (or \emph{element}) of the \emph{differential zero set}  $\vs[K,\delta]{A}$. (We write $\vs[\delta]{A}$ if $K$ is understood.) If every point of $\vs[K,\delta]{A}$ is also a point of $\vs[K,\delta]{B}$, we write $\vs[K,\delta]{A}\subseteq \vs[K, \delta]{B}$ (i.e., for all differential field extensions $L$ of $K$, all solutions of $A=0$ in $L^n$ are solutions of $B=0$.) If also $\vs[K,\delta]{B}\subseteq \vs[K, \delta]{A}$, we write $\vs[K,\delta]{A}=\vs[K,\delta]{B}$ and say that $A$ and $B$ have the same differential zero sets (or the same \emph{differential solutions}). (See Remark \ref{dnssrmk} about our use of the term ``set" in this context.)

We also refer to differential zero sets as  \emph{differential varieties}, \emph{differential algebraic sets}, or \emph{Kolchin closed sets} (over $K$, when we specify a differential field containing the coefficients of the defining differential polynomials).

Analogous definitions hold for \emph{differential constructible sets} (i.e., Boolean combinations of Kolchin closed sets). See, for instance, Example \ref{pdivex}. 
\end{definition}
  
\begin{remark}\label{dnssrmk}
We could have defined algebraic varieties in an analogous way, having points in any extension field containing the coefficients of the defining polynomials. In that case Hilbert's Nullstellensatz (Theorem \ref{nss}) would read as follows: Let $K$ be a field and let $A \subseteq K[\bx]$. Then a polynomial $f \in K[\bx]$ belongs to $\sqrt{(A)}$ if and only if for all fields $L$ extending $K$ and all $\mathbf{a}\in L^n$ such that $q(\mathbf{a})=0$ for all $q\in A$, we have $f(\mathbf{a})=0$. However, $\bbc$ is algebraically closed and thus has the property that a system of polynomial equations over $\bbc$ has a solution in some extension field if and only if it has a solution in $\bbc$ \cite{marker2005fields}. Hence there is no need to mention extensions in the algebraic case. We merely do so in the differential case because there is no natural analogue of $\bbc$ available.

Strictly speaking, the collection of all points of $\vs[K,\delta]{A}$ is not a set, but rather a proper class \cite{enderton:set_thy} containing the solutions of $A$ from all differential fields extending $K$. However, in this paper we never need to treat it as a single completed object and so do not risk set-theoretic difficulties. We simply use $\vs[K,\delta]{A}$ as shorthand for universal quantification over differential field extensions. The terms ``differential zero set" and ``differential algebraic set" are harmless abuses of terminology that we use to mirror the corresponding algebraic concepts. We introduce the notion so that we can conveniently state the \emph{differential Nullstellensatz} without a technical detour into differentially closed fields. Like its algebraic counterpart, the differential Nullstellensatz connects algebra and geometry, giving us a correspondence between radical differential ideals and differential varieties. In this paper, the differential Nullstellensatz is mainly used in the form of Theorems \ref{diffsatrad} and \ref{splitrad}. In turn, these will help prove correctness (Theorem \ref{modrgacorr}) of our main algorithm $\rgaexp$ in Section \ref{rga}.

%Moreover, we note that the $K$ in $\vs[K,\delta]{A}$ is superfluous. (The reader may safely ignore this issue if desired.)
%Given any two differential fields $K_1,K_2$ both containing the coefficients of the elements of $A$, we have $\vs[K_1,\delta]{A}=\vs[K_2,\delta]{A}$ since every point of $\vs[K_1,\delta]{A}$ is a point of $\vs[K_2,\delta]{A}$ and vice versa. This is because differential fields (and fields) have the \emph{joint embedding property} \cite{todo}: given differential fields $K_1,K_2$, there is a third differential field $K_3$ containing (isomorphic copies of) $K_1$ and $K_2$.  %%//actually, I think its the amalgamation property we need

\end{remark}

\begin{theorem}[Differential Nullstellensatz {\cite[Thm. 2]{rga};}{ Table \ref{galois}}] \label{dnss}
Let $K$ be a differential field of characteristic 0. Given $A\subseteq \dkx$, a differential polynomial $p\in\dkx$ vanishes at every point of $\vs[K,\delta]{A}$ if and only if $p\in\sqrt{[A]}$. (Here $\sqrt{[A]}$ is the radical in $\dkx$ of the differential ideal generated by $A$.)
\end{theorem} %p.84

Phrased differently, $\sqrt{[A]}$ consists precisely of those differential polynomials over $K$ that vanish at all differential solutions of $A$ in any differential extension field of $K$.
%// vanishing at every point of vk delta corresponds to universal quantification over every extension field; more explicitly than just any solutions anywhere, why not dcf; signposts for subsections so people can skip//

The differential Nullstellensatz immediately implies a differential analogue of the algebra-geometry dictionary from Corollary \ref{agdict}:
\begin{corollary}[{Differential algebra-geometry dictionary;}{ Table \ref{galois}}]\label{dagdict}
Let $A,B\subseteq \dkx$, where $K$ is a differential field of characteristic 0. Then $\vs[K,\delta]{A}\subseteq \vs[K, \delta]{B}$  if and only if $\sqrt{[A]}\supseteq \sqrt{[B]}$.\end{corollary}

Given an ideal $I\unlhd K\{\bx\}$ and set $S\subseteq K\{\bx\}\setminus \{0\}$, the saturation ideal $I:S^\infty$ is defined as before by $I:S^\infty:= \{p\in K\{\bx\} \mid sp\in I \text{ for some } s\in S^\infty\}$. We have $I\subseteq \sat{I}\subseteq \sat{[I]}$. If $I$ is a differential ideal, then $[I]:S^\infty=\sat{I}$ is also a differential ideal (called a \emph{differential saturation ideal}) \cite[Lemma 1.3]{marker2005model}. 

Much like algebraic saturation ideals, differential saturation ideals represent systems of polynomial differential equations and inequations. (Example \ref{pdivex} and Lemma \ref{difftriang} show that inequations, and not just equations, are important for differential elimination.) The differential Nichtnullstellensatz makes the following important connection between solutions of differential polynomial (in)equations and radicals of differential saturation ideals. (See Theorem \ref{algsatrad} for the algebraic version.)

\begin{theorem}[{Differential Nichtnullstellensatz \cite[Cor. 3]{rga};}{ Table \ref{galois}}]\label{diffsatrad}
Let $K$ be a differential field of characteristic 0, let $A\subseteq \dkx$, and let $0\notin S\subseteq \dkx$ be finite. A differential polynomial $p\in\dkx$ vanishes at every solution of $(A=0,S\neq 0)$ in every differential extension field of $K$ if and only if $p\in\sqrt{\sat{[A]}}$.
\end{theorem} %, p. 85

%A special case of this (just let $S=\{1\}$; then $\sat{[A]}=[A]$) is a

%As promised in Remark \ref{galoisrmk}, we gather the various Galois connection results into Table \ref{galois} on p. \pageref{galois}.

The last result in this subsection is a differential analogue of Theorem \ref{algsplitrad}:

\begin{theorem}[{Splitting, differential case \cite[Cor. 5]{rga}}]\label{splitrad}
Let $K$ be a differential field of characteristic 0, let $A\subseteq \dkx$, and let $0 \notin S\subseteq \dkx$ be finite. If $h\in \dkx \setminus \{0\}$, then \[\sqrt{\sat{[A]}}=\sqrt{\sat{[A,h]}}\cap \sqrt{\sat[(S\cup\{h\})]{[A]}}.\]
\end{theorem} % p. 103

\subsubsection{Explicit Form and Lie Derivatives in Differential Ideals}
Since we intend to compute algebraic invariants of polynomial vector fields, we naturally focus on differential polynomials of the following form:
\begin{definition}\label{explicit}
Let $K$ be a differential field. We say a differential polynomial $p\in K\{\bx\}$ of order 1 is in \emph{explicit form} (or is \emph{explicit}) if $p=x_i' + q$ for some $1\leq i \leq n $ and $q\in K[\bx]$. (In particular, no proper derivative appears in $q$.)  We also say the corresponding differential equation $x_i'=-q$ is in explicit form. (Note that explicit form is desirable because it essentially replaces a derivative with a polynomial.) %We say \emph{explicit of order} $k$ to specify the order.
%We say a differential polynomial $p$ is in \emph{explicit form} (or is \emph{explicit}) with respect to a given ranking if $p$ has the form $x^{(k)} + q$ and $x^{(k)}$ is the leader of $p$. (In particular, every derivative appearing in $q$ ranks strictly lower than $x^{(k)}$.)  We also say the corresponding differential equation $x^{(k)}=-q$ is in explicit form.
%//didn't fit w/ Peano's existence thm, so made it more restrictive, but still in agreement with std ODEs for hybrid sys
\end{definition}
In the following lemma and thereafter, if $I\unlhd \bbr \{\bx\}$ we write $\mathbf{x'}-\mathbf{f}(\bx) \in I$ as shorthand indicating that each element $x_i' -f_i(\bx)$ of $\mathbf{x'}-\mathbf{f}(\bx)$ belongs to $I$. This property connects differential ideals to invariant ideals and allows us to use differential algebra to find sets invariant with respect to $\xpef$:
\begin{lemma}\label{subslie}
If $\mathbf{x'}-\mathbf{f}(\bx) \in  I \unlhd \bbr\{\bx\}$ and $I$ is a differential ideal, then for any $p\in  I \cap \rx$ we have $\mathcal{L}_{\xpef}(p)=\dot{p} \in I \cap \rx$. (In other words, $I\cap \rx$ is an invariant ideal with respect to $\xpef$.)
\end{lemma}
%7-6-22 pretty careful
\begin{proof}
We have $p'\in  I$ because $I$ is a differential ideal. Since $p\in \rx$, all monomials $m_i$ of $p'$ have the form $(k_i)(a)x_1^{k_1}\cdots x_i^{k_i-1}x_i'\cdots x_n^{k_n}$ for some $a \in \bbr$ and $k_1,\ldots k_n \in \mathbb{N}$. (This is one of the summands produced by the product rule applied to a monomial $m=ax_1^{k_1}\cdots x_i^{k_i}\cdots x_n^{k_n}$ in $p$.) Subtracting $(k_i)(a)x_1^{k_1}\cdots x_{i}^{k_i -1}\cdots x_n^{k_n}(x_i' -f_i(\bx))$ from $p'$ produces an element of $I$ (since $x_i' -f_i(\bx)\in I$ by assumption) that replaces the monomial $m_i$ with $(k_i)(a)x_1^{k_1}\cdots x_i^{k_i-1}(f_i(\bx))\cdots x_n^{k_n}$. Doing this for each $x_i$ in monomial $m$ and summing the output replaces $m'$ in $p'$ with the Lie derivative $\dot{m}$ of $m$. It follows from Lemma \ref{ldderiv} that substituting $\mathbf{f}(\mathbf{x})$ for $\mathbf{x'}$ in $p'$ this way yields $\dot{p}\in I \cap \rx$. (Recall that $\dot{p}\in \rx$ by definition.)
\end{proof}

\subsection{Rankings and Reduction}\label{rank}
As discussed in the introduction, the central aim of this paper is to use \emph{differential elimination} to algorithmically generate algebraic invariants of polynomial dynamical systems.
Elimination identifies the core content of a system of polynomial (or differential polynomial) equations by ``reducing" some polynomials with respect to others. In addition to Gaussian elimination, another classic example is long division of one univariate polynomial by another. If the divisor does not evenly divide the dividend, we are left with a nonzero remainder. One way to extend this to multivariate differential polynomials is to use a \emph{differential ranking} (or ranking, if the context is clear) \cite{sit2002ritt}. Differential rankings identify a ``leading derivative" in any differential polynomial (e.g., to determine if $x'$ or $y''$ is eliminated first)  and ensure termination of algorithms by eliminating ``large terms" first. Rankings are analogous to \emph{monomial orderings} in the theory of \Grob bases \cite[Sect. 2.2, Def. 1]{clo1_4}. However, rankings only apply to individual variables and their derivatives, whereas monomial orderings concern monomials. %Formally, a monomial ordering is a well-ordering $<$ on power products (monomials with coefficient 1) such that $1$ is the least element and if $m_1<m_2$, then $cm_1<cm_2$ for any power product $c$ \cite{todo}.  
Many of the following results do not depend on the choice of ranking (the main exceptions are algorithm $\rgaexp$ on p. \pageref{rgaalgo} and the proof of Theorem \ref{modrgacorr}) and we do not specify a ranking except when necessary. % except for calculations.

\begin{definition}
A \emph{differential ranking} is a well-founded linear ordering $<$ of derivatives (i.e., every nonempty subset of derivatives has a least element with respect to $<$) that goes up with differentiation and respects $'$: $x^{(k)}<x^{(k+1)}$ and $x<y$ implies $x'<y'$.
\end{definition}

There are different rankings for different purposes. Two prominent classes are \emph{elimination rankings} (sort variables lexicographically irrespective of their order%give precedence to any instance of a higher-ranking variable
; e.g., if $x<y$, then $x^{(k)}< y$ for any $k$) and \emph{orderly rankings} (the derivative of highest order ranks highest, regardless of the base variable; ties in order are decided according to the ordering on the variables). For instance, if $x<y$ for an orderly ranking, then $y^{(l)}<x^{(k)}$ if $l<k$ and $x^{(k)}<y^{(l)}$ if $k \leq l$.

We identify several important components %(called the \emph{leader, initial,} and \emph{separant}
of a differential polynomial with respect to a ranking. These constituents are useful for specifying the control flow of algorithms involving systems of differential polynomials. (For instance, in Section \ref{rga} we split cases by either setting a given separant equal to 0 or not.) % In particular, we want to consider sets of differential polynomials in Section [TODO] we split cases depending on whether we set 

\begin{definition}
Fix a differential ranking. 
\begin{itemize}
    \item The highest-ranking derivative that appears in a non-constant differential polynomial $p\in K\{\bx\}\setminus K$ is the \emph{leader} of $p$ (denoted $l_p$). Note that $l_p$ is a derivative of a single differential indeterminate and does not involve powers or multiple variables.
    \item The \emph{initial} of $p$ (denoted $i_p$) %, a polynomial in variables other than the leader, 
is the coefficient (viewing $p$ as a univariate polynomial in $l_p$ with coefficients in $K\{\bx\}$ that do not involve $l_p$) of the highest power of $l_p$.
\item The \emph{separant} of $p$ (denoted $s_p$) is the initial of any derivative $p^{(k)}$ where $k\geq 1$ (equivalently, the formal partial derivative $\frac{\partial p}{\partial l_p}$ of $p$ with respect to $l_p$).
\end{itemize} Unlike the initial $i_p$, the separant $s_p$ could contain $l_p$, but with a lower degree than $l_p$ has in $p$. We do not define leaders, initials, or separants for constants (i.e., elements of the field $K$).%and the \emph{tail} of $p$ is $p-[(\text{initial})(\text{highest power of the leader})]$. 
\end{definition}

%Note that the separant is also the initial of \emph{any} proper derivative of $p$.
%The initial $i_p$ is a differential polynomial not containing the leader $l_p$; 

\begin{example}\label{diffpexp}
Fix any differential ranking such that $x<y$ and let $p:=x(x+1)(y'')^2+x'y'' + x^4$. Then 
\smallskip

\begin{align*}
p'&= \underline{2x(x+1)(y'')}(y''') + (x')(x+1)(y'')^2 +x(x')(y'')^2 + \underline{x'}y''' + x''y''+4x^3x'\\
&=  (2x(x+1)y''+x')(y''') + (2x+1)(x')(y'')^2  + x''y''+4x^3x'.
\end{align*}
\smallskip 

 \noindent The leader $l_p$ of $p$ is $y''$ (and not $(y'')^2$, $x(x+1)(y'')^2$, etc.), the initial $i_p$ is the coefficient $x(x+1)$ of $(y'')^2$, and the separant $s_p$ is $2x(x+1)y'' +x'$ (underlined for visibility in $p'$ above).

\end{example}
\begin{remark}\label{rankset}
Rankings induce a well-partial-ordering on differential polynomials  ($q<p$ if the leader of $p$ is greater than the leader of $q$ or if they are the same and the degree of the leader is greater in $p$ than in $q$) and, in turn, on \emph{sets} of polynomials. See \cite[Sect. I.10]{KolchinDAAG} and \cite[Ch. 5]{mishra} for exact definitions; also see our remarks about the RGA algorithm at the beginning of Section \ref{rga}.

Rankings for nondifferential polynomial rings $K[\bx]$ are simply linear orderings of the variables. The leader, initial, and separant of a polynomial $p\in K[\bx]$ are defined as in the differential case. %(only now $p$ has no proper derivatives of the base variables $\bx$). % Equivalently, the separant (of an algebraic or differential polynomial) is the formal partial derivative $\frac{\partial p}{\partial l_p}$ of $p$ with respect to the leader $l_p$.
\end{remark}
\subsubsection{Differential Pseudodivision}
Rankings give us a systematic way of generalizing long division via \emph{differential pseudodivision} (or just pseudodivision; we also say \emph{differential pseudoreduction} or \emph{Ritt reduction}). Here we view multivariate differential polynomials as univariate differential polynomials whose coefficients are differential polynomials in one fewer variable. Differential pseudodivision is like univariate long division except the coefficients are differential polynomials and we usually cannot divide without introducing fractions. We give a concrete example before describing the process more formally.

\begin{example}\label{pdivex} Choose any differential ranking with $x<y$ and let $p=(x+1)(y'')+ x^4$ and $q=(x^2-1)(y')^2$. Observe that $y'$ is the leader $l_q$ of $q$ and $y''$ is the highest derivative of $y'$ in $p$. Pseudodividing/pseudoreducing $p$ by $q$ consists of first differentiating $q$ to match the $y''$ and then ``dividing" $p$ by $q'$ (in quotation marks because we must \emph{premultiply} $p$ by something in order to divide by $q'$ without fractions). This eliminates $y''$. If the resulting \emph{pseudoremainder} contains $(y')^2$, we ``divide" the pseudoremainder by $q$ to obtain another pseudoremainder. The final \emph{reduced} pseudoremainder does not contain any proper derivative of $l_q$ and has degree less than 2 in $l_q$.

	\begin{itemize}
	\item  Differentiate $q$ to find $q'=2(x^2-1)(y')y'' +(2xx')(y')^2$. (Note that $y''$ has degree 1 in $q'$.)
	\item Then premultiply $p$ by $(x-1)y'$ because $x+1$ is the coefficient of $l_{q'}=y''$ in $p$ and $2(x+1)(x-1)y'=2(x^2-1)y'= i_{q'}=s_q$. 
	\item Now subtract $(\frac{1}{2})q'$ from $(x-1)y'p$ to obtain pseudoremainder $r=(x-1)(y')(x^4)- xx'(y')^2.$ 
	
	Notice that $p,q'$ vanish precisely when either $p,q',s_q$ all vanish or $q',r$ vanish and $s_q$ does not. %(Recall from Remark \ref{diffpolyrmk} that $q(\mathbf{a})=0$ implies $q'(\mathbf{a})=0$ for $q\in \dkx$ and any $\mathbf{a}$ with entries from a differential extension field of $K$.) 
Using the notation of Definition \ref{diffvar} we have $\vs[K,\delta]{p,q'}=(\vs[K,\delta]{p,q',s_q})\cup (\vs[K,\delta]{q',r}\setminus \vs[K,\delta]{s_q})$. In other words, using splitting and differential pseudodivision to eliminate differential polynomials splits differential varieties into a union of differential constructible sets. %This is a geometric analogue of Theorem \ref{splitrad}. //no, that was just splitting, not pdiv
The same idea appears in %the proof of Lemma \ref{difftriang} 
Equations \ref{diffspliteq},\ref{diffcons} on p. \pageref{diffcons}.
		\item With one more round of premultiplication and subtraction, we can eliminate  $xx'(y')^2$ in $r$. Premultiply $r$ by $x^2-1$ (this is the initial $i_q$ of $q$ and $xx'$ doesn't already contain any factors of $i_q$) and subtract $-xx'q$ to obtain the final answer $(x^2-1)(x-1)(y')(x^4)$. This pseudoremainder has no variables that can be eliminated using $q$.
		\end{itemize}

 %This is one step of pseudodivision, and the result has reduced the degree of $y''$ from 2 to 1. Pseudodividing $r$ by $q$ yields a pseudoremainder having no $y''$. 
\end{example}

In general, let $p,q$ be differential polynomials ($q$ non-constant so that leaders, etc., are well defined; in particular, $q\neq 0$). Pseudodividing $p$ by $q$  involves the following steps (see \cite[Sect. I.9]{KolchinDAAG} and \cite[Sect. 2.2]{boulier2006differential}  for other versions): % that generalize the concrete calculations in Example \ref{pdivex}:
\medskip

\noindent ${\tt DiffPseudoDiv}(p,q)$:

\begin{enumerate}

\item If a proper derivative of the leader $l_q$ of $q$ appears in $p$:
	\begin{itemize}
		\item Let $(l_q)^{(k)}$, $k\geq 1$, be the highest proper derivative of $l_q$ appearing in $p$. Compute $q^{(k)}$. 
		\item Let $r_k:= $ ${\tt DiffPseudoDiv}$$(p,q^{(k)})$. (Note that the initial of $q^{(k)}$ is the separant $s_q$ of $q$ since $k\geq 1$.)  Then return
	${\tt DiffPseudoDiv}$$(r_k,q)$. % compute apply step 2 to $p$ and $q^{(k)}$; i.e., perform step 2 with $q^{(k)}$ in place of $q$ 
	\end{itemize}
%If there is no such $k$, continue to step 2. 
(Note that the call ${\tt DiffPseudoDiv}$$(p,q^{(k)})$ must proceed to step 2 because $(l_q)^{(k)}$ is the leader of $q^{(k)}$. Also note that the degree of $(l_q)^{(k)}$ in $q^{(k)}$ is 1.)% with $p$ and $q$.
\item If $l_q$ (but no proper derivative thereof) appears in $p$ and the highest power $d$ of $l_q$ that appears in $p$ is at least the degree $e$ of $l_q$ in $q$:
	\begin{itemize}
	\item Let  $c\in \dkx$ be the coefficient of $l_q^d$ in $p$. Multiply $p$ by an appropriate factor $\alpha$ of the initial $i_q$ of $q$ to ensure that  $(\alpha )(c)$ is divisible by $i_q$. This is called 
\emph{premultiplication}. It suffices to premultiply $p$ by $\alpha:= i_q/g$, where $g$ is the GCD of $i_q$ and $c$. %[TODO: make sure GCD is OK; yes: formally they are just multivar polys]

	\item Now subtract the necessary multiple of $q$ (namely, $\left(c/{g}\right)(l_q)^{d-e}q$) to eliminate $(l_q)^d$ from $p$ and obtain a \emph{pseudoremainder} $r$ whose highest power of $l_q$ is less than $d$:	
 \[r:= (i_q/g)p -(c/{g})(l_q)^{d-e}q .\]
\phantomsection
\label{pdivdeg}
\noindent Note that $i_q/g$ and $c/g$ are both differential polynomials and not fractions because $g$ divides both $i_q$ and $c$ by definition.  Also note that the total degree $deg(r)$ of $r$ is at most the sum of the total degrees $deg(p),deg(q)$ of $p$ and $q$. (This is because degrees of products are additive and the degrees of $i_q/g$ and $(c/{g})(l_q)^{d-e}$ are at most $deg(q)-e$ and $deg(p)-d+(d-e)=deg(p)-e$, respectively. Hence $deg(r)\leq deg(p)+deg(q)-e$, but for simplicity we use the larger bound $deg(p)+deg(q)$.)

	\item Return ${\tt DiffPseudoDiv}$$(r,q)$.  %Repeat step 2 with $r$ in place of $p$ if $r$ has $l_q$ with degree at least $e$.% satisfies step 2's condition on $l_q$. 
	%If not, but $r$ has a proper derivative, perform step 1 with $r$ in place of $p$ but with the original $q$.
	
\end{itemize} 

If  $l_q$ (but no proper derivative thereof) appears in $p$ and $d<e$, return $p$.

\item If neither $l_q$ nor any proper derivative of $l_q$ appears in $p$, return $p$.

%\item Repeat as long as the current pseudoremainder $r$ satisfies either condition: either it contains a proper derivative of $l_q$ or a power of $l_q$ at least as large as the the degree of $l_q$ in $q$.
\end{enumerate}

\begin{remark}\label{premred}
The algorithm ${\tt DiffPseudoDiv}$ terminates because step 1 reduces the order of the highest proper derivative of $l_q$ that appears in $p$ and step 2 reduces the degree of $l_q$ in $p$ to a value below $e$. (In particular, $(l_q)^{(k)}$ is eliminated by the call ${\tt DiffPseudoDiv}$$(p,q^{(k)})$ because the degree of $(l_q)^{(k)}$ in $q^{(k)}$ is $e=1$.)% when there are no proper derivatives of $l_q$ in $p$.

The same concepts apply in the nondifferential case (i.e., rankings on variables in $K[\bx]$). Algebraic pseudodivision is the same process as ${\tt DiffPseudoDiv}$, except that step 1 never applies because there are no derivatives.

If ${\tt DiffPseudoDiv}$$(p,q)=r$, we say $r$ is the (differential) pseudoremainder resulting from (differential) pseudodivision of $p$ by $q$. The core component of pseudodivision is a single round of premultiplication and subtraction (the first and second items in step 2 of ${\tt DiffPseudoDiv}$). We refer to this as a ``single step" of pseudodivision and call the result a pseudoremainder even though it is an intermediate element that may be used for further steps of pseudodivision until the final pseudoremainder is reached. See Remark \ref{singlestepconv} immediately preceding the description of algorithm ${\tt Triangulate}$.

\end{remark}

The next proposition makes explicit the differential-algebraic relationship between $p,q$, and $r$ when ${\tt DiffPseudoDiv}$$(p,q)=r$. (The analogue for division of integer $a$ by nonzero integer $b$ is the equation $a=cb + d$, where $c$ is the quotient and $d$ is the remainder.)
\begin{proposition}[{Appendix}]\label{pdivsat}
Let ${\tt DiffPseudoDiv}$$(p,q)=r$. Then for some $\widetilde{s}$ a product of factors of $s_q$, $\widetilde{i}$ a product of factors of $i_q$, and  $\widetilde{q}\in [q]$ we have $(\widetilde{s})(\widetilde{i})p-\widetilde{q} = r$. In particular, we have $p\in \sat[\{s_q,i_q\}]{[q]}$ if $r$ is 0. 

In the nondifferential case (or if $p$ contains no proper derivatives of the leader $l_q$ of $q$) we have $(\widetilde{i})p-\widetilde{q} = r$, with $\widetilde{q}\in (q)$ now, and $p\in \sat[\{i_q\}]{(q)}$ if $r$ is 0.  %//there's some ambiguity between (A)\subseteq K[x] and (A)\subseteq K{x}. I don't think it can cause any problems (I checked where pdivsat is used), so leave ambig.
\end{proposition}

%\begin{proposition} \label{degpdiv}
%[TODO degree after reducing]
%\end{proposition}
%\begin{proof}
%This follows from the expression for pseudoremainder $r$ in step 2 of ${\tt DiffPseudoDiv}$.  
%\end{proof}

Just as the process in Example \ref{pdivex} eliminated $y''$ and cut the degree of $y'$ from 2 to 1, more generally we have the following definition:
%\vspace{-.15in}

\begin{definition}
Given a differential polynomial ring $K\{\bx\}$, a differential ranking, and $p,q\in \dkx$, we say $p$ is \emph{partially reduced} with respect to $q$ if no proper derivative of the leader $l_q$ of $q$ appears in $p$.  (In the nondifferential case there are no derivatives and so partial reducedness holds vacuously.) We say $p$ is \emph{reduced} with respect to $q$ if $p$ is partially reduced with respect to $q$ and any instances of $l_q$ in $p$ have strictly lower degree than the maximum degree of $l_q$ in $q$. (In particular, it is impossible to pseudodivide $p$ by $q$ any further.) We say that a subset $A\subseteq \dkx$ is \emph{partially reduced} if all members of $A$ are pairwise partially reduced. We likewise call $A$ \emph{autoreduced} if the elements of $A$ are pairwise reduced. Similarly, we say $p\in \dkx$ is (partially) reduced with respect to $A$ if $p$ is (partially) reduced with respect to each element of $A$. \end{definition}

It follows from Remark \ref{premred} that the pseudoremainder $r$ resulting from pseudodividing $p$ by $q$ is reduced with respect to $q$. After a ``single step" of pseudoreduction as defined in Remark \ref{premred}, the pseudoremainder is not necessarily reduced yet with respect to $q$, but it does have degree in $l_q$ strictly less than that of $p$.

%%%
\subsection{Regular Systems}\label{regsyssec}

Systems of polynomial equations and inequations can ``hide" their information in the sense that deciding the system's properties may require substantial computation. We focus on structured collections known as \emph{regular systems} and \emph{regular sets} that are easier to analyze. The basic intuition is that ``enough" pseudodivision has been done beforehand so that there is no redundancy left to obscure relations between the system's elements. See Hubert's articles \cite{hubert2001notesa,hubert2001notesb} for more information about these systems and algorithms for analyzing them.

\begin{definition}[{\cite[Def. 1]{rga}}] \label{regalgsys}
Let $K$ be a field, let $A,S\subseteq K[\bx]$ be finite with $0\notin S$, and fix a ranking. The equation/inequation pair $(A=0,S\neq 0)$ (or simply $(A,S)$) denoting $f(\bx)=0$ and $g(\bx)\neq 0$ for each $f\in A$ and $g\in S$, respectively, is a \emph{regular algebraic system} over $K$ if 
\begin{enumerate}
\item the elements of $A$ have distinct leaders and
\item  $S$ contains the separant of each element of $A$.
\end{enumerate}
\end{definition}

\begin{definition}[{\cite[Def. 7]{rga}}] \label{regdiffsys}
%good 7-11-22
Let $K$ be a differential field, let $A,S\subseteq K\{\bx\}$ be finite with $0\notin S$, and fix a differential ranking. The equation/inequation pair $(A=0,S\neq 0)$ is a \emph{regular differential system} over $K$ if 
\begin{enumerate}
\item $A$ is partially reduced and the elements of $A$ have distinct leaders and %(i.e., $A$ is \emph{differentially triangular}) and
\item the elements of $S$ are partially reduced with respect to $A$, and $S$ contains the separant of each element of $A$.
\end{enumerate}
\end{definition}
Our statement is simpler than the general definition \cite{rga} because we restrict ourselves to a single derivation. The ``coherence property" in \cite{rga} holds vacuously in the ordinary differential case. Definitions \ref{regalgsys} (1), \ref{regdiffsys} (1) force $A$ to be finite (since there are only $n$ indeterminates), so there is no loss of generality in assuming $A$ is finite to begin with.
%because there are no critical pairs. 

We often omit ``over $K$" when the field is clear from context. Also, nonzero constant multiples in $S$ are irrelevant because we interpret the elements of $S$ as inequations. Thus, for example, if $s_f:=2x$ is the separant of $f\in A$ and $x\in S$, we take that as satisfying the requirement that $S$ contain $s_f$. Similarly, we do not explicitly show nonzero scalars when listing the elements of $S$ even if, like in $x'-xy$, the separant is 1.

An equation/inequation pair $(A=0,S\neq 0)$ from $K[\bx]$ (respectively, $K\{\bx\}$) is an \emph{algebraic system} (respectively, \emph{differential system}) if it is not necessarily regular.

While an arbitrary differential system need not be regular, there always exists a decomposition into one or more regular differential systems (Theorem \ref{rgadecomp}). Our main algorithm $\rgaexp$ from Section \ref{rga} yields such a decomposition if the differential equations have explicit form $\xpef$ (Theorem \ref{modrgacorr}).

\begin{definition}
Let $K$ be a differential field and let $A,S\subseteq K\{\bx\}$. The \emph{restrictions} of $A$ and $S$, respectively, to $K[\bx]$ are $A_{K[\bx]} := A\cap K[\bx]$  and $S_{K[\bx]}:= S \cap K[\bx]$. If $\mathcal{C}:=(A=0,S\neq 0)$ is a differential system, we call $\mathcal{C}_{K[\bx]}:= (A_{K[\bx]}=0,S_{K[\bx]}\neq 0)$ the \emph{restriction} of $\mathcal{C}$ to $K[\bx]$.
\end{definition}

Every regular differential system naturally contains a regular algebraic system.

\begin{lemma}\label{regalg}
Let $\mathcal{C}:=( A=0,S\neq 0$) be a regular differential system. %such that all elements of $S$ are nondifferential polynomials. 
Then the restriction $\mathcal{C}_{K[\bx]}$ of $\mathcal{C}$ to $K[\bx]$ is a regular algebraic system (with respect to the ranking inherited from that of $\mathcal{C}$).
\end{lemma}
\begin{proof}
The elements of $A_{K[\bx]}\subseteq A$ have distinct leaders, as do all elements of $A$. By definition $S_{K[\bx]}=S\cap K[\bx]$; this implies that the separants of elements of $A_{K[\bx]}$ belong to $S_{K[\bx]}$ because $S$ contains the separants of elements of $A\supseteq A_{K[\bx]}$ and the separants of elements of $A_{K[\bx]}$ belong to $K[\bx]$.
\end{proof}

For some purposes regular systems are not strong enough. In particular, they do not suffice for testing membership in $\sat{(A)}$ or $\sat{[A]}$.  For that we need the related notion of a \emph{regular set}.  We first give the definition, but the connection to saturation ideal membership is contained in Theorem \ref{premtest}. This in turn plays an important role in the proof of Theorem \ref{lddecomp} in Section \ref{regsysresults}.

\begin{definition}[{\cite[Def. 2]{BoulierLMP21}}] \label{regsys}
Let $K$ be a field, let $A=\{p_1,\ldots, p_m\}\subseteq K[\bx]$, and fix a ranking. For $1 \leq j \leq m$ let $A_j=\{p_1,\ldots,p_j\}$ and let $S_j$ be the multiplicative set generated by the initials $i_1,\ldots, i_j$ of $p_1,\ldots ,p_j$. We say $A$ is a \emph{regular set} (or \emph{regular chain}) if 

\begin{enumerate}
\item the elements of $A$ have distinct leaders and
\item for $2 \leq j \leq m$, if $q\in K[\bx]$ and $q$ does not belong to $\sat[S_{j-1}]{(A_{j-1})}$, then $i_jq\notin \sat[S_{j-1}]{(A_{j-1})}$. 

\end{enumerate}

\end{definition}

In the usual terminology of the area, $A$ is a \emph{triangular set} (having distinct leaders is suggestive of the triangular shape of an invertible matrix in row echelon form) and for $2 \leq j \leq m$ the initial $i_j$ of $p_j$ is \emph{regular} (i.e., a non-zerodivisor) in the quotient ring $K[\bx]/ (\sat[S_{j-1}]{(A_{j-1})})$.

\subsubsection{Key Results for Regular Systems}\label{keyregres}
To round out the necessary mathematical background, we cite three deeper properties of regular systems and sets that we need in Sections \ref{regsysresults} and \ref{rga}. The first is a technical lemma that, in light of the Nullstellensatz, allows us to go back and forth between zero sets and ideals when dealing with regular algebraic systems.
\begin{lemma}[{Lazard's lemma \cite[Thm. 1.1 and Thm. 4]{rga}}]\label{lazard}
Let $(A=0,S\neq 0)$ be a regular algebraic system. Then $(A):S^{\infty}$ is a radical ideal; i.e., $(A):S^{\infty}=\sqrt{(A):S^{\infty}}$. Likewise, if $(A=0,S\neq 0)$ is a regular differential system, then $\sat{[A]}$ is a radical differential ideal; i.e., $\sat{[A]}=\sqrt{\sat{[A]}}$. 
\end{lemma}
%[TODO: remark about consistency and the convention that we consider the multiplic set generated by $S$ in forming the saturation?]
The second gives a decision procedure for membership in saturation ideals determined by regular sets.

\begin{theorem}[{\cite[Prop. 11]{BoulierLMP21}}{\cite[Thm. 6.1]{AubryLM99}}]\label{premtest}
Let $A$ be a regular set over field $K$ and let $S$ be the multiplicative set generated by the initials of the elements of $A$. Then for any $p\in K[\bx]$, we have $p\in \sat{(A)}$ if and only if the pseudoremainder of $p$ with respect to $A$ is 0 (i.e., the pseudoremainder after reducing as much as possible with respect to all elements of $A$).  
\end{theorem} %//I think the order of elts of A used doesn't matter, but I don't think that's obvious.

This % partial converse to Proposition \ref{pdivsat}
is analogous to how \Grob bases decide membership in arbitrary polynomial ideals \cite[Sect. 2.6, Cor. 2]{clo1_4}. See Theorem \ref{rgadecomp} for a differential version of Theorem \ref{premtest}.

The last result provides a link between differential and algebraic ideals. This is critical because computation becomes less complicated when we do not have to keep differentiating. Moreover, theory and tools for symbolic computation are more developed in the algebraic case than the differential.

\begin{lemma}[{Rosenfeld's lemma  \cite[Thm. 3]{rga}}]\label{rosenfeld}
Let $(A=0,S\neq 0$) be a regular differential system. Then for all differential polynomials $p\in K\{\bx\}$ partially reduced with respect to $A$, we have $p\in [A]:S^{\infty}$ if and only if $p\in (A):S^{\infty}$.
\end{lemma}
\begin{remark}\label{rosenfeldrmk}
Since nondifferential polynomials are (trivially) partially reduced with respect to any set, Rosenfeld's lemma implies that $(\sat{[A]}{}) \cap K[\bx] = (\sat{(A)}{})\cap K[\bx]$. 

The key lesson of Rosenfeld's lemma is that, given a regular differential system, a partially reduced polynomial $p\in \sat{[A]}$ belongs to the differential saturation for essentially algebraic reasons; we do not have to differentiate $A$ to prove it.
\end{remark}
\section{Regular Differential Systems and Algebraic Invariants of Polynomial Vector Fields} \label {regsysresults}

\subsection{Explicit Regular Differential Systems}

We have now covered the background needed for our new method that uses differential elimination to generate algebraic invariants of polynomial dynamical systems. The principal results of the current section are Theorems \ref{satinvar} and \ref{invarcor}. Example \ref{lorenzex1} demonstrates these theorems using the well-known Lorenz system.

To abbreviate theorem statements, we make the following definition that specifies our systems of interest. As stated earlier, we restrict to differential polynomials in explicit form (Definition \ref{explicit}) because we want to apply the theory to finding algebraic invariants of polynomial vector fields. Moreover, this application naturally concerns nondifferential inequations (which we can use, for instance, to indicate ``unsafe" locations/states).

\begin{definition}
Let $K$ be a differential field and let  $\mathcal{C}:=( A=0,S\neq 0$) be a differential system over $K$. We say $\mathcal{C}$ is an \emph{explicit differential system over $K$  with nondifferential inequations} (or just an \emph{explicit system with nondifferential inequations} if $K$ and the differential system are understood) if i) all elements of $A$ that have a proper derivative are in explicit form and ii)  all elements of $S$ are nondifferential polynomials (in which case $S=S_{K[\bx]}=S\cap K[\bx]$). \end{definition}

We prove an important technical lemma that, along with Rosenfeld's lemma (Lemma \ref{rosenfeld}), implies the central result of this section (Theorem \ref{satinvar}).
%[TODO: Confirm general explicit is OK; I suspect requiring $\xpf\in \sat{A}$ ensures explicit of order 1, anyway, but probably not necessary to write out.  I think would be ok w/ extra initial conditions, but best not to fuss, I think.]
\begin{lemma}\label{exist}
Let $\mathcal{C}:=( A=0,S\neq 0$) be an explicit regular differential system over $\mathbb{R}$ with nondifferential inequations. %//I think consistency isn't necessary, but if the algebraic part is consistent, then the whole thing is by orthonomic (not sure if that's necessary) and triangular.  % have the form $x_i'-g_i(\mathbf{x})$ for some variable $x_i$ and algebraic polynomial $g_i$ (i.e., are orthonomic of order 1).  that
Then $((A):S^{\infty})\cap \rx =((A):S^{\infty})_{\rx} = (A_{\rx}):S^{\infty}$.\end{lemma}
%10-12-21 careful check $
%10-13-21 even more careful check; needed fleshing out and more careful notation$$
%10-27-21: removed consistency hypoth; should be harmless, but check briefly
%2-2-22 Also check that I don't necessarily need a diff eqn for each i explicitly in A; right; can choose constant values for those: very carefully checked now

%//upshot is that restriction and saturation don't generally commute; we need a special theorem to get this.
\begin{proof}

The reverse containment $((A):S^{\infty})\cap \rx \supseteq (A_{\rx}):S^{\infty}$ is automatic. For the forward containment, we must show that if $q\in ((A):S^{\infty})\cap \rx$ is a nondifferential polynomial, then $q\in (A_{\rx}):S^{\infty}$. Each explicit differential polynomial in $A$ has the form $x_i'+g_i$ for some variable $x_i$ and nondifferential polynomial $g_i$. Let $I\subseteq \{1,\ldots, n\}$ be the subset of indices $i$ for which such an $x_i'+g_i$ belongs to $A$. The remaining elements of $A$ belong to $A_{\bbr[\bx]}$. Since $q\in (A):S^{\infty}$ and $S=S_{\rx}$, there exist nondifferential polynomials $r\in S_{\rx}$ and $h_j\in A_{\rx}$ as well as differential polynomials $\alpha_j,\beta_i$ such that $rq = \sum_{j}\alpha_j h_j + \sum_{i\in I}  \beta_i(x_i'+g_i)$.

We claim that $q$ vanishes at all \emph{complex} solutions of the restriction $\mathcal{C}_{\bbr[\bx]}$. (In other words, $q\in \vi{\bbc}{\mathbf{V}_{\mathbb{C}}(A_{\bbr[\bx]})\setminus \mathbf{V}_{\mathbb{C}}(\mathrm{\Pi}  S)}$. Recall that $\mathrm{\Pi}  S$ is the product of the (finitely many) elements of $S$.) Let $\mathbf{a}= (a_1,a_2,\ldots, a_n) \in \mathbb{C}^n$ be such that $h(\mathbf{a})=0$ and $s(\mathbf{a})\neq 0$ for all $h \in A_{\rx}$ and $s\in S_{\rx}$ (in particular, the $h_j$ and $r$ from the previous paragraph). By Peano's existence theorem for (complex-valued) ODEs \cite[Thm. X, p. 110]{walter98}, there exists a solution $\mathbf{x}(t)$ to the initial value problem (IVP) $x_i'(t)=-g_i(\mathbf{x}(t))$, $x_i(0)=a_i$, $1\leq i \leq n$. (If $i \notin I$, for the IVP simply let $x_i'(t)=0, x_i(0)=a_i$. For such $i$, the choice  $x_i'(t)=0$ does not affect the following argument.)% (Split each $x_i$ into two real variables, one for the real part and one for the imaginary. Then the existence theorem for real IVPs applies.) 

Substituting $\mathbf{x}(t)$ into the various differential and nondifferential polynomials, we obtain 

\[r(\mathbf{x}(t))q(\mathbf{x}(t)) = \sum_{j}\alpha_j(\mathbf{x}(t)) h_j(\mathbf{x}(t)) + \sum_{i\in I}  \beta_i(\mathbf{x}(t))(x_i'(t)+g_i(\mathbf{x}(t))).\]

\noindent Evaluating at $t=0$, we find 

\[r(\mathbf{a})q(\mathbf{a}) = \sum_{j}(\alpha_j(\mathbf{x}(t))(0)) h_j(\mathbf{a}) + \sum_{i\in I}  (\beta_i(\mathbf{x}(t))(0))(x_i'(0)+g_i(\mathbf{a})).\] 

\noindent (We have written $\alpha_j(\mathbf{x}(t))(0),\beta_i(\mathbf{x}(t))(0) $ instead of $\alpha_j(\mathbf{a}),\beta_i(\mathbf{a})$ because $\alpha_j,\beta_i$ are differential polynomials and the function $\bx (t)$ must be substituted into $\alpha_j,\beta_i$ and differentiated before evaluating at $t=0$. We similarly write $x_i'(0)$ instead of $a_i'$, which is simply 0.) Since $r(\mathbf{a})\neq 0$ and $h_j(\mathbf{a})=0$ (by assumption on $r, h_j$, and $\mathbf{a}$) and  $x_i'(0)+g_i(\mathbf{a})=0$ (because $\bx(t)$ solves the IVP), this  proves that $q(\mathbf{a})=0$ and establishes the claim that $q\in \vi{\bbc}{\mathbf{V}_{\mathbb{C}}(A_{\bbr[\bx]})\setminus \mathbf{V}_{\mathbb{C}}(\mathrm{\Pi}  S)}$.

Lemma \ref{vanishclos} (1) now implies that $q\in \vi{\bbc}{\overline{\mathbf{V}_{\mathbb{C}}(A_{\bbr[\bx]})\setminus \mathbf{V}_{\mathbb{C}}(\mathrm{\Pi}  S)}^{\bbc}}$.  Lemma \ref{satgeom} (2) converts this to $q\in \vi{\bbc}{\vsc{\sat{(A_{\rx})_{\bbc}}}}$, which equals $\sqrt{(A_{\rx})_{\mathbb{C}}:S^{\infty}}$ by the Nullstellensatz (Theorem \ref{nss}). It follows from the definitions of radical and saturation ideals that $(\mathrm{\Pi}  S)^Mq^N\in (A_{\rx})_{\bbc}$ for some $M,N$, so %by Lemma \ref{linfield} 
we have $q\in \sqrt{(A_{\rx}):S^{\infty}}\subseteq \mathbb{R}[\mathbf{x}]$.  %an alternative that doesn't require the restriction yoga is to invoke model completeness and the ACF-agnostic version of the NSS (Thm 2, BoulierEt09); model completeness of ACF shows that vanishing on all the complex soln's of C_alg is enough.
%Similarly, Lazards works for complex, but then I'd need a quick ideal version of Lemma \ref{restrictsat}
Lazard's lemma (Lemma \ref{lazard}) gives $ \sqrt{(A_{\rx}):S^{\infty}}=(A_{\rx}):S^{\infty}$ since  $\mathcal{C}_{\rx}$ is a regular algebraic system by Lemma \ref{regalg}. This completes the proof. 
\end{proof}

The proof of Lemma \ref{exist} illustrates the theme of proving things about real polynomial systems by first going up to the complex numbers  (see Definition \ref{trdef} and the comments preceding it). This strategy also appears in Theorem \ref{invarcor}, Section \ref{lorenzex1}, and Section \ref{check}.

We are ready to give the connection between explicit regular differential systems and algebraic invariants. For convenience going forward, we sometimes refer to this result as the ``regular invariant theorem''.
\begin{theorem}[Regular invariant theorem] \label{satinvar}
Let $\mathcal{C}:=( A=0,S\neq 0$) be an explicit regular differential system over $\mathbb{R}$ with nondifferential inequations. 
%Let $\mathcal{C}:=( A=0,S\neq 0$) be a regular differential system over $\mathbb{R}$ such that all elements of $A$ that have a proper derivative are in explicit form and all elements of $S$ are nondifferential polynomials. Also let $\mathcal{C}_{\rx}:=(A_{\rx}=0,S\neq 0)$ be the regular algebraic system over $\mathbb{R}$ such that $A_{\rx}$ consists of the nondifferential polynomials of $A$. 
Let $\mathbf{x'}=\mathbf{f}(\mathbf{x})$ be a polynomial vector field such that $\mathbf{x'}-\mathbf{f}(\mathbf{x})\in [A]:S^{\infty}$. Then $\mathbf{V}_{\bbr}((A_{\rx}): S^{\infty})$ is an algebraic invariant set of $\mathbf{x'}=\mathbf{f}(\mathbf{x})$.

\end{theorem}
%10-12-21 314 careful check $
%10-15-21 very careful check; needed to split up because of subtleties with totally real...didn't end up being so subtle after all$$
%10-27-21: removed consistency hypoth; should be harmless, but check briefly
%7-8-22: very careful check; done.
%? I think I can strengthen this to *nonempty* alg invar set; check it out in the case that C_alg has a real soln X
\begin{proof}
By Lemma \ref{fullinvarcrit}, it suffices to prove that $(A_{\rx}):S^{\infty}$ is an invariant ideal. Let $p\in (A_{\rx}):S^{\infty}\unlhd \rx$ with the goal of showing $\dot{p}\in (A_{\rx}):S^{\infty}$ (recall that $\dot{p}$ is the Lie derivative of $p$ with respect to $\mathbf{x'}=\mathbf{f}(\bx)$). The saturation $\sat{[A]} \unlhd \bbr\{\mathbf{x}\}$ is a differential ideal, $p\in \sat{(A_{\rx})}\subseteq [A]:S^{\infty}$,  and by assumption $\mathbf{x'}-\mathbf{f}(\mathbf{x})\in [A]:S^{\infty}$, so by Lemma \ref{subslie} we obtain $\dot{p} \in [A]:S^{\infty}$. As $\dot{p}$ is nondifferential and thus partially reduced with respect to $A$, Rosenfeld's lemma (Lemma \ref{rosenfeld}) yields $\dot{p} \in (A):S^{\infty} \unlhd \bbr\{\mathbf{x}\}$. Because $\dot{p}\in \rx$, Lemma \ref{exist} implies %allows us to replace $(A) \unlhd \, \bbr\{\mathbf{x}\}$ with $(A_{\rx}) \unlhd \, \rx$, establishing  
$\dot{p} \in (A_{\rx}):S^{\infty}$ as desired. %Hence $(A_{\rx}):S^{\infty}$ is an invariant ideal.and the real zero set of  $(A_{\rx}):S^{\infty}$ forms an algebraic invariant set of $\mathbf{x'}=\mathbf{f}(\mathbf{x})$. 
\end{proof}

While it is not necessary that $\mathbf{x'}-\mathbf{f}(\mathbf{x})\in [A]$ for the regular invariant theorem to hold (see the example in Section \ref{lorenzex1}), the hypothesis that $\mathbf{x'}-\mathbf{f}(\mathbf{x})\in [A]:S^{\infty}$ implies that $\mathbf{x'}=\mathbf{f}(\bx)$ is a differential-algebraic consequence of $(A=0,S\neq 0)$ (by the differential Nichtnullstellensatz, Theorem \ref{diffsatrad}.)
\emph{How} we obtain explicit regular differential systems in the first place is a central topic in Section \ref{rga}.

\subsection{Alternate Representations from Additional Hypotheses}
As we will see in Section \ref{lorenzex1} and Section \ref{rga}, Theorem \ref{satinvar} opens the door to novel ways of finding and analyzing algebraic invariants of polynomial vector fields. However, it poses the challenge of finding generators of the saturation ideal $\sat{(A_{\rx})}$ if we want to explicitly write equations for the invariant. While this is possible using \Grob bases \cite[p. 205]{clo1_4}, it adds complexity to the process (Remark \ref{eqeffrmk}). We would much prefer to read off the invariant directly from $A_{\rx}$ and $S$. To explore this possibility, we start by noting that 
$\mathbf{V}_{\bbr}((A_{\rx}): S^{\infty})$ is sandwiched between two more convenient sets. (The following lemma does not depend on regular systems, so we use generic names in place of $A_{\rx}$ and $S$.)

\begin{lemma}  \label{closcont} If $B,C\subseteq \rx$, with $0\notin C$ finite, then $\overline{\mathbf{V}_{\bbr}(B)\setminus \mathbf{V}_{\bbr}(\mathrm{\Pi}  C)}^{\bbr\text{-euc}}\subseteq \mathbf{V}_{\bbr}((B): C^{\infty}) \subseteq \vsr{B}$.
\end{lemma}
%7-8-22 careful
\begin{proof}
We have $\overline{\mathbf{V}_{\bbr}(B)\setminus \mathbf{V}_{\bbr}(\mathrm{\Pi}  C)}^{\bbr\text{-euc}}= \overline{\mathbf{V}_{\bbr}(B)\setminus \mathbf{V}_{\bbr}(\mathrm{\Pi}  C)}^{\bbr}$ by Lemma \ref{euclidclos} since $\mathbf{V}_{\bbr}(B)\setminus \mathbf{V}_{\bbr}(\mathrm{\Pi}  C)$ is a real constructible set. (Recall that $\overline{X}^{\bbr\text{-euc}}$ denotes the Euclidean closure of a set $X\subseteq \bbr^n$. Though in this case the closures coincide, we invoke the Euclidean topology because of its visually intuitive nature compared to the Zariski topology.)  

The containment  $\overline{\mathbf{V}_{\bbr}(B)\setminus \mathbf{V}_{\bbr}(\mathrm{\Pi}  C)}^{\bbr} \subseteq  \mathbf{V}_{\bbr}((B): C^{\infty})$ holds because $\mathbf{V}_{\bbr}((B): C^{\infty})$ is a real Zariski closed set containing $\mathbf{V}_{\bbr}(B)\setminus \mathbf{V}_{\bbr}(\mathrm{\Pi}  C)$ (this follows quickly from the definition of a saturation ideal; see the $\supseteq$ case in the proof of Lemma \ref{satgeom} (2)). The last containment holds because $(B)\subseteq (B): C^{\infty}$. 
\end{proof}

It is possible that both $\overline{\mathbf{V}_{\bbr}(A_{\rx})\setminus \mathbf{V}_{\bbr}(\mathrm{\Pi}  S)}^{\bbr\text{-euc}}$ and $\vsr{A_{\rx}}$, in addition to $\vs[\bbr]{\sat{(A_{\rx})}}$, must be invariant under the hypotheses of the regular invariant theorem. %Furthermore, it appears that all three sets are typically equal if $\mathcal{C}_{\rx}$ has a solution in $\bbr$. 
We cannot yet prove or disprove this conjecture. However, we \emph{can} prove invariance of the various sets using additional hypotheses that are commonly satisfied (see also Theorem \ref{radinvartest}, the discussion following it, and Remark \ref{eqeffrmk}): % for a related use of totally real varieties

\begin{theorem}[Alternative criteria for regular invariants]  \label{invarcor} %Use the same notation and hypotheses as in Theorem \ref{satinvar}.  
Let $\mathcal{C}:=( A=0,S\neq 0$) be an explicit regular differential system over $\mathbb{R}$ with nondifferential inequations. Let $\mathbf{x'}=\mathbf{f}(\mathbf{x})$ be a polynomial vector field such that $\mathbf{x'}-\mathbf{f}(\mathbf{x})\in [A]:S^{\infty}$. Then each of the following conditions is sufficient for the indicated set to be an algebraic invariant set of $\mathbf{x'}=\mathbf{f}(\mathbf{x})$.
\begin{enumerate}
\item If $\mathbf{V}_{\mathbb{C}}(A_{\rx})\setminus \mathbf{V}_{\mathbb{C}}(\mathrm{\Pi}  S)$ is a totally real constructible set, then $\overline{\mathbf{V}_{\bbr}(A_{\rx})\setminus \mathbf{V}_{\bbr}(\mathrm{\Pi}  S)}^{\bbr\text{-euc}}$ $=\mathbf{V}_{\bbr}((A_{\rx}): S^{\infty}) $ is invariant. 
% 10-19-21 1: good check
%11-1-21 2: careful check
%11-1-21 3: careful check
%1-14-22 : remembered to put  $\mathbf{x'}-\mathbf{f}(\mathbf{x})\in [A]:S^{\infty}$ into 1,2
% 7-8-22 careful check
\item  If $\mathbf{V}_{\mathbb{R}}(A_{\rx})$ is irreducible over $\bbr$ and $\mathbf{V}_{\bbr}(A_{\rx}) \setminus \mathbf{V}_{\bbr}(\mathrm{\Pi}  S)$ is nonempty (i.e., $\mathcal{C}_{\rx}$ has a real solution), then  $\overline{\mathbf{V}_{\bbr}(A_{\rx}) \setminus \mathbf{V}_{\bbr}(\mathrm{\Pi}  S)}^{\bbr\text{-euc}}=\mathbf{V}_{\bbr}((A_{\rx}): S^{\infty})$ $=\mathbf{V}_{\mathbb{R}}(A_{\rx})$ is invariant.
\item Suppose %$(A_{\rx})$ is radical
$[A]\cap \rx = (A_{\rx})$ and for each $q\in A_{\rx}$ and monomial $ux_i'$ in the derivative $q'$, we have $u(x_i' -f_i(\bx))\in [A]$. Then $\mathbf{V}_{\mathbb{R}}(A_{\rx})$ is invariant.
 (The first condition strengthens Rosenfeld's lemma (Lemma \ref{rosenfeld}) and the second strengthens the requirement $\mathbf{x'}-\mathbf{f}(\mathbf{x})\in [A]:S^{\infty}$ from Theorem \ref{satinvar}.)

\end{enumerate}
\end{theorem}
\begin{proof}
\begin{enumerate}
\item The following chain of equalities establishes the claim:
 
 \begin{align*}
 \overline{\mathbf{V}_{\bbr}(A_{\rx})\setminus \mathbf{V}_{\bbr}(\mathrm{\Pi}  S)}^{\bbr \text{-euc}} &= \overline{\mathbf{V}_{\bbr}(A_{\rx})\setminus \mathbf{V}_{\bbr}(\mathrm{\Pi}  S)}^{\bbr} && \text{\hspace{.5cm}(Lemma  \ref{euclidclos})}\\
 &=  \overline{\mathbf{V}_{\bbr}(A_{\rx})\setminus \mathbf{V}_{\bbr}(\mathrm{\Pi}  S)}^{\mathbb{C}} \cap \mathbb{R}^n && \text{\hspace{.5cm}(Lemma  \ref{restrictclos})}\\
 &= \overline{\mathbf{V}_{\mathbb{C}}(A_{\rx})\setminus \mathbf{V}_{\mathbb{C}}(\mathrm{\Pi}  S)}^{\mathbb{C}} \cap \mathbb{R}^n && \text{\hspace{.5cm}(totally real constructible set)}\\
&=  \mathbf{V}_{\mathbb{C}}((A_{\rx})_{\mathbb{C}}: S^{\infty})\cap \mathbb{R}^n && \text{\hspace{.5cm}(Lemma  \ref{satgeom} (2))}\\ %annih gens, annih ideal
 &= \mathbf{V}_{\mathbb{C}}((A_{\rx}): S^{\infty})\cap \mathbb{R}^n && \text{\hspace{.5cm}(Lemma \ref{restrictsat})}\\
 &= \mathbf{V}_{\bbr}((A_{\rx}): S^{\infty}), 
 \end{align*}
 \noindent which is invariant by the regular invariant theorem.
 \item Since $\overline{\mathbf{V}_{\bbr}(A_{\rx})\setminus \mathbf{V}_{\bbr}(\mathrm{\Pi}  S)}^{\bbr\text{-euc}}\subseteq \mathbf{V}_{\bbr}((A_{\rx}): S^{\infty}) \subseteq \vsr{A_{\rx}}$ by Lemma \ref{closcont}, it suffices by the regular invariant theorem to prove equality of the first and last sets. As in part 1, the real Euclidean closure equals the real Zariski closure. Then $\overline{\mathbf{V}_{\bbr}(A_{\rx})\setminus \mathbf{V}_{\bbr}(\mathrm{\Pi}  S)}^{\bbr}$ equals $\vsr{A_{\rx}}$ %by Lemma \ref{irredclos} 
since by hypothesis $\vsr{A_{\rx}}$ is irreducible and $\mathbf{V}_{\bbr}(A_{\rx}) \setminus \mathbf{V}_{\bbr}(\mathrm{\Pi}  S)$ is nonempty.
 \item As indicated, the hypotheses are chosen to mimic the proof of the regular invariant theorem using  $(A_{\rx})$ instead of $(A_{\rx}):S^{\infty}$. In particular, let $p\in A_{\rx}$ with the goal of showing $\dot{p}\in (A_{\rx})$. (By Lemma \ref{ldderiv} it suffices to consider generators of $(A_{\rx})$.)  Now $p' \in [A]$ and it follows from the assumption about monomials in derivatives of elements of $A_{\rx}$ that $\dot{p} \in [A]$. (Replace each $x_i'$ in $p'$ with $x_i'-f_i(\bx)+f_i(\bx)$ and distribute. The assumption about monomials implies that $p'=\dot{p} + \text{an element of } [A]$.) Since $\dot{p}$ is nondifferential and by assumption $[A]\cap \rx = (A_{\rx})$, we conclude that $\dot{p}\in (A_{\rx})$.
 
 \end{enumerate}
 %Since $\mathbf{V}(A_{\rx})\setminus \mathbf{V}(S)$ is a real constructible set, by Lemma \ref{euclidclos} we have $\overline{\mathbf{V}(A_{\rx})\setminus \mathbf{V}(S)}^{euc}= \overline{\mathbf{V}(A_{\rx})\setminus \mathbf{V}(S)}$.

%\cite[Cor. 3]{rga} and \cite[Thm. 10(iii), p. 203]{clo1_4}, the \emph{complex} zero set $\mathbf{V}_{\mathbb{C}}((A_{\rx}): S^{\infty})$ is the complex Zariski closure of the complex solutions of $\mathcal{C}_{\rx}$. This closure of the complex solutions equals the complex Zariski closure of the real solutions of $\mathcal{C}_{\rx}$ because $\mathbf{V}_{\mathbb{C}}((A_{\rx}): S^{\infty})$ is totally real. By Lemma \ref{restrictclos}, the real points of....  Restricting to $\mathbb{R}$, we find that the real Zariski closure of the real solutions is an invariant set of $\mathbf{x'}=\mathbf{f}(\mathbf{x}).$ (The real Zariski topology on $\mathbb{R}^n$ is the restriction of the complex Zariski topology on $\mathbb{C}^n$ 

\end{proof}

%[alternative statements of some of them and examples (confirm that the various hypotheses hold)]

\begin{remark}\label{primecor}
A common way for $[A]\cap \rx = (A_{\rx})$ to hold is for $(A_{\rx})\unlhd \, \rx$ to be a prime ideal containing no element of $S$.  Clearly $(A_{\rx})\subseteq [A]\cap \rx$. For the other containment, note that $[A]\cap \rx \subseteq (\sat{[A]}{} )\cap \rx$, which equals $(\sat{(A)}{})\cap \rx$ by Rosenfeld's lemma (see Remark \ref{rosenfeldrmk}). In turn, $(\sat{(A)}{})\cap \rx= \sat{(A_{\rx})}{}=(A_{\rx})$ by Lemma \ref{exist}, Lemma \ref{satprime}, and the fact that $(A_{\rx})$ is prime and has no element of $S$.

\end{remark} %7-8-22 careful; leaves S vs mult set gen by S under rug; primality makes this harmless

We illustrate Theorems \ref{satinvar} and \ref{invarcor} with a nontrivial example from the physical sciences. %several nontrivial examples coming from scientific and engineering applications. 
The various hypotheses--in spite of their seemingly technical statements--are all satisfied and readily checked. 

\subsection {Example: Lorenz equations}\label{lorenzex1}
%good final 7-8-22
%good check 6-10-22; maple mdgb...

The ODEs $x' = \sigma(y - x), y'=\rho x -y-xz,  z'=xy-\beta z$ comprise the famous \emph{Lorenz equations} \cite{lorenz1963deterministic}. Depending on the parameters, this nonlinear system can display widely varying behavior, including chaotic dynamics \cite{sparrow2012lorenz}. The literature also contains studies of algebraic invariants for the Lorenz system \cite{{llibre2002invariant},{swinnerton2002invariant}} . The benchmark collection from \cite{SogokonMTCP21} considers the parameters $\sigma=1, \rho = 2, \beta=1$, yielding the particular equations $x' = y - x, y'=2x -y-xz,  z'=xy-z$ that we abbreviate as $\mathbf{x'}=\mathbf{f}(\mathbf{x})$.

 Fix an orderly ranking with $x>y>z$. We analyze $A= \{y'-2x +y+xz,  z'-xy+z,2x^2 - y^2 - z^2\}, S= \{x\}$. Thus $A_{\rx}=\{2x^2 - y^2 - z^2\}$. We discuss how we obtained these sets in Section \ref{lorenzex2}; for now we take them as given and confirm the hypotheses of our preceding theorems. %[TODO: graphic?]  

We claim $\mathcal{C}:= (A=0,S\neq 0)$ satisfies the regular invariant theorem (Theorem \ref{satinvar}) and each part of Theorem \ref{invarcor}. In particular, the strongest conclusion holds: $\overline{\mathbf{V}_{\bbr}(A_{\rx}) \setminus \mathbf{V}_{\bbr}(\mathrm{\Pi}  S)}^{\bbr\text{-euc}}=\mathbf{V}_{\bbr}((A_{\rx}): S^{\infty}) = \mathbf{V}_{\mathbb{R}}(A_{\rx})$ is invariant with respect to $\mathbf{x'}=\mathbf{f}(\mathbf{x})$. 

\smallskip

\noindent \emph{(Theorem \ref{satinvar} applies)}
Clearly, $\mathcal{C}$ is an explicit regular differential system over $\mathbb{R}$ with nondifferential inequations. Differential polynomials $y'-2x +y+xz$ and $z'-xy+z$ belong to $A$, but the presence of $p:=2x^2 - y^2 - z^2$ implies that $x'-y+x$ cannot also be in $A$ lest $A$ not be partially reduced. However, we can check $x'-y+x\in \sat{[A]}{}$ with standard computer algebra systems (CAS) by confirming $x(x' -y +x)\in (y'-2x +y+xz,  z'-xy+z,2x^2-y^2-z^2, 4xx' -2yy' -2zz'=(2x^2-y^2-z^2)')$, where $x',y',z'$ are new algebraic indeterminates. That is, we consider ideal membership in the nondifferential polynomial ring $\bbr[x,y,z,x',y',z']$.  This establishes $\mathbf{x'}-\mathbf{f}(\mathbf{x})\in [A]:S^{\infty}$ and Theorem \ref{satinvar} implies that $\mathbf{V}_{\bbr}((A_{\rx}): S^{\infty})$ is an algebraic invariant of $\mathbf{x'}=\mathbf{f}(\mathbf{x})$.

\smallskip
\noindent \emph{(Theorem \ref{invarcor} applies)}
To check the additional hypotheses in Theorem \ref{invarcor}, we first show that $p=2x^2 - y^2 - z^2\in \rx$ is irreducible over  $\bbc$. (The weaker condition of irreducibility over $\bbr$ suffices for parts 2 and 3 of Theorem \ref{invarcor}, but we prefer to prove the stronger result that is helpful for part 1.) In this case it is simple to work directly: reducibility would imply that $p= (a_1x+b_1y+c_1z)(a_2x+b_2y+c_2z)$ for some $a_i,b_i,c_i\in \bbc$. But distributing and comparing to the coefficients of $2x^2 - y^2 - z^2$ gives an inconsistent system: $a_1a_2=2, b_1b_2= c_1c_2=-1, a_1b_2+a_2b_1=a_1c_2+a_2c_1 = b_1c_2+b_2c_1=0.$ (Inconsistency is conveniently shown by using a CAS to conclude $1\in (a_1a_2-2, b_1b_2+1, c_1c_2+1, a_1b_2+a_2b_1,a_1c_2+a_2c_1, b_1c_2+b_2c_1))$. Thus the polynomial $p$ and the variety $\vsc{A_{\rx}}$ are irreducible over $\bbc$.

\begin{enumerate}
\item \emph{(Part 1)}
Since $\vsc{A_{\rx}}$ only has one irreducible component, we just need to find one real smooth point to prove that $\vsc{A_{\rx}}$ (and hence $\vsc{A_{\rx}}\setminus \vsc{\mathrm{\Pi}  S}$, by Lemma \ref{trtrcons}) is totally real.  An obvious choice is $(1,1,1)\in \bbr ^3$, which is smooth because $(p)=(A_{\rx})$ is prime (and hence radical) and $(\frac{\partial p}{\partial x}, \frac{\partial p}{\partial y},\frac{\partial p}{\partial z})$ evaluated at $(1,1,1)$ is $(4,-2,-2)\neq \mathbf{0}$. Alternatively,  $\vsc{A_{\rx}}$ is totally real by Theorem \ref{signchange} since $p$ is irreducible, $p(1,0,0)>0$, and $p(0,1,0)<0$. It follows that $\mathcal{C}$  satisfies Theorem \ref{invarcor} (1).
% Real irredu: (Equation $c_1c_2=-1$ shows that $c_1,c_2$ must be nonzero and have opposite sign. Multiply $ b_1c_2+b_2c_1=0$ by $b_2$ and use $b_1b_2=-1$ to conclude $-c_2 +b_2^2c_1=0$, which is impossible by the preceding sentence.)

\item \emph{(Part 2)} Irreducibility over $\bbc$ implies irreducibility over $\bbr$ and $(1,1,1)\in \mathbf{V}_{\bbr}(A_{\rx}) \setminus \mathbf{V}_{\bbr}(\mathrm{\Pi}  S)$, so $\mathcal{C}$ satisfies part 2. %and we conclude $\overline{\mathbf{V}_{\bbr}(A_{\rx}) \setminus \mathbf{V}_{\bbr}(\mathrm{\Pi}  S)}^{\bbr\text{-euc}}$ $=\mathbf{V}_{\bbr}((A_{\rx}): S^{\infty})$ $= \mathbf{V}_{\mathbb{R}}(A_{\rx})$ is invariant with respect to $\mathbf{x'}=\mathbf{f}(\mathbf{x})$.

\item\emph{(Part 3)} %For the first condition of part 3, note that $(A_{\rx})$ is a prime ideal in $\rx$ since $2x^2 - y^2 - z^2$ is irreducible over $\bbr$. Moreover, 
Ideal $(A_{\rx})$ is prime and $(A_{\rx})\cap S$ is empty since $2x^2 - y^2 - z^2$ does not divide $x$. Remark \ref{primecor} thus implies that $[A]\cap \rx = (A_{\rx})$.

\indent \hspace*{15pt} Lastly, since $(2x^2-y^2-z^2)' =4xx' -2yy' -2zz'$, we must show that $4x(x' -y +x), -2y(y'- 2x +y+xz)$, and $-2z(z'-xy+z)$ belong to $[A]$. This was done above in our proof of $\mathbf{x'}-\mathbf{f}(\mathbf{x})\in [A]:S^{\infty}$.
\end{enumerate}

%[Something about finding generators of the saturation (can be done w/ GB; might be simpler than usual GB because you already have a reg sys); can also do geometrically because even for non-ACF the zero set of a radical sat'n ideal can be written as the intersection of associated primes that miss S. (explain the geom picture?) ]%Trying to figure out totally real hypotheses as well as transversality hypotheses to get rid of inequations.]

%[Example]

%[I think the following is how things should be for reg alg sys A,S: 1. if A is autoreduced, I think A=0 should define an invariant set (by the older version (not BoulierEt09's for systems) of Rosenfeld's lemma) 2. If A is absolutely irred, I think A=0 <=> cl(A=0, S<>0)^R <=> V_R(A:S^infty) is an invar 3. (proven already by me) if V_C(A)\V_C(S) is t.r. const set (including if V_C(A) is t.r.), then the cl(A=0, S<>0)^R <=> V_R(A:S^infty) is an invar 4. (harder/less clear; probably leave for future paper) I think cl(A=0, S<>0)^R should always be invar; if empty, trivial; if not empty, geom (pp 325-6) seems plausible that I should be able to prove the elusive equality of zero sets on pp325-6  (equality of the sat'n and the rrad of the sat'n, and also their zero loci, is not true in general as shown here pp.~325-330; not sure if the ideals should be equal, or even their rrads, if the sys is solvable in R) 5. (proven already by me) V_R(A:S^infty) is an invar]%

\section{The Rosenfeld-\Grob Algorithm for Algebraic Invariants}\label{rga}
A modified version of the Rosenfeld-\Grob algorithm (RGA) of Boulier et al. \cite{{rga},{BoulierLOP95}} is our main tool for differential elimination. We explain the basic ideas behind RGA and then formulate $\rgaexp$, our handcrafted version that produces explicit regular differential systems with nondifferential inequations (and hence invariants by the regular invariant theorem). The subscript \emph{o} stands for ``ordinary'' because our setting involves ODEs of a special form. %and elaborate on our intended uses in Examples \ref{rgaalg}, \ref{rgaaero}, \ref{rgasemi}, and \ref{closureex}.  The first two examples concern algebraic constraints while the latter two consider more general semialgebraic sets.

\subsection{$\rgaexp$ for Explicit Systems}\label{rgaexp}

Unlike \Grob basis techniques that output generators of an ideal, %\cite[p.77]{clo1}, 
RGA uses a differential generalization of \emph{characteristic sets}. Given a ranking, a characteristic set $C$ of ideal $I$ is by definition a minimal-rank autoreduced subset of $I$ (Remark \ref{rankset}). %such that no element of $C$ can be further pseudoreduced with respect to any of the others (we say $C$ is \emph{autoreduced}). 
Such a $C$ is not necessarily unique and might not be a generating set of $I$, but $C$ is nonempty, finite, and pseudoreduces every element of $I$ to zero \cite[p. 175]{mishra}. The same definition and properties hold for \emph{differential characteristic sets} if we use \emph{differential} rankings, \emph{differential} ideals, and \emph{differential} pseudoreduction.
\phantomsection
\label{genericchain}
Given a differential ranking and a system of polynomial differential equations and inequations, RGA outputs a finite collection of special differential characteristic sets called \emph{regular differential chains} \cite{BoulierL10} or simply \emph{chains}; see parts 2, 3 of Theorem \ref{rgadecomp}. (Recall Definition \ref{regsys} and Theorem \ref{premtest} for the nondifferential version.) Multiple chains in the output come from case splits over the vanishing of initials and separants. For this reason, RGA typically outputs a ``generic" chain (obtained by placing initials and separants with the inequations as much as possible while maintaining consistency) and several more specific ones; see \cite{Hubert99} for precise definitions and Section \ref{lorenzex2} for an example that continues the one in Section \ref{lorenzex1}.

Regular differential chains contain important ``geometric" information. Take an input differential system $(A=0,S\neq 0)$ of equations and inequations and perform RGA. Then a differential polynomial $p$ is zero at all points that satisfy $(A=0,S\neq 0)$  if and only if $p$ is differentially pseudoreduced to zero by each chain that RGA returns given input $(A=0,S\neq 0)$ \cite[Cor. 3 and Thm. 9]{rga}. By the differential Nichtnullstellensatz (Theorem \ref{diffsatrad}), this gives an algorithm for testing \emph{radical} differential saturation ideal membership. This is especially noteworthy because, unlike the algebraic case, general differential ideal membership is undecidable (at least for partial differential polynomial rings with two or more derivations; the ordinary case remains open) \cite{umirbaev2016algorithmic}.

We cite the important properties guaranteed by RGA. Compare Theorem \ref{rgadecomp} (due to Boulier et al.) below to Theorems \ref{modrgacorr} and \ref{radlddecomp}, our main results in this section.
\begin{theorem}[{\cite[Thm. 9]{rga}}]\label{rgadecomp}
Let $K$ be a differential field of characteristic 0, let $A,S\subseteq K\{\bx\}$ be finite with $0\notin S$, and fix a differential ranking.
%Then there exist differential systems $(A_1,S_1),\ldots, (A_r,S_r)$ such that 

Then RGA applied to $(A,S)$ returns differential systems $(A_1,S_1),\ldots, (A_r,S_r)$ such that 
\begin{enumerate}
\item $\sqrt{\sat{[A]}}= (\sat[S_1]{[A_1]})\cap \cdots \cap (\sat[S_r]{[A_r]})$,
\item each $(A_i,S_i)$ is a regular differential system (together with item 3, this makes each $A_i$ a regular differential chain), and
\item  for all $p\in K\{\bx\}$ and $1\leq i\leq r$, we have $p\in \sat[S_i]{[A_i]}$ if and only if the differential pseudoremainder of $p$ with respect to $A_i$ is 0.
%\item if $K$ is a \emph{computable} differential field (e.g., $\mathbb{Q}$; more generally, addition, multiplication, differentiation, and checking equality with 0 can be done algorithmically in $K$), then RGA computes the $(A_i,S_i)$ from $(A,S)$.
\end{enumerate}
\end{theorem}

Strictly speaking, for differential elimination results to be computationally meaningful we must be able to algorithmically add, multiply, divide, differentiate, and check equality with 0. We always tacitly assume this about the finitely many differential field elements that appear during a computation. (Our inputs are finite sets of differential polynomials and so there are only finitely many coefficients; all intermediate elements result from these via arithmetic operations or differentiation.) The assumption is mild because in practice coefficients typically belong to $\mathbb{Q}$. 

\phantomsection
\label{rgaapp}
The original authors of RGA gave two versions of the algorithm \cite[pp. 162-3]{BoulierLOP95} \cite[p. 111]{rga}.  %(RGA further applies to polynomial \emph{partial differential equations}, but this proposal centers on ODEs.)
RGA has been implemented in the Maple computer algebra system, where it forms the heart of the ${\tt DifferentialAlgebra}$ package \cite{mapleDiffAlg}. While this tool is convenient (for example, we use it in the example from Section \ref{lorenzex2}), the proprietary nature of Maple impedes a full analysis of the implementation and its performance. As an alternative, Boulier makes freely available the C libraries on which the Maple implementation is based \cite{blad}. RGA has proven its versatility by admitting refinements and extensions over the past two decades \cite{{fakouri2018new},{HashemiT14},{rga},{GolubitskyKMO08},{BouzianeKM01},{Hubert00}}.

%For another variant see also\cite[pp. 162-3]{BoulierLOP95} \cite[ p. 587]{GolubitskyKMO08}\cite[p. 111]{rga}.

%[TODO: include in the simplified version the convention of preserving the original $x'-f(x)$. Why this is only a sketch and why we don't bother to represent w/ pseudocode: our purpose is to explain the properties of a family of implementations; the theorems depend only on guaranteeing a particular form of output, and optimizing the details for conducting the calculations is out of scope. There have been several versions: Boulier95, 09, Ovchinn 08, etc. Moreover, for practicality, need a lot of other things like heuristics that are not part of the current story.]

Published applications of RGA include parameter estimation for continuous dynamical systems \cite{{UshirobiraEB19},{VerdiereZD18}} and preprocessing of %DAE
systems for later numerical solution \cite{boulier2008differential}. The literature contains various case studies from control theory \cite{HarringtonHM19}, medicine \cite{hong2016minimal}, and mathematical biology \cite{boulier2006differential}.% that employ RGA. %The algorithm is very general, applies to ODE's, DAE's, and their partial differential analogues, and does not depend on any particular class of polynomial differential equation. 

%\begin{lemma}
%Informal statement: RGA preserves $\mathbf{x'}=\mathbf{f}(\mathbf{x})$. (satisfaction and orthonomic form in the output chain)
%\end{lemma}
%\begin{proof}
%(under construction) ~1-20-22 370.

%\end{proof}

We give a modified version of RGA that we call $\rgaexp$ and that is well-suited for analyzing systems in explicit form. In some ways our algorithm is simpler than those in the literature. In particular, we produce regular differential systems but do not guarantee that they are regular differential chains. This helps us control the form of the output and more easily obtain explicit bounds. Our application to algebraic invariants (culminating in Theorem \ref{radlddecomp}) does not require radical differential ideal membership testing, so regular differential systems suffice for our use case.

First we describe a purely algebraic algorithm, ${\tt Triangulate}$, that we use as a subroutine in $\rgaexp$. The general structure of ${\tt Triangulate}$ mirrors that of $\rgaexp$ but does not have to deal with derivatives. After presenting both algorithms, proving them correct, and interpreting the output, we analyze their complexity. %Though we assume coefficients in $\bbr$, the same algorithms work for any computable field of characteristic 0. [TODO: adjust to make precise where we need to be in $\bbr$; careful w/ def of computable field]

%[TODO: quick intuitive summary, emphasizing the recursive nature] 
Intuitively, ${\tt Triangulate}$ is a recursive divide-and-conquer algorithm that decomposes the radical of a saturation ideal into an intersection of radical ideals determined by regular algebraic systems. (The same high-level description also applies to $\rgaexp$, but in that case using \emph{differential} polynomial rings and ideals.)  %${\tt Triangulate}$ only pseudodivides when the initial in question is assigned to the inequations.
%In the case of inequations, we then perform pseudodivision.

%//RESPECT: must immediately precede Triangulate//
\begin{remark} \label{singlestepconv}
To facilitate the complexity analysis in Section \ref{rgabds} (e.g., Theorem \ref{cxtyTri}), we use the following convention in describing ${\tt Triangulate}$ and $\rgaexp$: unless stated otherwise, ``pseudodivision" refers to a single step of pseudodivision as discussed in Remark \ref{premred}.
\end{remark}
%7-9-22 good final
%6-15-22 good check

\noindent ${\tt Triangulate}$: \label{Triangulate}
\smallskip

	\begin{itemize}
    \item \emph{Informal summary}: ${\tt Triangulate}$ transforms an input pair $(A,S)$ into more refined pairs that are ``closer" to being regular algebraic systems. At each step, $\tri$ includes an initial or separant with either the equations or inequations and reduces via pseudodivision. Then $\tri$ is called on the resulting pairs. This recursion creates a tree whose leaves are the output of the algorithm.  

    \item \emph{Detailed description}:
    \begin{itemize}
	\item Choose the ranking $x_n> x_{n-1} > \cdots > x_1$. (The particular choice is not essential; we specify a ranking for concreteness.) %because we never allow derivatives to exceed order 1.)

	 \item Input: a finite set of polynomials $A=\{p_1,p_2,\ldots,p_m\}\subseteq \rx$ that determines the equations $p_1=0,\ldots,p_m=0$. Let $S$ be a finite set of polynomials in $\rx\setminus \{0\}$ corresponding to inequations. (That is, the input is the algebraic system $(A=0,S\neq 0)$). We abuse terminology slightly by calling $A$ itself a set of equations and $S$ itself a set of inequations.

 	\item Output: Pairs $(A_1,S_1),\ldots, (A_r,S_r)$ such that $A_1, \ldots, A_r \subseteq \rx, S_1,\ldots, S_r \subseteq \rx\setminus\{0\}$ are finite sets of polynomials and

\smallskip

 \[\sqrt{\sat{(A)}}=(\sat[S_1]{(A_1)})\cap \cdots \cap (\sat[S_r]{(A_r)}),\] 
 
 \smallskip
 \noindent where $(A_i=0,S_i\neq 0)$ is a regular algebraic system (i.e., $A_i$ is triangular and $S_i$ contains at least the separants of $A_i$).

	\end{itemize}

\begin{enumerate}

\item If $A$ is already triangular and separant $s_q$ belongs to $S$ for every $q\in A$, then return $(A,S)$. %(If an initial or separant is a constant (necessarily nonzero), then we always consider it to belong to $S$. Moreover, we also consider nonzero constant multiples of elements of $S$ to belong to $S$.)

Otherwise, choose the highest-ranking leader $x$ that either appears in multiple elements of $A$ or appears in only one $q\in A$ but $s_q \notin S$. We call $x$ the \emph{target variable} for the pair $(A,S)$. Choose some $q\in A$ that has minimal degree in $x$ among all elements of $A$ having leader $x$; if there are multiple such polynomials, pick one that has least total degree among those with minimal degree in $x$ (again there might be several).

We define two auxiliary sets that we need in steps 2 and 3.

\begin{itemize}
\item Let $\widetilde{A}:= (A\cup \{i_{q}, q- i_{q}x^{deg_{x}(q)}\})\setminus \{q\}$, where $q- i_{q}x^{deg_{x}(q)}$ is the \emph{tail} of $q$ (i.e., what is left of $q$ if we substitute zero for the initial of $q$). 

The notation $deg_{x}(q)$ represents the degree of variable $x$ in polynomial $q$.  The set $\widetilde{A}$ is finite; the parentheses around $A\cup\{i_q, \ldots \}$ separate the union from the set difference and do not indicate an ideal.\phantomsection \label{trinote} Note that the solutions of $(\widetilde{A}=0,S\neq 0)$ are the same as those of $(A\cup\{i_q\}=0,S\neq 0)$. This replacement is not technically pseudodivision, but it behaves similarly and in Theorem \ref{triangcorr} (correctness of $\tri$) and Lemma \ref{difftriang} we analyze this case together with the pseudodivision steps. \phantomsection
\label{tristep1}

Note that $x$ does not appear in $i_{q}$ and $deg_{x}(q- i_{q}x^{deg_{x}(q)})< deg_{x}(q)$. Also note that the target variable of $(\widetilde{A},S)$ could still be $x$ or might have strictly lower rank, but cannot have higher rank than $x$.

\item Let $\hat{S}:= S\cup \{i_{q}\}$. (Note that $i_{q}$ is not 0.) 
\end{itemize}
\item %If $q$ is the lone element of $A$ with leader $x$, proceed to step 3. Otherwise 
If $x$ appears in multiple elements of $A$, choose some $p\neq q$ that has maximal degree in $x$ among all elements of $A$ having leader $x$.  If there are multiple such polynomials, pick any one that has greatest total degree among those with maximal degree in $x$ (again there might be several). Now pseudodivide $p$ by $q$ and let $r$ be the resulting pseudoremainder. Update the equations by omitting $p$ and including $r$; let $\hat{A}:= (A\cup \{r\})\setminus \{p\}$.

Note that $deg_{x}(r)< deg_x(p)$ and $deg_x(q)\leq deg_x(p)$, but we do not guarantee $deg_x(r)<deg_x(q)$ because we use a single step of pseudodivision (Remark \ref{singlestepconv}) and $r$ is not necessarily reduced with respect to $q$. % $deg_{x}(r)<deg_x(q)\leq deg_x(p)$ because $r$ is reduced with respect to $q$. 

Return the union ${\tt Triangulate}(\widetilde{A},S)\, \cup \, {\tt Triangulate}(\hat{A}\cup\{s_{q}\},\hat{S})\, \cup$ \\${\tt Triangulate}(\hat{A},\hat{S}\cup\{s_{q}\})$.  %(Note that $deg_{x}(s_{q})< deg_{x}(q)$.)

(Note: It is convenient to describe these recursive calls in terms of splitting the computation into different branches. Here in step 2 we have split twice. On one branch we included $i_q$ with the equations, left the inequations $S$ unchanged, and called $\tri$ on the updated system $(\widetilde{A},S)$. On the other branch we included $i_q$ with the inequations, pseudodivided $p$ by $q$, and then split again over including $s_{q}$ with the equations or the inequations. This led to ${\tt Triangulate}(\hat{A}\cup\{s_{q}\},\hat{S})$ and ${\tt Triangulate}(\hat{A},\hat{S}\cup\{s_{q}\})$, respectively.)

\item 
If $q$ is the lone element of $A$ with leader $x$, % (recall that $s_q\notin S$ by definition of $x$), % In this situation 
then we consider two cases:
	\begin{enumerate}
		\item  If $s_q\in \hat{S}$, return the union ${\tt Triangulate}(\widetilde{A},S)\, \cup \,{\tt Triangulate}(A,\hat{S})$. %(Omit ${\tt Triangulate}(\widetilde{A},S)$ if we did not split in step 1.).
		\item If $s_q\notin \hat{S}$, then $0<deg_x(s_q)<deg_x(q)$ because if $x$ does not appear in $s_q$, then  $deg_x(q)=1$ and  $s_q=i_q\in \hat{S}$. %Taking step 2 into account, we see that
Pseudodivide $q$ by $s_q$ and let $r$ be the resulting pseudoremainder; note that  $deg_{x}(r)< deg_x(q)$.

Return the union ${\tt Triangulate}(\widetilde{A},S)\, \cup \,{\tt Triangulate}((A\cup \{s_q,r\})\setminus \{q\}, \hat{S})\, \cup \,{\tt Triangulate}(A,\hat{S}\cup\{s_q\})$.
		
(Note: As in step 2, we say that we have split over including $i_q$ with the equations or inequations, and then likewise over $s_q$.) %However, $q$ is the only element of $A$ with leader $x$, so pseudodivision is different here in step 3.)
	\end{enumerate}

%\item If $A$ was already triangular, choose (if there is one) the polynomial $p\in A$ with the highest-ranking leader such that the separant $s_p$ of $p$ is nonconstant and does \emph{not} belong to $S$. For concreteness, suppose this variable is $x_1$ and the polynomial is $p_1$. Split into two branches, one including $s_{p_1}$ with the equations and one not. There are two cases for the branch that adds $s_{p_1}$ to the equations:

%\begin{itemize}
%\item If the $deg_{x_1}(p_1)>1$, then $0<deg_{x_1}(s_{p_1})<deg_{x_1}(p_1)$. Compute ${\tt Triangulate}(A\cup\{s_{p_1}\},S)$. %and take the union of the output with that of the branch that adds $s_{p_1}$ to the inequations.

%\item If $deg_{x_1}(p_1)=1$, then $s_{p_1}$ is also the initial of $p_1$. Compute  ${\tt Triangulate}((A\cup\{s_{p_1}, p_1- s_{p_1}x_1\})\setminus \{p_1\}, S)$, where $p_1- s_{p_1} x_1$ is the tail of $p_1$. Note that $x_1$ does not appear in either $s_{p_1}$ or $p_1- s_{p_1} x_1$.

%\end{itemize}

%For the other branch, compute ${\tt Triangulate}(A,S\cup\{s_{p_1}\})$. Note that $A$ is still triangular but the polynomial $q\in A$ with the highest-ranking leader such that the separant $s_q$ of $q$ is nonconstant and does not belong to $S\cup\{s_{p_1}\}$ has leader $l_q < x_1$.  Return the union of the outputs of the two branches.

\end{enumerate}
\end{itemize}

\begin{remark} \label{redund_rmk}
%7-9-22 good
%[TODO remark on consistency; how to trim inconsistent branches: use sqfr rc and test product of other elts in $S$ at the time; if not in the sat'n, then whole sys (eqns and ineqns) at that point in the branch is consistent.

Though not necessary for correctness of ${\tt Triangulate}$, in practice it is essential to trim inconsistent branches (i.e., pairs $(A=0,S\neq 0)$ that have no solution in the reals) as the algorithm progresses. Discarding such branches does not lose any solutions. Analogously, state-of-the-art versions of Buchberger's algorithm for \Grob bases avoid redundant $S$-polynomial calculations by finding an appropriate subset of the possibilities \cite[Sect. 2.10]{clo1_4}. %As indicated by Theorem \ref{premtest} and Proposition \ref{satsqfr}, if $C$ is a squarefree regular chain and $H_C$ is the set of initials and separants of $C$, we can test membership in $\sat[H_C]{C}$ by pseudodivision. In particular, $C=0,H_C\neq 0$ is consistent and if $s:= \prod{\left(\widetilde{S}_C\setminus H_C \right)}$ has nonzero pseudoremainder with respect to $C$, then $C=0,\widetilde{S}_C \neq 0$ is also consistent. (Here $s$ is the product of all elements of $\widetilde{S}_C\setminus H_C$.) This follows from the Nullstellensatz \ref{nss} because $\sat[H_C]{C}$ is a proper radical ideal by Proposition \ref{satsqfrrad} and nonzero pseudoremainder implies $s\notin \sat[H_C]{C}$. [TODO: make sure all refs to NSS/DNSS stuff are in; in particular, put in something (or cite if it's already here) saying that not in sat'n means doesn't vanish at some soln of eqns $=0$ and ineqns $\neq 0$.] 
A simple optimization of $\tri$ would be to omit branches that add a nonzero constant to the equations or 0 to the inequations. Likewise, we can spot inconsistency by inspection if a nonzero constant multiple of a polynomial in the equations shows up in the inequations or vice versa. In general, though, it may be necessary to compute a regular chain (see Theorem \ref{premtest}) or a \Grob basis of $\sat[S_B]{(B)}$ to determine if $(B=0,S_B\neq 0)$ is solvable in $\bbc$ (in $\bbr$ would require even more, such as real quantifier elimination \cite[Prop. 5.2.2]{bochnak1998real}). These calculations can be expensive, so good judgment is required to pick a strategy for efficiently detecting most of the inconsistent branches. In any case, the worst-case complexity (Theorem \ref{cxtyTri}) is not affected by the lack of consistency checking in ${\tt Triangulate}$ since we must always take into account the possibility of splitting. 
\end{remark}

%[TODO: remark on why the termination pf is so complicated]
We now prove termination and correctness of ${\tt Triangulate}$. The termination argument is somewhat delicate because many different cases can emerge during a run of the algorithm. To ensure termination we desire a well ordering that strictly decreases in all cases we could encounter. This demands multiple criteria for ranking sets of equations and inequations and requires a cleverly crafted ordering of the list at the beginning of Theorem \ref{triangcorr}'s proof.

%We analyze $\rgaexp$ by first proving correctness and describing complexity of {\tt Triangulate}. We then do the same for the whole of $\rgaexp$.

% 6-15-22 good check through step 1 termination
%6-20-22 good check of whole thing; just need to modify treatment of pseudoremainders to indicate just one step; don't need reduced wrt pdivisor.
%7-11-22 good final check

\begin{theorem}[{Termination and correctness of ${\tt Triangulate}$}] \label{triangcorr}
Given finite sets of polynomials $A\subseteq \rx,S\subseteq \rx\setminus\{0\}$, algorithm ${\tt Triangulate}$ terminates and the output  $(A_1, S_1),\ldots, (A_r, S_r)$ satisfies

\[\sqrt{\sat{(A)}}=(\sat[S_1]{(A_1)})\cap \cdots \cap (\sat[S_r]{(A_r)}),\] 

\noindent where $A_1,\ldots, A_r \subseteq \rx, S_1,\ldots, S_r \subseteq \rx\setminus \{0\}$, and $(A_i=0,S_i\neq 0)$ is a regular algebraic system (i.e., $A_i$ is triangular and $S_i$ contains at least the separants of $A_i$).

\end{theorem}

\begin{proof} We first prove termination. This follows from a partial well ordering on pairs $(B,S_B)$ where $B,S_B \subseteq \rx$ are finite sets of polynomials. Given such pairs $(B_1,S_{B_1}),$ $(B_2,S_{B_2})$, we say that $(B_2,S_{B_2})$ has higher rank than $(B_1,S_{B_1})$ and write $(B_1,S_{B_1})<(B_2,S_{B_2})$ if one of the following conditions holds. (Recall that the target variable of pair $(B,S_B)$ is the variable $x$ of highest rank such that either $B$ contains multiple elements having leader $x$ or $B$ has only one element $p$ with leader $x$ and the separant $s_p$ of $p$ does not belong to $S_B$.)
	\begin{enumerate}
		\item The target variable $x$ of $(B_2,S_{B_2})$ has higher rank than the target variable $y$ of $(B_1,S_{B_1})$.
		\item The two pairs have the same target variable $x$, but $x$ appears with strictly larger degree in $B_2$.
		\item The target variable and maximal degree in the target are the same for both pairs, but $B_2$ has strictly more elements with maximal degree in the target.
		\item All the previous quantities are the same for both pairs, but $B_2$ has strictly more elements whose leader is the target.
		\item  All the previous quantities are the same for both pairs, but an element of $B_2$ with minimal nonzero degree in the target variable $x$ has strictly higher degree in $x$ than the corresponding element of $B_1$.
	\end{enumerate}

 More formally, we are using $\{1,2,\ldots, n\}\times \bbn^4$ ordered lexicographically with the standard orders on $\{1,2,\ldots, n\}$ and $\bbn$. The subscript of the target variable belongs to $\{1,2,\ldots,n\}$ and the four copies of $\bbn$ represent the highest degree in the target variable, the number of elements with the highest degree in the target, the number of elements having the target as their leader, and the minimal degree in the target, respectively.

 To show termination it suffices to show that each recursive call in the body of ${\tt Triangulate}$ acts on a pair of strictly lower rank than that of $(A,S)$. We first note that in all cases the target variable can only remain the same or decrease. This is because we never remove anything from the inequations and the only variables that can appear in the elements added to the equations are the original target variable or variables of lower rank. If the target variable remains the same, then the maximal degree in the target must remain the same or decrease because anything we add to the equations has degree in the target that is strictly less than the original maximal degree in that variable.  %Hence in the following analysis we assume that neither the target nor the maximal degree in the target changes unless one of these must decrease in the case at hand. 
Hence we must show that one of the remaining quantities decreases if the quantities that lexicographically precede it stay the same. 

%then either the number of elements of maximal degree in the target decreases, or that number stays the same and the total number of elements with the target as leader decreases, or all of these stay the same and the minimal degree in the target decreases.

	\begin{itemize}
		\item Step 2 of ${\tt Triangulate}$ deals with the case that $A$ has multiple elements with leader $x$. First we produce the pair $(\widetilde{A},S)$, where $\widetilde{A}:= (A\cup \{i_{q}, q- i_{q}x^{deg_{x}(q)}\})\setminus \{q\}$, $x$ is the target variable of $(A,S)$, and $q$ is an element of $A$ having leader $x$ and minimal degree in $x$. To see that $(\widetilde{A},S)<(A,S)$, we observe that
			\begin{itemize}
				\item If $deg_x(q)$ is maximal as well as minimal (i.e., all elements of $A$ that have leader $x$ have the same degree in $x$), then the number of elements of maximal degree in $x$ decreases because we removed $q$, the initial $i_q$ does not contain $x$ and $deg_x(q- i_{q}x^{deg_{x}(q)})< deg_x(q)$.
				\item If $deg_x(q)$ is not maximal, then the number of elements of maximal degree in $x$ stays the same. If $x$ does not appear in $q- i_{q}x^{deg_{x}(q)}$, then 				the number of elements with $x$ as the leader decreases (we replaced $q$ with $q- i_{q}x^{deg_{x}(q)}$ and $x$ does not appear in $i_q$). If $x$ does appear in $q- 				i_{q}x^{deg_{x}(q)}$, then the number of elements with $x$ as the leader stays the same, but the minimal degree in $x$ decreases (because $deg_x(q- i_{q}					x^{deg_{x}(q)})< deg_x(q))$. 
			\end{itemize}
  Step 2 also produces the pairs $(\hat{A}\cup\{s_{q}\},\hat{S})$ and $(\hat{A},\hat{S}\cup\{s_{q}\})$, where $\hat{S}:= S\cup \{i_{q}\}$ and $\hat{A}:= (A\cup \{r\})\setminus \{p\}$ for $p\neq q$ having maximal degree in $x$ and $r$ the pseudoremainder upon pseudodividing $p$ by $q$.  The number of elements having maximal degree in $x$ decreases because $r$ and $s_q$ have strictly lower degree in $x$ than $p$ and we remove $p$.
		\item  Step 3 of ${\tt Triangulate}$ deals with the case that $q$ is the lone element of $A$ with leader $x$.  
			\begin{itemize}
				\item If $s_q\in \hat{S}$, then step 3 produces the pair $(A,\hat{S})$. (We already have shown that $(\widetilde{A},S)<(A,S)$, so here we just look at $(A,\hat{S})$). In this situation the target variable must decrease because $q$ is the only element of $A$ with leader $x$ and now $s_q$ belongs to the inequations (so by definition $x$ is no longer a target variable).
				\item  If $s_q\notin \hat{S}$, then step 3 produces the pairs $((A\cup \{s_q,r\})\setminus \{q\}, \hat{S})$ and $(A,\hat{S}\cup\{s_q\})$, where $x$ appears in $s_q$ and $r$ is the pseudoremainder upon pseudodividing $q$ by $s_q$. In the first instance, $((A\cup \{s_q,r\})\setminus \{q\}, \hat{S}) < (A,S)$ because either the target variable decreases or it remains the same and the maximal degree in $x$ decreases (since $deg_{x}(r)<deg_x(q),deg_x(s_q)< deg_x(q)$, and $q$ was the lone element of $A$ with leader $x$). In the second, $(A,\hat{S}\cup\{s_q\})<(A,S)$ because the target variable must decrease ($q$ was the only element of $A$ with leader $x$ and now $s_q$ belongs to the inequations).

			\end{itemize}	
	\end{itemize}

This proves termination. We now establish the properties of the output systems $(A_i=0,S_i \neq 0)$.

Each $A_i$ is triangular and $S_i$ contains the separants of all elements of $A_i$ (i.e.,  $(A_i=0,S_i\neq 0)$ is a regular algebraic system) because this is the base case that returns an explicit answer instead of a recursive call. (Note that $0\notin S_i$ because we only ever add $i_q$ or $s_q$ to the inequations, and these are never the zero polynomial.)

Lastly, we show that the ideals $\sat[S_i]{(A_i)}$ decompose $\sqrt{\sat{(A)}}$ as claimed. Observe that $\sqrt{\sat[S_1]{(A_1)}} \cap \cdots \cap \sqrt{\sat[S_r]{(A_r)}} = (\sat[S_1]{(A_1)}) \cap \cdots \cap (\sat[S_r]{(A_r)})$ because each $\sat[H_i]{(A_i)}$ is radical by Lazard's lemma (Lemma \ref{lazard}). Each operation during a run of  ${\tt Triangulate}$ converts a system $(B=0,S_B\neq 0)$ into another system $(B_1=0, S_{B_1}=0)$ or two systems $(B_1=0, S_{B_1}=0), (B_2=0, S_{B_2}=0)$. By induction on the number of splitting and pseudodivision operations during a run of ${\tt Triangulate}$, it suffices to confirm at each step that $\sqrt{\sat[S_B]{(B)}}=\sqrt{\sat[S_{B_1}]{(B_1)}}$ or $\sqrt{\sat[S_B]{(B)}}=\sqrt{\sat[S_{B_1}]{(B_1)}}\cap \sqrt{\sat[S_{B_2}]{(B_2)}}$. %There are two main cases: 

\begin{enumerate}

\item (Splitting over an initial or separant)  Steps 2 and 3 both split the computation into branches, one where the initial $i_q$ or separant $s_q$ of a chosen polynomial $q\in B$ is included with the equations $B$ and one where it is included with the inequations $S_B$. Thus by Theorem \ref{algsplitrad}, we have, respectively,  \[\sqrt{\sat{(B)}}=\sqrt{\sat{(B,i_q)}}\cap \sqrt{\sat[(S_B\cup\{i_q\})]{(B)}}\] or \[\sqrt{\sat{(B)}}=\sqrt{\sat{(B,s_q)}}\cap \sqrt{\sat[(S_B\cup\{s_q\})]{(B)}}.\]

\item (Pseudodivision)
Steps 2 and 3b of ${\tt Triangulate}$ replace some $g\in B$ with its pseudoremainder $r_g$ upon pseudodividing by some $h$, where $h\in B$ and the initial $i_h\in S_B$. (In step 2, $g$ is $p$ and $h$ is $q$ while in step 3b, $g$ is $q$ and $h$ is $s_q$. In the latter case, the initial $i_{s_q}$ of $s_q$ is a nonzero constant multiple of $i_q$ (which \emph{does} belong to the inequations $\hat{S}$) because separants are partial derivatives with respect to the leader.) %For our purposes this is equivalent to $i_{s_q}$ itself belonging to the inequations; replacing an inequation with a nonzero constant multiple of itself doesn't change the solutions.) %That is because the next paragraph uses Theorem \ref{algsatrad}, which relies only on the solutions to a system of equations and inequations) 

We must show
 \begin{equation} \label{premrad} \sqrt{\sat[S_B]{(B)}}=\sqrt{\sat[S_B]{((B\cup \{r_g\} )\setminus \{g\})}}. \end{equation}
\noindent By Proposition \ref{pdivsat} (which makes explicit the relationship between a pseudoremainder and the original polynomials) there exist $\alpha\in \rx$ and a product $\widetilde{i}$ of factors of $i_h$ such that $(\widetilde{i})g - \alpha h = r_g$. Thus for any point $\mathbf{a}\in \bbc^n$ we see that $\mathbf{a}$ causes both $g$ and $h$ to vanish but not $i_h$ if and only if $\mathbf{a}$ causes $h$ and $r_g$ to vanish but not $i_h$. (If $i_h(\mathbf{a})\neq 0$, then $\widetilde{i}(\mathbf{a})\neq 0$ because $\widetilde{i}$ is a product of factors of $i_h$.)  Symbolically,

\begin{align}\label{premgeom} \vsc{B}\setminus \vsc{\Pi S_B} &= \vsc{B\cup \{r_g\} \setminus \{g\} }\setminus \vsc{\Pi S_B}.\end{align}

(Similar reasoning applies to Equation \ref{diffcons} in the proof of Theorem \ref{difftriang}.) Then Equation \ref{premrad} follows by Theorem \ref{algsatrad} (Hilbert's Nichtnullstellensatz, which links radicals of saturation ideals to solutions of systems of equations and inequations).

 \phantomsection
 \label{tristep1pf}
 Technically we don't perform pseudodivision when converting $(A\cup \{i_q\},S)$ to $(\widetilde{A},S)$, but the outcome is analogous. (As noted in step 1 on p. \pageref{trinote}, the solutions of $(\widetilde{A}=0,S\neq 0)$ are the same as those of $(A\cup\{i_q\}=0,S\neq 0)$. See also the proof of Lemma \ref{difftriang} on p. \pageref{rgatristep1pf}.) Hence
 \[\sqrt{\sat{(A\cup\{i_q\})}}=\sqrt{\sat{((A\cup \{i_{q}, q- i_{q}x^{deg_{x}(q)}\})\setminus \{q\})}}=\sqrt{\sat{(\widetilde{A})}}.\]
 
\end{enumerate}

This proves the claimed properties of ${\tt Triangulate}$. 
\end{proof}

We can now describe the full $\rgaexp$ algorithm. $\rgaexp$ eliminates ODEs in explicit form using differential pseudodivision; moreover, no new ODEs are introduced. Importantly, the eliminated ODEs still belong to all differential saturation ideals that appear during the computation. This is ensured by a technical assumption on $\xpf$ that is given in the input description below. (The correctness proof for $\rgaexp$, Theorem \ref{modrgacorr}, shows how the assumption accomplishes this.)  %ensures that this holds for recursive calls to $\rgaexp$. 
Without the assumption, the operations of $\rgaexp$ could produce differential saturation ideals that do not satisfy the hypotheses of the regular invariant theorem (Theorem \ref{satinvar}). In turn, that would break the connection between $\rgaexp$ and algebraic invariants that we establish in Theorem \ref{radlddecomp}.

%7-11-22 good final
%6-20-22 good check
\medskip
\phantomsection \label{rgaalgo}
\noindent $\rgaexp$:
\smallskip

	\begin{itemize}
    \item \emph{Informal summary}: $\rgaexp$ alternates between applying ${\tt Triangulate}$ to the nondifferential elements of the current system and using differential pseudodivision to eliminate ODEs in explicit form. The computation recursively splits into multiple branches whose leaves form the desired differential radical decomposition. The two interleaved operations assure, respectively, that the nondifferential elements form a regular algebraic system and that the output is partially reduced as required for regular differential systems.
    
    \item \emph{Detailed description}:
    \begin{itemize}
	\item Choose the \emph{orderly} differential ranking such that $x_n> x_{n-1} > \cdots > x_1$. (We will use the fact that the ranking is orderly to prove correctness, but the particular sequence of variables is not essential.) %(The particular choice and other details, such as being orderly or not, are not essential; others would work as well.) %because we never allow derivatives to exceed order 1.)

	 \item Input: Let $\xpef$ be a polynomial vector field. Let $A=\{\xtpf, p_1,p_2,$ $\ldots,p_m\}$ be a finite set where $\xtpf$ is a subset of $\xpf$ and the $p_1,$ $ \ldots, p_m \in \rx$ are nondifferential polynomials. Also, we have a finite set of nondifferential polynomials  $S\subseteq \rx\setminus \{0\}$. As usual in ${\tt Triangulate}$, the elements of $A$ correspond to equations and the elements of $S$ to inequations.  Lastly, the ``technical assumption" mentioned above: we assume that for each member $x'-f(\bx)$ of $\xpf$ we either have $x'-f(\bx) \in A$ (equivalently, $x'-f(\bx)\in\xtpf$) or  
 \[s(x'-f(\bx)) =  \beta(\bx) +\sum_j \gamma_{j}(\bx)(x_j'-f_j(\bx))\] 
 \noindent for some $s\in S, \beta(\bx)\in \big[A \cap \rx\big] \unlhd \, \dkx[\bbr], \gamma_{j}(\bx)\in \rx, x_j<x$, and $x_j'-f_j(\bx)\in A$. (This is a specific case of $\xpf \in \sat{[A]}$. The key point is that for all $x'-f(\bx)\in \xpf$, either $x'-f(\bx)$ explicitly belongs to the equations $A$ or $x'=f(\bx)$ is implied by $A$ and the inequations $S$.) %$s\in S, % \sum_i \alpha_{i}(\bx)p_i'+ \sum_i\beta_{i}(\bx)p_i +\sum_j \gamma_{j}(\bx)(x_j'-f_j(\bx))\]  \alpha_{i}(\bx),\beta_{i}(\bx), \gamma_{j}\in \rx, x_j<x$, and $x_j'-f_j(\bx)\in A$. (This is a specific case of $\xpf \in \sat{[A]}$.)
	 
 	\item Output: Pairs $(A_1,S_1),\ldots, (A_r,S_r)$ such that $A_1, \ldots, A_r \subseteq \dkx[\bbr], S_1,\ldots, S_r \subseteq \rx\setminus\{0\}$ are finite sets and
%\smallskip
 \[\sqrt{\sat{[A]}}=(\sat[S_1]{[A_1]})\cap \cdots \cap (\sat[S_r]{[A_r]}),\] 
 \noindent where each  $(A_i=0,S_i\neq 0)$ is a regular differential system. That is, $A_i$ is partially reduced (no element contains a proper derivative of the leader of another element) and triangular (no two elements have the same leader), and $S_i$ is partially reduced with respect to $A_i$ and contains at least the separants of $A_i$. Moreover, the elements of $A_i$ containing proper derivatives form a subset of $\xtpf$, and for each $x'-f(\bx)\in \xpf$ and $1\leq i \leq r$ we either have $x'-f(\bx)\in A_i$ or 
 \[\hat{s}(x'-f(\bx)) = \hat{\beta}(\bx) +\sum_j \hat{\gamma}_{j}(\bx)(x_j'-f_j(\bx))\]
for some $\hat{s}\in S_i, \hat{\beta}(\bx)\in \big[A_i\cap \rx\big]\unlhd \dkx[\bbr], \hat{\gamma}_j(\bx)\in\rx, x_j<x$, and $x_j'-f_j(\bx)\in A_i$. (The proof of Theorem \ref{modrgacorr} shows that the technical assumption on $\xpf$ at the start is preserved by each operation of $\rgaexp$, and hence holds of the output.)

\end{itemize}

\begin{enumerate}

\item If $A$ is already partially reduced and triangular, and separant $s_q$ belongs to $S$ for every $q\in A$, then return $(A,S)$. %(If an initial or separant is a constant (necessarily nonzero), then we always consider it to belong to $S$. Moreover, we also consider nonzero constant multiples of elements of $S$ to belong to $S$.)

\item If $A$ is not triangular or separant $s_q\notin S$ for some nondifferential $q\in A$, compute ${\tt Triangulate}(A\cap\rx,S)=\tri(\{p_1,\ldots, p_m\},S)$. (Recall the $p_i$ are the nondifferential elements of $A$. Also, an element of the form $x_{j}'-f_{j}(\bx)$  has separant 1, so as usual we assume such a separant belongs to the inequations $S$.) This yields a finite collection of regular algebraic systems involving only nondifferential polynomials. Return the union $\bigcup_{(B,S_B)} \rgaexp(B\, \cup \, \xtpf, S_B)$ over all $(B,S_B)$ returned by {\tt Triangulate}$(\{p_1,$ $\ldots,p_m\},S)$. (That is, for each such $(B,S_B)$, re-include with the equations $B$ the elements of $A$ that contain proper derivatives and call $\rgaexp$ on the resulting system. This produces a finite collection of differential systems; take the union of all these collections.)
%after applying ${\tt Triangulate}$ to the nondifferential elements, 

\item If $A$ is triangular and $s_q\in S$ for all $q\in A$ but $A$ is not partially reduced, choose the highest-ranking leader $x_{j}$ such that $x_{j}'-f_{j}(\bx)\in A$ and there exists $q\in A\cap \rx$ having leader $x_j$. (Such  $x_{j}'-f_{j}(\bx)$ and $q$ exist because $A$ is not partially reduced.  Because $A$ is triangular and $x_j$ is the highest-ranking variable with the desired property, $x_{j}'-f_{j}(\bx)$ and $q$ are unique.) %As in {\tt Triangulate}, we call $x$ the \emph{target variable} for the pair $(A,S)$.  
For concision we write $x_{j}'-f_{j}(\bx)$ as $x'-f(\bx)$.

Recalling that $s_q\in S$ by assumption in the current case, %differentially pseudoreduce $x'-f(\bx)$ by $q$. 
differentiate $q$ and pseudodivide $x'-f(\bx)$ by $q'$. %Even if the resulting pseudoremainder is not reduced with respect to $q$, do not continue pseudodividing by $q$ at this stage. The algebraic part of the reduction will be handled by {\tt Triangulate} in step 3.) //will be reduced wrt q' bec of deg 1 in x for q' and x'-f(x)
Let $\widetilde{r}$ be the resulting pseudoremainder; note that  $x'$ is not present in $\widetilde{r}$ because the degree of $x'$ is 1 in $x'-f(\bx)$ and $q'$. (In this case, the pseudoremainder is reduced with respect to $q'$ after a single step of pseudodivision.) However, derivatives of other variables in $q$ may be present. Let $y$ be such a variable; there is some $g(\bx)\in \rx$ such that $y'-g(\bx)$ is a member of $\xpf$. %by assumption $y'-g(\bx)\in \sat{[A]}$ for some $g(\bx)\in \rx$. 
Replace every instance of $y'$ in $\widetilde{r}$ with $g(\bx)$. (We justify this move in Theorem \ref{modrgacorr}, the correctness proof for $\rgaexp$.) Doing so for each proper derivative present in $\widetilde{r}$ produces a nondifferential polynomial $r$. (Abusing terminology, we also refer to $r$ as a pseudoremainder.) Update the equations by omitting $x'-f(\bx)$ and including $r$; let $\hat{A}:= (A\cup \{r\})\setminus\{x'-f(\bx)\}$. Note that the elements of $\hat{A}$ having proper derivatives form the set $\xtpf\setminus\{x'-f(\bx)\}$ because $r$ is nondifferential. 

Now compute ${\tt Triangulate}(\hat{A}\cap\rx,S)$. Return the union $\bigcup_{(B,S_B)} \rgaexp(B\cup(\xtpf\setminus\{x'-f(\bx)\}), S_B)$ over all $(B,S_B)$ returned by ${\tt Triangulate}(\hat{A}\cap\rx,S)$.  

\end{enumerate}
\end{itemize}

The next result is an important ingredient of the correctness proof for $\rgaexp$ (Theorem \ref{modrgacorr}). The lemma performs the critical job of uniting the algebraic part (calls to ${\tt Triangulate}$) and differential part (differential saturation ideals) of $\rgaexp$.

\begin{lemma} \label{difftriang}
%Let $A=\{x_{i_1}'-f_{i_1}(\bx),$ $\ldots, x_{i_s}'-f_{i_s}(\bx), p_1,p_2,\ldots,p_m\}$ where $x_{i_j}\neq x_{i_k}$ for $j\neq k$, the $x_{i_j}'-f_{i_j}(\bx)$ are ODEs in explicit form, $\{x_{i_1},\ldots, x_{i_s}\} \subseteq \{x_1,\ldots, x_n\}$, and the $p_1, \ldots, p_m \in \rx$ are nondifferential polynomials. 

%7-11-22 good final

Let $\xpef$ be a polynomial vector field. Let $A=\{\xtpf, p_1,p_2,$ $\ldots,p_m\}$ be a finite set where $\xtpf$ is a subset of the differential polynomials $\xpf$ and the $p_1, \ldots, p_m \in \rx$ are nondifferential polynomials. Let $S\subseteq \rx\setminus \{0\}$ be a finite set of nondifferential polynomials. Let $(B_1,S_1), \ldots,(B_N, S_N)$ be the output of ${\tt Triangulate}(A\cap\rx,S)$; by correctness of $\tri$ (Theorem \ref{triangcorr}) we know that
\[\sqrt{\sat{(A\cap\rx)}}=\sqrt{\sat[S_1]{(B_1)}}\cap \cdots \cap \sqrt{\sat[S_N]{(B_N)}} = (\sat[S_1]{(B_1)})\cap \cdots \cap (\sat[S_N]{(B_N)}).\] 

Then we have

\begin{align*}
 \sqrt{\sat{[A]}}&=\sqrt{\sat{[A\cap\rx, \xtpf]}} &=\sqrt{\sat[S_1]{[B_1, \xtpf]}}\cap \cdots \\ 
& &\phantom{=} \cap \sqrt{\sat[S_N]{[B_N, \xtpf]}}.
\end{align*}
\end{lemma}

\begin{proof}
The key observation is that the two operations of ${\tt Triangulate}$ (namely, splitting and nondifferential pseudodivision) preserve the differential solutions of the augmented system that includes the ODEs with the equations. (See Remark \ref{difftriangrmk} for a metamathematical issue concerning this approach.) Moreover, these operations only involve nondifferential polynomials.  For instance, if $q\in A\cap\rx$ we have 
 \[\sqrt{\sat{(A\cap\rx)}}=\sqrt{\sat{(A\cap\rx,i_q)}}\cap \sqrt{\sat[(S\cup\{i_q\})]{(A\cap\rx)}}\]

\noindent after a splitting step because the solutions of $(A\cap \rx=0,S\neq 0)$ in $\bbc$ are the union of the solutions of $(A\cap\rx=0,i_q=0, S\neq 0)$ and $(A\cap\rx=0, S\cup\{i_q\}\neq 0)$ in $\bbc$ (see Theorem \ref{algsplitrad}). %[TODO: maybe expand on and cite existential closedness version of NSS]. 
The analogous differential decomposition holds if we include the ODEs and consider \emph{differential} solutions in any differential extension field.
By definition of $A$ and by Theorem \ref{splitrad} we have
\begin{equation}
\label{diffspliteq}
\begin{aligned}
\sqrt{\sat{[A]}}&=\sqrt{\sat{[A\cap\rx, \xtpf]}}\\
&=\sqrt{\sat{[A\cap\rx, i_q, \xtpf]}}\, \cap \\
&\phantom{=} \cap\sqrt{\sat[(S\cup\{i_q\})]{[A\cap\rx, \xtpf]}}.
\end{aligned}
\end{equation}

If $g\in A\cap\rx$ and $r_g$ is the pseudoremainder from pseudodividing $g$ by some other element of $A\cap\rx$ whose initial belongs to $S$, then much like Equation \ref{premgeom} in the proof of Theorem \ref{triangcorr} we have 
\begin{equation} \label{diffcons}
\begin{aligned} 
\vs[\delta]{A}\setminus \vs[\delta]{\Pi S} &=\vs[\delta]{(A\cap\rx)\cup\{\xtpf\}}\setminus \vs[\delta]{\Pi S} \\
 &=\vs[\delta]{(((A\cap\rx) \cup \{r_g\} )\setminus \{g\} )\cup\xtpf }\setminus \vs[\delta]{\Pi S}.
\end{aligned}
\end{equation}

\noindent Theorem \ref{diffsatrad} then gives
 \[ \sqrt{\sat[S]{[A]}}=\sqrt{\sat[S]{[(((A\cap\rx) \cup \{r_g\} )\setminus \{g\} ),\xtpf ]}}.\]
 
 \noindent Similarly, for $q\in A\cap\rx$ we have
  \[\sqrt{\sat{[A\cup\{i_q\}]}}=\sqrt{\sat{[(((A\cap \rx)\cup \{i_{q}, q- i_{q}x^{deg_{x}(q)}\})\setminus \{q\}),\xtpf]}}.\]
\phantomsection
\label{rgatristep1pf}
\noindent (See %step 1, p. \pageref{tristep1}, of $\tri$ and 
the final part of the proof of Theorem \ref{triangcorr}, p. \pageref{tristep1pf}.) 

Thus re-inserting the ODEs $\xtpf$ after a single  splitting or pseudodivision step preserves the differential decomposition. The full decomposition follows by induction on the number of splitting and pseudodivision operations during a run of ${\tt Triangulate}$. (In other words, since the ODEs $\xtpf$ are not involved in the operations of ${\tt Triangulate}$, it is equivalent to perform the entire run of ${\tt Triangulate}$ on the nondifferential elements and then re-introduce the $\xtpf$.)
\end{proof}

The system $(B_i=0,S_i\neq 0)$ is a regular algebraic system and so $\sqrt{\sat[S_i]{(B_i)}}=\sat[S_i]{(B_i)}$ by Lazard's lemma (Lemma \ref{lazard}). However, $(B_i=0,\xtpf=0, S_i\neq 0)$ might not be partially reduced so we need the radicals in the differential decomposition given by Lemma \ref{difftriang}.

\begin{remark} \label{difftriangrmk}
%good final 7-11-22
%good check 6-20-22
%[TODO: why Lemma \ref{difftriang} is important; connection to translation between diff and alg; couldn't guarantee unless we only use splitting and diff pred'n] 
Though the proof of Lemma \ref{difftriang} is straightforward, a subtle issue lurks nearby. We initially hoped to use more efficient radical ideal decomposition methods (e.g., that of Sz\'{a}nt\'{o} \cite{{Szanto97},{szanto1999computation}}, which we discuss in Section \ref{rgabds}, p. \pageref{szanto}) in place of ${\tt Triangulate}$. This provides an algebraic decomposition of the desired form

\begin{equation}\label{algdecomp}
\begin{aligned} \sqrt{\sat{(A\cap\rx)}}&=\sqrt{\sat[S_1]{(B_1)}}\cap \cdots \cap \sqrt{\sat[S_N]{(B_N)}}\\ &= (\sat[S_1]{(B_1)})\cap \cdots \cap (\sat[S_N]{(B_N)}), \end{aligned}
\end{equation}

\noindent but requires methods beyond simple splitting and (nondifferential) pseudodivision. We were unable to prove %the conclusion of 
%differential decomposition guaranteed by 
Lemma \ref{difftriang} (and hence our main results Theorems \ref{modrgacorr} and \ref{radlddecomp}) using the algebraic decomposition in Equation \ref{algdecomp} without the fact that splitting and pseudodivision preserve differential solutions. This illustrates a major challenge of differential algebra: even when differential polynomials have a simple form, drawing conclusions about differential ideals from restrictions to algebraic ideals is nontrivial. It also explains why we stopped with regular differential systems in $\rgaexp$ instead of continuing with the calculations needed to get regular differential chains.

\end{remark}

\begin{theorem}[{Termination and correctness of $\rgaexp$}] \label{modrgacorr}

%7-13-22 final check good
%6-20-22 partial careful check to see if 1 step psdiv vs full psdiv mattered; didn't

%work: ~5-13-22 453
 Let $\xpef$ be a polynomial vector field. Let $A=\{\xtpf, p_1,p_2,\ldots,p_m\}$ be a finite set where $\xtpf$ is a subset of the differential polynomials $\xpf$  %$A=\{x_{i_1}'-f_{i_1}(\bx),\ldots, x_{i_s}'-f_{i_s}(\bx), p_1,p_2,\ldots,p_m\}$ be a finite set where $x_{i_j}\neq x_{i_k}$ for $j\neq k$, the $x_{i_j}'-f_{i_j}(\bx)$ are a subset of the differential polynomials $\xpf$, 
and the $p_1, \ldots, p_m \in \rx$ are nondifferential polynomials. Let $S\subseteq \rx\setminus \{0\}$ be a finite set of nondifferential polynomials. Lastly, assume that for each member $x'-f(\bx)$ of $\xpf$ we either have $x'-f(\bx) \in A$ (equivalently, $x'-f(\bx)\in\xtpf$) or  
 \[s(x'-f(\bx)) =  \beta(\bx) +\sum_j \gamma_{j}(\bx)(x_j'-f_j(\bx))\] 
 \noindent for some $s\in S, \beta(\bx)\in \big[ A\cap \rx \big ] \unlhd \, \dkx[\bbr], \gamma_{j}(\bx)\in \rx, x_j<x$, and $x_j'-f_j(\bx)\in A$.
  %\in \sat{[A\setminus \{x_j'-f_j(\bx)\in \xpf \mid x_j> x_i\}]}$. %(i.e., any member of $\xpf$ not in $A$ belongs to the differential saturation of $A$ by $S$).

%Let $A=\{x_{i_1}'-f_{i_1}(\bx),$ $\ldots, x_{i_s}'-f_{i_s}(\bx), p_1,p_2,\ldots,p_m\}$ where $x_{i_j}\neq x_{i_k}$ for $j\neq k$, the $x_{i_j}'-f_{i_j}(\bx)$ are ODEs in explicit form, $\{x_{i_1},\ldots, x_{i_s}\} \subseteq \{x_1,\ldots, x_n\}$, and the $p_1, \ldots, p_m \in \rx$ are nondifferential polynomials. Let $S\subseteq \rx\setminus \{0\}$ be a finite set of nondifferential polynomials. 
Then the algorithm $\rgaexp$ terminates on input $(A,S)$ and the output $(A_1,S_1),$ $\ldots,(A_r,S_r)$ satisfies
 \[\sqrt{\sat{[A]}}=(\sat[S_1]{[A_1]})\cap \cdots \cap (\sat[S_r]{[A_r]}),\] 

\noindent where  $A_1, \ldots, A_r \subseteq \dkx[\bbr], S_1,\ldots, S_r \subseteq \rx\setminus\{0\}$ are finite sets such that each $(A_i=0,S_i\neq 0)$ is a regular differential system. %obey the following:
%	\begin{enumerate}
%		\item each $A_i$ is partially reduced and triangular,
%		\item the elements of $A_i$ containing proper derivatives form a subset of $\xtpf$, and
%		\item each $S_i$ is partially reduced with respect to $A_i$ and contains at least the separants of $A_i$.
%	\end{enumerate}
Moreover, the elements of $A_i$ containing proper derivatives form a subset of $\xtpf$, and for each $x'-f(\bx)\in \xpf$ and $1\leq i \leq r$ we either have $x'-f(\bx)\in A_i$ or 
\[\hat{s}(x'-f(\bx)) = \hat{\beta}(\bx) +\sum_j \hat{\gamma}_{j}(\bx)(x_j'-f_j(\bx))\] for some $\hat{s}\in S_i, \hat{\beta}(\bx)\in \big[A_i\cap \rx\big]\unlhd \dkx[\bbr], \hat{\gamma}_j(\bx)\in \rx, x_j<x$, and $x_j'-f_j(\bx)\in A_i$. (That is, the ``technical assumption" on $\xpf$ that we made of the input also holds of the output.)
\end{theorem}

\begin{proof}

%[TODO: put termination in the correctness pf. Be careful because the minimal degree in the target variable could go up in Triang, but the overall number of polys with the target goes down. So my wording in the next line needs to be modified, but we still have termination.] Note that once $x_j'-f_j(\bx)$ is eliminated, $x_j$ is never again differentiated or used as the target variable. Observe also that on each branch at steps 1 and 4 we either decrease the degree in the target variable $x_j$, reduce the number of polynomials in the system having minimal nonzero degree in $x_j$, or advance to a new variable. (Regarding steps 2 and 3, once we reach step 2 we're assured of going to the next variable in step 4, so the algorithm still makes progress.) It follows that the process terminates.
We first argue that $\rgaexp$ terminates. On any branch of the computation, once the input $(A,S)$ has the property that $A$ is triangular and $s_q\in S$ for all $q\in A$, this property continues to hold for each new call to $\rgaexp$ along that branch. This is so because the property activates step 3 and step 3 calls ${\tt Triangulate}$ (which terminates by Theorem \ref{triangcorr}) right before recursively calling $\rgaexp$. Hence every subsequent call to $\rgaexp$ either terminates immediately or proceeds to step 3 and eliminates one of the ODEs $x'-f(\bx)$. Note that at most $n$ of the original equations contain proper derivatives and eliminated proper derivatives are never re-introduced to the set of equations (${\tt Triangulate}$ only involves nondifferential polynomials and the pseudoremainder $r$ introduced by step 3 of $\rgaexp$ is nondifferential). It follows that, at most, a branch of $\rgaexp$ calls ${\tt Triangulate}$ and $\rgaexp$ $n+1$ times (the first time in step 2 to make the equations triangular and put separants in the inequations; thereafter only step 3 applies) before terminating.

%Each $(A_i=0,S_i\neq 0)$ is a regular differential system since $A_i$ and $S_i$ have the form required by Lemma \ref{augsqfr}.  (In particular, squarefree regular chains are triangular and step 2 of $\rgaexp$ only admits $x_j'-f_j(\bx)$ to $A_i$ if $x_j$ is not a leader of some nondifferential polynomial in the system. Step 3 ensures that $H_i\subseteq S_i$.)

We justify the decomposition by showing that after each intermediate step we have a decomposition with the claimed form and properties, except possibly partial reducedness and being radical. However, these will both hold of the final output; see the final paragraph of the proof.   During a run of $\rgaexp$, new differential systems are produced either by calling ${\tt Triangulate}$ on the nondifferential equations (followed by re-inserting the ODEs) or performing differential pseudodivision. By Lemma \ref{difftriang}, a call to ${\tt Triangulate}$, followed by re-inserting the ODEs, decomposes $\sqrt{\sat{[A]}}$ into a finite intersection of radicals of differential ideals of the right form (except possibly for partial reducedness)
\[\sqrt{\sat{[A]}}=\sqrt{\sat[S_1]{[B_1, \xtpf]}}\cap \cdots \cap \sqrt{\sat[S_N]{[B_N, \xtpf]}},\]

\noindent where $(B_1,S_1),\ldots, (S_N,B_N)$ is the output of ${\tt Triangulate}(A\cap\rx, S)$. In particular, the equations with proper derivatives are exactly $\xtpf$. We must confirm that each $(B_l\cup \xtpf,S_l)$ preserves the technical assumption on the elements of $\xpf$. If $x'-f(\bx)\in \xpf$ belongs to $A$, then $x'-f(\bx)$ remains in $B_l\cup \xtpf$ because the proper derivatives are unchanged. Otherwise, we have

\begin{equation} \label{triangode}
s(x'-f(\bx)) =  \beta(\bx) +\sum_j \gamma_{j}(\bx)(x_j'-f_j(\bx))
\end{equation}
 
 \noindent for some $s\in S, \beta(\bx)\in \big[A\cap \rx\big]=[p_1,\ldots,p_m], \gamma_{j}(\bx)\in \rx, x_j<x$, and $x_j'-f_j(\bx)\in A$. These same $x_j'-f_j(\bx)$ are now in $B_l\cup \xtpf$. By the algebraic decomposition 
 
 \[\sqrt{\sat{(A\cap\rx)}}=(\sat[S_1]{(B_1)})\cap \cdots \cap (\sat[S_N]{(B_N)})\] 
 
 \noindent guaranteed by correctness of $\tri$ (Theorem \ref{triangcorr}), we know that each $p_i\in A\cap\rx$ belongs to each $\sat[S_l]{(B_l)}\subseteq \rx$. It follows that $p_i^{(k)}\in \sat[S_l]{[B_l]}\subseteq \dkx[\bbr]$ for all $k\in \bbn$ and hence $[p_1,\ldots,p_m]\subseteq \sat[S_l]{[B_l]}$. %//the p_i aren't necessary here; could just say (A\cap\rx)\subseteq (B_l):S_l^infty % (and in particular we only need first derivatives of elements of $B_l$ to represent $p_i'$).[TODO: clarify/more detail?] 
Multiplying both sides of equation \ref{triangode} by an appropriate element of $S_l$, we thus obtain
 
 \begin{equation} 
\hat{s}(x'-f(\bx)) = \hat{\beta}(\bx) +\sum_j \hat{\gamma}_{j}(\bx)(x_j'-f_j(\bx))  
\end{equation}
 
 \noindent for some $\hat{s}\in S_l$ (recall that $S\subseteq S_l$ because on a given branch no inequations are ever removed), $\hat{\beta}(\bx)\in [B_l],$ and $\hat{\gamma}_{j}(\bx)\in\rx$. This shows that $(B_l\cup \xtpf,S_l)$ preserves the technical assumption %inductive hypothesis 
 on the elements of $\xpf$.

 %Also, by assumption $\xpf\in \sat{[A]}\subseteq \sqrt{\sat{[A]}}$, so $\xpf$ must belong to the radical differential ideals whose intersection equals $\sqrt{\sat{[A]}}$. 

We now check the differential pseudodivision from step 3 of $\rgaexp$. %The argument is similar to the pseudodivision case in the correctness proof for {\tt Triangulate} (Theorem \ref{triangcorr}). //uses eqn from prem and proves equality of two saturations, but not geometric like the the former pf and mixes in a lot of proving the "assumption on x'-f(x)"; omit
Step 3 of $\rgaexp$ produces a pseudoremainder $\widetilde{r}$ by pseudodividing some $x'-f(\bx)\in A$ by $q'$, where the separant $s_q\in S$ and $q\in A\cap \rx$ has leader $x$. Recall that $s_q$ is the initial of $q'$ and belongs to $\rx$. Because the initial of $x'-f(\bx)$ is 1 and the degree of $x'$ is 1 in both $x'-f(\bx)$ and $q'$, we have $s_q(x'-f(\bx)) -q' = \widetilde{r}$. (This is an instance of Proposition \ref{pdivsat}.) Step 3 of $\rgaexp$ goes on to substitute nondifferential polynomials (using the relations $\xpef$) for the remaining proper derivatives (which necessarily have lesser rank than $x'$) in $\widetilde{r}$. This produces $r\in \rx$ such that $\widetilde{r} = \sum_l \zeta_l(\bx) (x_l'-f_l(\bx)) + r$ for some $\zeta_l(\bx)\in \rx$ and $x_l<x$. (As in the proof of Lemma \ref{subslie}, we replace $x_l'$ in $\widetilde{r}$ with $(x_l' -f_l(\bx))+f_l(\bx)$. Then distribute to get an equation of the claimed form for $\widetilde{r}$.) Hence $s_q(x'-f(\bx)) =q'+ \sum_l \zeta_l(\bx) (x_l'-f_l(\bx)) + r$. Let $\hat{A}:= (A\cup \{r\})\setminus\{x'-f(\bx)\}$; note that $\hat{A}$ contains all nondifferential polynomials of $A$ and of the others only omits $x'-f(\bx)$. Step 3 of $\rgaexp$ then proceeds using the system $(\hat{A}, S)$. If $\hat{A}$ is no longer triangular or $S$ does not contain $s_r$, the next call to $\tri$ will restore those properties.

We show that $(\hat{A}, S)$ preserves the technical assumption on the elements of $\xpf$ and that $\sat{[A]}=\sat{[\hat{A}]}$. Though $x'-f(\bx)\in A$, by definition $x'-f(\bx)\notin \hat{A}$. However, by assumption each $x_l'-f_l(\bx)$ in the equation 

\begin{equation}\label{premeq}
s_q(x'-f(\bx)) =q'+ \sum_l \zeta_l(\bx) (x_l'-f_l(\bx)) + r
\end{equation}

\noindent either belongs to $A$ (and hence to $\hat{A}$ since $x_l'-f_l(\bx)\neq x'-f(\bx)$) or can be written as 

 \begin{equation}\label{indodeeq}
 s_l(x_l'-f_l(\bx)) = \beta_l(\bx) +\sum_j \gamma_{j,l}(\bx)(x_j'-f_j(\bx))
 \end{equation}

 \noindent for some $s_l\in S,\beta_{l}(\bx)\in \big[A\cap\rx \big]\subseteq \big[\hat{A}\cap\rx\big], \gamma_{j,l}(\bx)\in \rx, x_j<x_l$, and $x_j'-f_j(\bx)\in A$. Since $x_j<x_l<x$, we have $x_j'-f_j(\bx)\neq x'-f(\bx)$ and so $x_j'-f_j(\bx)\in \hat{A}$. Multiplying both sides of \ref{premeq} by appropriate elements of $S$, using \ref{indodeeq}, and using the fact that $q, r$ belong to $\hat{A}\cap\rx$, we see that the technical assumption holds of $x'-f(\bx)$ with respect to $(\hat{A},S)$.  We now consider all $y'-g(\bx)\in \xpf$ such that $y'-g(\bx)\neq x'-f(\bx)$. If $y'-g(\bx)\in A$, then $y'-g(\bx)$ still belongs to $\hat{A}$. If $y'-g(\bx)\notin A$, by assumption
 
 \begin{equation}\label{notAeq}
 s(y'-g(\bx)) = \beta(\bx)+\sum_j \gamma_{j}(\bx)(x_j'-f_j(\bx))
 \end{equation}

 \noindent for some $s\in S,\beta(\bx)\in \big[A\cap\rx\big], \gamma_{j}(\bx)\in \rx, x_j<y$, and $x_j'-f_j(\bx)\in A$. Here one of the $x_j'-f_j(\bx)$ might be $x'-f(\bx)$, but we showed above that $x'-f(\bx)$ has the right form with respect to $(\hat{A},S)$. Hence we can multiply both sides of \ref{notAeq} by appropriate elements of $S$, substitute for $x'-f(\bx)$, and conclude that $y'-g(\bx)\in \sat{[\hat{A}]}$ has the form required by the technical assumption.  %(This is a specific case of $\xpf \in \sat{[A]}$.)

To prove that $\sat{[A]}=\sat{[\hat{A}]}$, it suffices to show that $x'-f(\bx)\in \sat{[\hat{A}]}$ and $r\in \sat{[A]}$. Both memberships follow from equations \ref{premeq} and \ref{indodeeq} (recall $q\in A\cap\rx, r\in \hat{A}\cap\rx$, and $x'-f(\bx)\in A$).
Thus differential pseudodivision replaces a system with an equivalent one whose corresponding differential saturation ideal is the same (in particular, the decomposition does not change). 

%Step 2 of $\rgaexp$ replaces some $x'-f(\bx)$ with its pseudoremainder $r$ upon pseudodividing by $q'$, where $q\in A\cap \rx$ and the separant $s_q\in S$. By Proposition \ref{pdivsat} there exist $\alpha\in \rx$ and a product $\widetilde{s}$ of factors of $s_q$ such that $(\widetilde{s})(x'-f(\bx)) - \alpha q' = r$. (Recall that $s_q$ is the initial of $q'$.) 

%Since for any point $\mathbf{a}$ the equation $q(\mathbf{a})=0$ implies $q'(\mathbf{a})=0$ (see Remark \ref{diffpolyrmk}), we see that $\mathbf{a}$ causes both $x'-f(\bx)$ and $q$ to vanish but not $s_q$ if and only if $\mathbf{a}$ causes $q$ and $\widetilde{r}$ to vanish but not $s_q$.  Symbolically,
%\[\vs[\delta]{A}\setminus \vs[\delta]{\Pi S} = \vs[\delta]{A\cup \{\widetilde{r}\} \setminus \{x'-f(\bx)\} }\setminus \vs[\delta]{\Pi S}.\]
%Hence by Corollary \ref{diffsatrad} we have \[ \sqrt{\sat[S]{[A]}}=\sqrt{\sat[S]{[A\cup \{\widetilde{r}\} \setminus \{x'-f(\bx)\}]}}.\] Step 2 of $\rgaexp$ goes on to replace every remaining proper derivative $y'$ in $\widetilde{r}$ by $g(\bx)$ (where $y'-g(\bx)$ is a member of $\xpf$). We claim that 
%\[\sqrt{\sat[S]{[A\cup \{\widetilde{r}\} \setminus \{x'-f(\bx)\}]}}=\sqrt{\sat[S]{[A\cup \{r\} \setminus \{x'-f(\bx)\}]}}\]

%\noindent and that $\xpf \in \sat[S]{[A\cup \{r\} \setminus \{x'-f(\bx)\}]}$. 
%[TODO: I think I can just go algebraically (i.e., show the rad diff ideals are identical directly) to prove preservation when substituting using things from the sat]

Lastly, we explain why each output $(A_i=0,S_i\neq 0)$ is a regular differential system. (Each $\sat[H_i]{[A_i]}$ is then radical by Lazard's lemma (Lemma \ref{lazard}), which is why $\sqrt{\sat[S_1]{[A_1]}}\cap \cdots \cap \sqrt{\sat[S_r]{[A_r]}} = (\sat[S_1]{[A_1]})\cap \cdots \cap (\sat[S_r]{[A_r]})$.) Both $A_i$ and $S_i$ are finite because only finitely many elements can enter at any step. Initials and separants are never identically zero, so 0 was not added to the inequations during the computation (i.e., $0\notin S_i$). Termination only occurs when $A_i$ is partially reduced, so that property is assured. The nondifferential polynomials in $A_i$ form a triangular set by the action of ${\tt Triangulate}$. The remaining differential polynomials form a subset of $\xtpf$. The orderly ranking ensures that an element $x_j'-f_j(\bx)$ of $\xtpf$ has leader $x_j'$, so the entire set $A_i$ (i.e., the union of the nondifferential elements of $A_i$ and a subset of the original ODEs) remains triangular.  The separant of any element of $\xtpf$ is $1$ and the separants of nondifferential elements of $A_i$ belong to $S_i$ by ${\tt Triangulate}$. Since $S_i\subseteq \rx$ has no proper derivatives, it is partially reduced with respect to $A_i$. Thus  $(A_i=0,S_i\neq 0)$ is a regular differential system. This completes the proof.

%and the fact that no new proper derivatives appear in $A_i$ ensure that $A_i$, $S_i$ have the form required by Lemma \ref{augsqfr}. (This is where we use the assumption from $\rgaexp$ that the differential ranking is orderly.) 

%and (In particular, a branch only stabilizes when $A_i$ is triangular; otherwise subroutine {\tt Triangulate} would have continued. Step 2 of $\rgaexp$ only admits $x_j'-f_j(\bx)$ to $A_i$ if $x_j$ is not a leader of some nondifferential polynomial in the system. Step [TODO] of {\tt Triangulate} ensures that $H_i\subseteq S_i$.)

\end{proof}

We need the following theorem to interpret the output of $\rgaexp$.  The result is intuitively natural but slightly fussy to prove. It connects differential ideals to invariant ideals of polynomial vector fields: in particular, in the presence of an explicit system $\xpf$, the nondifferential polynomials in a differential ideal form the smallest invariant ideal containing the nondifferential polynomials in the original set that generated the differential ideal.

\begin{theorem} \label{lddecomp}
%1-22-22 fairly careful
%1-24-22 another careful: good
%6-20-22 quite careful; just fixed notational
%7-13-22 :good final

Let $A=\{\xpf, p_1, p_2,\ldots, p_m\}$ with $p_1,\ldots, p_m \in \rx$. Then 
$(\ldfid{A\cap\rx})=(\ldfid{A_{\rx}})=(\ldfid{p_1,\ldots, p_m})=[A]\cap \rx $, where $\mathbf{F}$ is the polynomial vector field $\mathbf{x'}=\mathbf{f}(\bx)$.
\end{theorem}

\begin{proof}
($\subseteq$): Immediate from Lemma \ref{subslie}. 

\noindent ($\supseteq$) Let $q\in [A]\cap \rx$. Then for some $l\geq 1$, we have a representation

\[
q= \sum_{1\leq i \leq n, 1\leq k \leq l} g_{i,k}(\bx) (x_i^{(k)} -f_i^{(k-1)}(\bx)) + \sum_{1\leq j \leq m,1\leq k\leq l} h_{j,k}(\bx) p_j^{(k)} + \sum_{1\leq j \leq m} h_j(\bx) p_j,
\]

\medskip
\noindent where $g_{i,k}, h_{j,k}, h_j \in \bbr\{\bx\}$. Proceeding as in the proof of Lemma \ref{subslie} (i.e., replacing $x_i'$ with $x_i' -f_i(\bx)+f_i(\bx)$), for $k\geq 1$ we may replace $p_j^{(k)}$ with $\mathcal{L}_{\mathbf{F}}^{(k)}(p_j)$ plus a sum of multiples of the various $x_i'-f_i(\bx),\ldots, x_i^{(k)}(\bx)-f_i^{(k-1)}(\bx)$. %//plug in LD into derivs; lin combo of x'-f(x) and LD's; take derivs of those for higher derivs.
Regrouping and renaming the coefficient polynomials $g_{i,k}$, etc., as necessary, we may thus assume $q$ has the form

\[
q= \sum_{1\leq i \leq n, 1\leq k \leq l} g_{i,k}(\bx) (x_i^{(k)} -f_i^{(k-1)}(\bx)) + \sum_{1\leq j \leq m,1\leq k\leq l} h_{j,k}(\bx)\mathcal{L}_{\mathbf{F}}^{(k)}(p_j) + \sum_{1\leq j \leq m} h_j(\bx) p_j.
\]

\medskip
\noindent Because $q$ and the various $f_i(\bx) ,\mathcal{L}_{\mathbf{F}}^{(k)}(p_j),$ and $p_j$ belong to $\rx$, they are not altered by formally substituting $f_i^{(r-1)}(\bx)$ for $x_i^{(r)}$. With this substitution the sum 
\[\sum_{1\leq i \leq n, 1\leq k \leq l} g_{i,k}(\bx) (x_i^{(k)} -f_i^{(k-1)}(\bx))\]  

\medskip
\noindent vanishes and the coefficients of the $p_j$ and their Lie derivatives become nondifferential polynomials, proving $q\in (\ldfid{p_1,\ldots, p_m})$.

\end{proof}

Our work on $\rgaexp$ culminates with Theorem \ref{radlddecomp}, which interprets the output as a canonical algebraic invariant of a polynomial vector field $\mathbf{F}$. In particular, parts 1 and 3 connect differential ideals and the maximal invariant contained in $\vsr{A\cap\rx}$. Part 3 also indicates the structure of this invariant in terms of smaller invariant sets. These are the major contributions of Theorem \ref{radlddecomp}. Part 2 is implicit in \cite{GhorbalSP14} because there the authors obtain $(\ldfid{A_{\rx}})$, from which membership in the radical can be tested. (See our remarks on \cite{GhorbalSP14} in Section \ref{gencheck}. We discuss radical ideal membership in Section \ref{check} and the two paragraphs preceding it.) However, our mechanism is different (regular differential systems instead of \Grob bases) and shows that to test membership in $\sqrt{(\ldfid{A_{\rx}})}$ it is unnecessary to first find generators of $(\ldfid{A_{\rx}})$.

%Computability of $\sqrt{(\ldfid{A_{\rx}})}$ (part 2) %and  $\vsr{[A]\cap \rx}$ (part 3) are 
%is implicit in \cite{GhorbalSP14}, but our mechanism is different and  (However, \cite{GhorbalSP14} obtains $(\ldfid{A_{\rx}})$ and not just its radical.) %[TODO more detail somewhere; clarify implicit/explicit]

%%LEAVE RIGHT BEFORE radlddecomp

\begin{theorem}\label{radlddecomp}
%7-13-22 Good final
Let $A=\{\xpf, p_1, p_2, \ldots, p_m\}$ with $p_1,\ldots, p_m \in \rx$ and let $\mathbf{F}$ be the polynomial vector field $\mathbf{x'}=\mathbf{f}(\bx)$. If $(A_1,S_1),\ldots, (A_r,S_r)$ are the regular differential systems returned by $\rgaexp$ given $(A,\emptyset)$ as input, then 
\begin{enumerate}
\item %There exist regular differential systems $(A_1,S_1),\ldots, (A_r,S_r)$ that are explicit with nondifferential inequations such that 
\begin{align*}
\sqrt{(\ldfid{A_{\rx}})}   &= \sqrt{[A]}\cap \rx\\&= (\sat[S_1]{(A_1\cap \rx)})\cap \cdots \cap  (\sat[S_r]{(A_r\cap \rx)}).
\end{align*}
\item  Membership in $\sqrt{(\ldfid{A_{\rx}})}$ is decidable.  %$\rgaexp$ gives a decision procedure for membership in $\sqrt{(\ldfid{A})}$.

\item  $\vsr{[A]\cap \rx}$ is the largest invariant with respect to $\mathbf{F}$ that is contained in $\vsr{A\cap\rx}$. Moreover,  $\vsr{[A]\cap \rx}$  is a finite union of invariants of the form  $\vsr{\sat[S_i]{(A_i\cap \rx)}}$.
\end{enumerate}

\end{theorem}

\begin{proof}

\begin{enumerate}
\item The first equation follows from Theorem \ref{lddecomp} and Lemma \ref{radrestr} (2). For the second, we have
\begin{align*}
\sqrt{[A]}\cap\rx  &=((\sat[S_1]{[A_1]})\cap \cdots \cap  (\sat[S_r]{[A_r]}))\cap\rx && \text{\hspace{.5cm}(Theorem \ref{modrgacorr})}\\
 &=  ((\sat[S_1]{[A_1]})\cap \rx)\cap \cdots \cap  ((\sat[S_r]{[A_r]})\cap \rx) && \text{\hspace{.5cm}}\\
 &= ((\sat[S_1]{(A_1)})\cap \rx)\cap \cdots \cap  ((\sat[S_r]{(A_r)})\cap \rx)&& \text{\hspace{.5cm}(Theorem \ref{rosenfeld})}\\
&=  (\sat[S_1]{(A_1\cap \rx)})\cap \cdots \cap  (\sat[S_r]{(A_r\cap \rx)}).&& \text{\hspace{.5cm}(Lemma  \ref{exist})}
  \end{align*}
 
 \noindent Theorem \ref{modrgacorr}, correctness of $\rgaexp$, applies because $\xpf\in A$. (That is, the ``technical assumption" in the statement of Theorem \ref{modrgacorr}  holds true at the start because all of the  $x_i' - f_i(\bx)$ belong to $A$ before we run $\rgaexp$.) Lemma \ref{exist} applies because the $(A_i,S_i)$ are explicit regular differential systems by Theorem \ref{modrgacorr}.

  \medskip

\item $\rgaexp$ computes the regular differential systems $(A_i,S_i)$. Membership in $\sat[S_i]{(A_i\cap\rx)}$ is decidable using \Grob bases or regular chains (Section \ref{keyregres}). Membership in a finite intersection of ideals is decidable using \Grob bases \cite[p. 194]{clo1_4}. %By Theorem ??, we may assume the $A_j\cap \rx$ in part 1 are regular sets [TODO: check; do we need to say anything about the $S_j$?]. Thus by Theorem \ref{premtest} a polynomial $p\in \rx$ belongs to  $\sqrt{(\ldfid{A})}$ if and only if the pseudoremainder of $p$ with respect to $A_j\cap \rx$ is 0 for each $1\leq j\leq r$.

\item By part 1 of this theorem, Lemma \ref{capcup}, and part 3b of Proposition \ref{radprop} we have
\begin{align*}
\vsr{\ldfid{A_{\rx}}} &=\vsr{[A]\cap \rx}\\ &= \vsr{\sat[S_1]{(A_1\cap \rx)}}\cup \cdots \cup  \vsr{\sat[S_r]{(A_r\cap \rx)}} \end{align*}

The claims then follow from  Lemma \ref{maxinvar} and Theorem \ref{satinvar}. (Theorem \ref{satinvar} applies because Theorem \ref{modrgacorr} implies that $\xpf\in \sat[S_i]{[A_i]}$ and $(A_i,S_i)$ is an explicit regular differential system.)
\end{enumerate}
\vspace{-.1cm}
\end{proof}

Note that our Theorems \ref{modrgacorr},\ref{radlddecomp} have more specialized hypotheses than the RGA theorem of Boulier et al. (Theorem \ref{rgadecomp}) and do not guarantee that all elements of $\sat[S_i]{[A_i]}$ are differentially pseudoreduced to 0 by $A_i$. However, Theorem \ref{rgadecomp} does not identify the dynamical meaning of the nondifferential polynomials in the output. Moreover, our Theorem \ref{radlddecomp} still allows for testing membership in  $\sqrt{[A]}\cap \rx$. The greater specificity of Theorem \ref{radlddecomp} may also lead to better computational complexity; see Remark \ref{eqeffrmk}.

\subsection{Algebraic Invariants from RGA Using Parameters} \label{lorenzex2}
We now explain how we obtained the Lorenz system invariant in Section \ref{lorenzex1}. Recall that the ODEs in question are $x' = y - x, y'=2x -y-xz,  z'=xy-z$. Our goal is to identify at least one nontrivial algebraic invariant of this system using the algorithms and theorems from Section \ref{rgaexp}. Consider the parameterized polynomial constraint $g(x,y,z)=ax^2+by^2 + cz^2=0$, where $a,b,c$ are unknown constants. (We call $g$ a \emph{template}.) %The form of a potential invariant $g$ is constrained by the monomials that appear. 
To settle on this choice, we first restricted to polynomials of degree 2 and looked at possibilities containing $g$ but also having additional monomials. %By trial-and-error we found that 
The parameters for these additional monomials ended up being 0 for interesting invariants, so here we only present $g$ as the template. % MDGB.mws 3-31-21
The input system is $A:=(x'-y+x, y'-2x +y+xz,  z'-xy+z, g, a', b', c')$, where we interpret the polynomials as equations; the inequation set is empty. The terms $a',b',c'$ express that $a,b,c$ are constant and do not depend on time like $x,y,z$. Computing RGA in Maple (using the {$\tt RosenfeldGroebner$} command) on the system $A$ with an orderly ranking $x>y>z>a>b>c$ yields five regular differential chains: % (labeled in ascending order): %(*Maple ref?*):

%\begin{align*} &(1)\, [x' + x - y, xz - 2x + y + y', -xy + z + z', a, b, c],(2)\, [xz - 2x + y + y', -xy + z + z', c',\\ & 2x^2 - y^2 - z^2, a + 2c, b - c],(3)\, [z' + z, a', b', x, y, c], (4)\,[b', c', x - y, y^2 - 1, z - 1, a + b + c],\\&(5)\, [a', b', c', x, y, z]. 
%\end{align*}

\begin{enumerate}
\item $x' + x - y, xz - 2x + y + y', -xy + z + z', a, b, c$,
\item $xz - 2x + y + y', -xy + z + z', c', 2x^2 - y^2 - z^2, a + 2c, b - c$,
\item $z' + z, a', b', x, y, c$,
\item $b', c', x - y, y^2 - 1, z - 1, a + b + c$,
\item $a', b', c', x, y, z$. 
\end{enumerate}

%Maple uses brackets around the differential polynomials to denote a chain; this is suggestive of (but not necessarily equal to) the differential ideal generated by the polynomials. 
Even though we used a black-box commercial version of RGA for the calculation, Section \ref{lorenzex1} shows that the hypotheses of Theorems \ref{satinvar} and \ref{invarcor} hold of the output. (Specifically, chain 2; the others give invariants by inspection. See the next paragraph.) By default, Maple suppresses the inequations associated to each chain; however, the user can print them with the {$\tt Inequations$} command. When implemented, $\rgaexp$ should give equivalent output (but explicitly showing the inequations).  

%We consider the invariant algebraic sets defined by nondifferential polynomials in the chains. In particular, 

In the output we interpret $a,b,c$ as parameters and the invariants as algebraic sets in $\bbr^3$ defined by polynomials in $\bbr[x,y,z]$. Chain 1 defines the trivial invariant set $\mathbb{R}^3$ because $a=b=c=0$ causes $g$ to vanish at every point in  $\mathbb{R}^3$. Chain 2 requires $a=-2c$ and $b=c$, which is the relationship of coefficients in $g_{a=2, b=-1,c=-1}=2x^2 - y^2 - z^2$. In other words, chain 2 defines the two-dimensional surface $2x^2 - y^2 - z^2=0$, the invariant we confirmed in Section \ref{lorenzex1}. (Any nonzero scalar multiple of $(2,-1,-1)$, or equivalently any nonzero solution of $a=-2c,b=c$, would define the same surface.) %, indicating the parameters $a,b,c$ should be understood as determining a point in projective space \cite{clo1}.) 
This appears to be the ``most generic chain" or ``general solution" produced by RGA (p. \pageref{genericchain}; roughly, no additional equations hold beyond those implied by the input). %(see Hubert99 \cite{todo} for precise definitions of these notions).
Chain 3 requires $x,y$ to constantly be zero; the final term of $g$ vanishes because $c=0$. Thus chain 3 defines the $z$-axis, a one-dimensional curve. Chain 4 is even simpler, defining the (zero-dimensional) points  $(1,1,1)$ and $(-1,-1,1)$. Lastly, chain 5 defines the origin, $(0,0,0)$. %As suggested above, 

Theorem \ref{radlddecomp} implies that the union of these sets contains \emph{every} algebraic invariant of the form $ax^2+by^2+cz^2=0$ in $\bbr^3$ for the vector field $x' = y - x, y'= 2x -y-xz,  z'=xy- z$. (Moreover, invariants that are proper subsets of $\bbr^3$ are contained in the union of the invariants given by the final four chains.) This is essentially a completeness result describing the sense in which $\rgaexp$ finds \emph{all} invariants having a given parameterized form. Moreover, the output invariant is structured, being composed of smaller invariants of various dimensions (some of which are proper subsets of others). The example suggests that a generic chain corresponds to an invariant that has maximal dimension among the proper invariants. Containment of the invariants produced by $\rgaexp$ is controlled by containment of the corresponding ideals and there may be redundancy or proper containment (as is typical for these kinds of decomposition algorithms). Smaller invariants that are proper subsets of larger ones may still be interesting, though, so we do not pursue minimal decompositions. However, see Section \ref{conc} for a possible future optimization.
%This is a completeness result: given a parameterized polynomial template, $\rgaexp$ (or RGA, if the hypotheses of Theorem \ref{satinvar} hold of the output) finds all invariants having that form. //a little strong....maybe// Moreover, the output invariant contains smaller invariants of various dimensions. The example suggests that a generic chain corresponds to an invariant that has maximal dimension among the proper invariants. See Section \ref{conc}. %By Theorem \ref{radlddecomp} we know that $\vsr{[\Sigma]\cap \rx}$ contains every algebraic invariant of $x' = y - x, y'= 2x -y-xz,  z'=xy- z, a'=b'=c'=0$ that is contained in $\vsr{g}$ (viewed as a subset of $\bbr^6$; for the moment we view $a,b,c$ as variables, just like $x,y,z$). But such an invariant simply corresponds to a ... too awkward to write out for what it's worth. But formally we're going up from R^3 to R^6 and then back down; homogeneity of the constraints on a,b,c show that no matter what a,b,c we chose, the result satisfies one of the five alg sets we found; I guess if g not homog, then would have a family of invar surfaces, not just an isolated one.

%[TODO: RGA gives decomposition of the maximal invar; algorithmically can produce irred comps and test containment]
%??invariance and irred components? known?

%[TODO: Checking: other char set-like]; checking invar ideal w/o requiring ideal membership of original/GB: sat'ns/characterizable/regular set; usu (?) are radical ideals (]

\subsection{Bounds on $\rgaexp$} \label{rgabds}
We now find upper bounds (Theorems \ref{cxtyRGA},\ref{expRGA}) on the complexity of $\rgaexp$. The maximal degree of polynomials in the output is a common measure of complexity for algorithms that operate on polynomials; in this paper we use the term \emph{degree complexity}. Like many authors we consider degree complexity as a function of the maximal degree $d$ of the input polynomials and the number of variables $n$. Algorithms involving polynomial ideals frequently have degree complexity that is singly or doubly exponential in $n$ and polynomial in $d$ (e.g., $O(d^{2^n}))$ \cite{DickensteinFGS91}. Thus it is reasonable to ignore matters like time complexity of addition and multiplication, the number of polynomials in the input, etc., that affect the output in only polynomial or singly exponential fashion. %bec of num of monomials for given degree is O(dn)^n, and that is a coarse/excessive estimate.

We start by analyzing ${\tt Triangulate}$. This involves a recursively defined function $T(d,n)$ that bounds the degree complexity. It is not surprising that the \emph{Fibonacci numbers} appear in the definition of $T(d,n)$ because a single step of pseudodivision produces a pseudoremainder whose degree is at most the sum of the degrees of the dividend and divisor (see step 2 of ${\tt DiffPseudoDiv}$ on p. \pageref{pdivdeg}). The Fibonacci sequence also appears in other works on RGA's complexity, e.g., Cor. 29 of \cite{GustavsonOP18}.

\begin{notation}
We write $F_j$ for the $j$-th \emph{Fibonacci number}, defined recursively by $F_0:=0, F_1:=1, F_{j}:=F_{j-1}+F_{j-2}$ for $2\leq j$.
\end{notation}

\begin{proposition}[{Binet's formula \cite[p. 92]{wallis2010introduction}}] \label{Binet}
\[F_j= \frac{\phi_1^j - \phi_2^j}{\sqrt{5}},\]
where $\phi_1:= \frac{1+\sqrt{5}}{2}$ and $\phi_2:= \frac{1-\sqrt{5}}{2}$.
\end{proposition}

It is somewhat tricky to pin down the degree complexity of ${\tt Triangulate}$ for the same reason that termination (Theorem \ref{triangcorr}) was delicate to establish: depending on the inputs, ${\tt Triangulate}$ must accommodate many different patterns of degrees throughout the computation. This includes the difficulty that multiple intermediate polynomials can have the same degree, so the degree of a given output is not necessarily tied to how many pseudodivision steps the calculation required.  To handle this, we induct on the number of times that the minimal and maximal degrees in the target variable change. % as well as the maximal total degree at any stage.

\begin{theorem}[{Degree complexity of ${\tt Triangulate}$}]\label{cxtyTri} 
%7-15-22 good final
%6-20-22 good check; now just reconcile w/ Tri. The pdiv steps used here aren't full pdiv; just cut at least one deg from the pdividend,  not necess the pdivisor. 
Given finite sets of polynomials $A\subseteq \rx,S\subseteq \rx\setminus\{0\}$ in $n=|\mathbf{x}|$ variables such that each polynomial has total degree at most $d$, algorithm ${\tt Triangulate}$ returns polynomials that have total degree at most $T(d,n)$, where $T(d,k)$ is defined recursively by $T(d,0):= d$, $T(d, k):= (F_{2\cdot T(d,k-1) +1})\cdot T(d,k-1)$ for $1 \leq k$. %$T(d,n)_0:= d$, $T(d, n)_{k+1}:= T(d,n)_k\cdot F_{2\cdot T(d,n)_k +1}$ for $0\leq k < n$, and $T(d,n):=T(d,n)_n$.
\end{theorem}
\begin{proof}
Throughout a run of ${\tt Triangulate}$, the inequations are updated with initials and separants of updated equations. Hence it suffices to bound the total degrees of the equations. Equations result from including initials or separants, which doesn't increase the degree, or from a single pseudodivision step. %which at most adds the degrees of the dividend and divisor. 
Hence to find an upper bound we may assume that each new equation results from pseudodivision and adds the degrees of the two polynomials involved. 

For $0\leq k \leq n$ we claim that $T(d,k)$ is a bound on the total degree of any equation after eliminating $k$ target variables, which implies the theorem statement. (That is, $T(d,k)$ is a bound after $\tri$ ensures that each of those $k$ target variables is the leader in at most one equation and the separant the corresponding polynomial belongs to the inequations.) %Given the claim, it follows that $T(d,n)$ is an upper bound on the total degrees of polynomials returned by ${\tt Triangulate}$.) 
Initially, total degrees are at most $T(d,0)=d$ because we haven't done anything yet (no variables eliminated). We prove that $T(d,1)=F_{2 d+1}d$ is a bound on the total degree after eliminating one target variable. This is sufficient because the process is identical going from $k-1$ to $k$. In particular, if $T(d,k-1)$ bounds total degrees after we have eliminated $k-1$ variables, the same argument (just replacing $d$ with $T(d,k-1)$) shows that $T(T(d,k-1),1)=(F_{2\cdot T(d,k-1)+1})\cdot T(d,k-1)$ is a bound after eliminating $k$ variables. The latter expression is $T(d,k)$ by definition. %(in particular, the new bound at the end of a cycle is determined by the bound at the beginning of the cycle, so the same argument shows that $T(d,2)=T(T(d,1),1)=(F_{2\cdot T(d,1)+1})\cdot T(d,1)$, and so on).

We now show that $T(d,1)=F_{2 d+1}d$ is a bound on the total degree after eliminating one target variable. Let $x$ be the first target variable. We introduce the following terminology to assist with the proof. We refer to the polynomial $q$ chosen in step 1 of ${\tt Triangulate}$ as the \emph{minimal pseudodivisor} for the next pseudodivision step. (Recall that by definition $q$ has minimal degree in $x$ among the equations with leader $x$ and also has minimal total degree among equations with that minimal degree in $x$.) We refer to the polynomial $p$ chosen in step 2 of ${\tt Triangulate}$ as the \emph{maximal pseudodividend} ($p$ has maximal total degree among all equations of maximal degree in $x$). Suppose we have performed $i$ pseudodivision steps so far. The key values are the current minimal nonzero degree in $x$ of any equation (call it $x_{\min,i}$) and the current maximal degree in $x$ of any equation (call it $x_{\max,i}$). Then the current minimal pseudodivisor has degree $x_{\min,i}$ in $x$ and the maximal pseudodividend has degree  $x_{\max,i}$ in $x$. Let $t_{\max,i}$ be the current maximal total degree of any equation (not necessarily having $x$ as its leader). For instance, $x_{\min,0}\leq x_{\max,0}\leq t_{\max,0}\leq d$ (the initial bound before any pseudodivisions). %If $x_{min,0}< x_{max,0}$, then 

As long as elements of degree $x_{\min,0}$ in $x$ are used to pseudodivide elements of degree $x_{\max,0}$ in $x$, we get pseudoremainders of total degree at most $2d$. (In other words, as long as $x_{\min}$ and $x_{\max}$ do not change.) This is because any polynomials not present initially are pseudoremainders that have degree in $x$ strictly less than $x_{\max,0}$. Hence they have not been used as pseudodividends yet. Also, any polynomial with degree $x_{\min,0}$ in $x$ that becomes the minimal pseudodivisor in place of the original must have total degree no greater than that of the original minimal pseudodivisor (which was at most $d$). Otherwise it would not be minimal. Hence as soon as $x_{\min,i+1}<x_{\min,i}$ or $x_{\max,i+1}<x_{\max,i}$ for the first time, we still have $t_{\max,i+1}\leq 2d$. % and $t_{max,j}$ can only exceed $2d$ if $j> i+1$. 
Note that $2d=1d+1d=F_1 d + F_2 d = (F_1+F_2)d = F_3d$. Note also that if exactly one of  $x_{\min,i+1}<x_{\min,i}$ or $x_{\max,i+1}<x_{\max,i}$ occurs, then at least one of the subsequent minimal pseudodivisor and maximal pseudodividend still has total degree at most $d=1d=F_2d$. This is because the subsequent minimal pseudodivisor has total degree at most $d$ if $x_{\min,i+1}=x_{\min,i}$. If $x_{\max,i+1}=x_{\max,i}$, then one of the original equations of maximal degree in $x$ is still present and the subsequent maximal pseudodividend has total degree at most $d$. If both $x_{\min,i+1}<x_{\min,i}$ and $x_{\max,i+1}<x_{\max,i}$ occur, then both the new minimal pseudodivisor and maximal pseudodividend have degree at most $2d=F_3d$. %while either $x_{min}$ decreases once or $x_{max}$ decreases once (but not both).
 %If $x_{min}=x_{max}$, then the immediately $x_{min}$ goes down and the total degree again increases to at most $2d$. 
%We could put subscripts to indicate how many times we have pseudodivided to get the current values of $x_{min},x_{max},$ and $t_{max}$, but 
From now on we omit the numerical subscripts from $x_{\min,i},x_{\max,i},$ and $t_{\max,i}$ because the important quantity is \emph{how many times} $t_{\max}$ \emph{can increase}, not how many pseudodivision steps are performed in all. The amount that $t_{\max}$ can change in any individual step is bounded (e.g., earlier in this paragraph we saw that $t_{\max}$ is at most $2d$ after the first change in $x_{\min}$ or $x_{\max}$), so bounding the \emph{number} of changes allows us to bound the final value of $t_{\max}$.

In particular, the crucial observation is that $t_{\max}$ cannot increase more times than the combined number of decreases in $x_{\min}$ and $x_{\max}$, plus 1. Each of $x_{\min},x_{\max}$ can decrease at most $d-1$ times and keep a nonzero value. So after decreasing a combined number of $2d-2$ times, at most $x_{\min}=x_{\max}=1$. Then each subsequent pseudodivision eliminates $x$, so the pseudoremainders are not used as long as $x$ is still the target variable. Hence at this stage there can be no more changes to $x_{\min}$ and $x_{\max}$ (they will both still be 1 when only one polynomial with leader $x$ remains) and $t_{\max}$ can only increase one more time (due to pseudodivision involving elements having degree 1 in $x$), making for $2d-1$ increases at most.  

As alluded to immediately before the theorem statement, we induct on the combined number of decreases in $x_{\min}$ and $x_{\max}$ to show that $t_{\max}$ never exceeds $F_{2d+1} d=T(d,1)$. In particular, we assert that after $j$ total decreases in $x_{\min}$ and $x_{\max}$ (with $1\leq j\leq 2d-2$), we satisfy two conditions: 1. we have $t_{\max}\leq F_{j+2}d$ and 2. at least one of the current minimal pseudodivisor and maximal pseudodividend has total degree at most $F_{j+1}d$. The paragraph before last establishes the base case $j=1$. %//could use j=0 as base case (assumps hold trivially) and then do inductive assuming j and prove j+1. That would only write the argument once. Could do, but I think it's subtle enough to bear repeating.
The reasoning is analogous for the inductive case: suppose after $j-1$ decreases we have $t_{\max}\leq F_{(j-1)+2}d=F_{j+1}d$ and at least one of the current minimal pseudodivisor and maximal pseudodividend has total degree at most $F_{(j-1)+1}d=F_{j}d$. %As explained in the base case, the inductive hypothesis is preserved as long as no further decreases occur. 
As in the base case, we perform pseudodivision until we decrease either one or both of $x_{\min}$ and $x_{\max}$ for a total of $j$ or $j+1$ decreases. Since by assumption either the pseudodivisor or pseudodividend has total degree at most $F_{j}d$ and neither has total degree greater than $F_{j+1}d$, the pseudoremainder has total degree at most $F_{j}d+F_{j+1}d= F_{j+2}d<F_{(j+1)+2}d$. This proves the assertion about $t_{\max}$. If $x_{\min}$ stays the same, then the subsequent minimal pseudodivisor has total degree at most $F_{j+1}d$. If $x_{\max}$ stays the same, then the subsequent maximal pseudodivisor has total degree at most $F_{j+1}d$. If both $x_{\min}$ and $x_{\max}$ decrease, the condition on minimal pseudodivisors and maximal pseudodividends continues to hold: we now have $j+1$ decreases and both the new minimal pseudodivisor and new maximal pseudodividend have total degree at most $F_{j+2}d=F_{(j+1)+1}d$. This proves the inductive case. 

Hence when $j=2d-2$ (necessarily forcing $x_{\min}=x_{\max}=1$), we have a current minimal pseudodivisor of total degree at most $F_{2d-2+1}d=F_{2d-1}d$ and current maximal pseudodividend of total degree at most $F_{2d-2+2}d=F_{2d}d$. The final pseudodivision steps at most add $F_{2d-1}d$ and $F_{2d}d$ to yield the claimed value $F_{2d+1} d=T(d,1)$. This completes the proof.  \end{proof}

%[TODO: write out simple $F_7, d=3$ example?]

We now find an explicit function that bounds the recursively defined function $T(d,n)$. We start with a couple of straightforward inequalities.
 
 \begin{lemma}[{Appendix}]\label{fibbd}
 For all $j\in \bbn$, the $j$-th Fibonacci number $F_j$ satisfies $F_j< 2^j$.
 \end{lemma}
 %7-15-22 good final
 %\begin{proof}
 %\begin{align*}
%F_{j} \,&= \, \frac{\phi_1^{j} - \phi_2^{j}}{\sqrt{5}} \,= \, \frac{\left(\frac{1+\sqrt{5}}{2}\right)^{j} - \left(\frac{1-\sqrt{5}}{2}\right)^{j}}{\sqrt{5}}
%&& \hspace{.5cm} \text{(Binet's formula, Proposition \ref{Binet})} \\
%&< \frac{2^j +1}{\sqrt{5}} && \hspace{.5cm} \text{(} \phi_1<2, |\phi_2| <1\text{)}\\
%&< 2^j. && \hspace{.5cm} \text{(}  \sqrt{5}>2, 1\leq 2^j \text{)}
%\end{align*}
% \end{proof}
 
\begin{lemma}\label{doubledouble}
%good final 7-15-22
Let $M,N\in \bbr$ with $M> 2$. Then $N\cdot\left(2^{M}\right)+ 2M<(N+1)\cdot\left(2^{M}\right)$. 
\end{lemma}
\begin{proof}
Note that $N\cdot\left(2^M\right)+ 2M <(N+1)\cdot\left(2^{M}\right)$ if and only if $2M< 2^{M}$.
%N\cdot\left(2^{2^M}\right)+ 2^{M+1} &<(N+1)\cdot\left(2^{2^M}\right) && \hspace{.5cm} \text{if and only if} \\
%2^{M+1}&< 2^{2^M} && \hspace{.5cm} \text{if and only if} \\
%1 & < 2^{2^M-(M+1)} && \hspace{.5cm} \text{if and only if} \\
%M+1&< 2^M, 
This holds by calculus: $2M= 2^M$ for $M=2$ and the derivatives satisfy $(2M)' =2 < \ln{2}\cdot \left(2^M\right)=\left(2^M\right)'$ when $M\geq 2$.
\end{proof}

\begin{notation}\label{deftow}
For natural numbers $d\geq 1$ and $n=1$ we define 

\[\text{Tower}(d,1):= 2^{3d+1}.\]

\noindent For $d\geq 1$ and $n>1$ we define
\[\text{Tower}(d,n):= 2^{2^{\iddots ^{{3d+n+1}}}},
\]
\noindent where the right-hand side is an exponent tower of height $n+1$ (i.e., the tower consists of nested exponents with $n$ copies of 2 followed by a final exponent $3d+n+1$).
\end{notation} 

\begin{lemma}[{Appendix}]\label{3tow}
%7-15-22 good final
All natural numbers $d,n\geq 1$ satisfy $2^{4\cdot \text{Tower}(d,n)} \leq \text{Tower}(d,n+1)$.
\end{lemma}

\begin{comment}
\begin{proof}
First let $n=1$. Then 

\[2^{4\cdot \text{Tower}(d,1)} = 2^{4\cdot 2^{3d+1}} 
= 2^{2^{3d+1+2}}=  2^{2^{3d+3}}= \text{Tower}(d,2).\]

\noindent Now let $n>1$. We find

\begin{align*}
 2^{4\cdot \text{Tower}(d,n)} &=2^{4\cdot 2^{2^{\iddots ^{{3d+n+1}}}} }\\
 &= 2^{2^{\left(2^{\iddots ^{{3d+n+1}}}+2\right)}} \\
 &\leq  2^{2^{\left(2\cdot 2^{\iddots ^{{3d+n+1}}}\right)}}.
\end{align*}

\noindent The net effect is to add 1 to the exponent of the third 2 from the bottom of the tower. %Multiplying that exponent by 2 instead leaves the lower 2's unchanged and adds 1 to the next higher exponent. 
Repeat the cycle of ``adding 1 to an exponent is less than multiplying the exponent by 2, which adds 1 to the next exponent up" as long as the exponent is 2. This propagates addition of 1 up the tower to the final exponent, whence

\begin{align*}
 2^{4\cdot \text{Tower}(d,n)} &\leq 2^{ 2^{2^{\iddots ^{{3d+n+2}}}} }=\text{Tower}(d,n+1).
\end{align*}
\end{proof}
\end{comment}

\begin{theorem}[{Explicit bounds on degree complexity of ${\tt Triangulate}$}; {Proof in appendix}]\label{expTri}
Let $T(d,n)$ be the recursive function from Theorem \ref{cxtyTri} that bounds the degree complexity of the output of ${\tt Triangulate}$. The following inequality holds for natural numbers $d,n\geq 1$:

\[
T(d,n)<\text{Tower}(d,n).%= 2^{2^{\iddots ^{ 2^{3d+n+1}}}},
\]

\noindent %where the right-hand side is a tower of powers of 2 of height $n+1$.
(See Notation \ref{deftow} for the definition of $\text{Tower}(d,n)$.)
\end{theorem}

\begin{theorem}[{Degree complexity of $\rgaexp$}]\label{cxtyRGA} 
%7-15-22 good final
%6-21-22: good check, including notation improvement
Let $\xpef, A, S$ satisfy the hypotheses of Theorem \ref{modrgacorr}, the termination and correctness result for $\rgaexp$. Suppose further that  each element of $\xpf,A,S$ has degree at most $d$. Then every nondifferential polynomial returned by $\rgaexp(A,S)$ has degree at most $R(d,n)$, defined recursively by $R(d,n)_0:= T(d,n)$, $R(d,n)_{k+1}:= T(d+R(d,n)_k,n)$ for $0\leq k < n$, and $R(d,n):= R(d,n)_n$. (The values $R(d,n)_k$ give bounds on intermediate polynomials' degrees after $\rgaexp$ has recursively called itself $k$ times out of the maximum possible $n$. The outputs containing a proper derivative form a subset of %$\{x_{i_1}'-f_{i_1}(\bx),$ $\ldots, x_{i_s}'-f_{i_s}(\bx)\},$
$\xpf$, so their degree is bounded by $d$.)
 
%Let $A=\{x_{i_1}'-f_{i_1}(\bx),$ $\ldots, x_{i_s}'-f_{i_s}(\bx), p_1,p_2,\ldots,p_m\}$ where $x_{i_j}\neq x_{i_k}$ for $j\neq k$, the $x_{i_j}'-f_{i_j}(\bx)$ are ODEs in explicit form, $\{x_{i_1},\ldots, x_{i_s}\} \subseteq \{x_1,\ldots, x_n\}$, and the $p_1, \ldots, p_m \in \rx$ are nondifferential polynomials. Let $S\subseteq \rx\setminus \{0\}$ be a finite set of nondifferential polynomials. Suppose that each $f_{i_j}(\bx),p_i$, and element of $S$ has degree at most $d$. 

\end{theorem}

\begin{proof}
As noted in the correctness proof (Theorem \ref{modrgacorr}), $\rgaexp$ alternates between calls to ${\tt Triangulate}$ and to itself. The first possible call is to ${\tt Triangulate}$, producing nondifferential polynomials of degree at most $T(d,n)=R(d,n)_0$. Then in step 3 an ODE $x'-f(\bx)$ may be pseudodivided %[TODO just substitute/pdiv reverse but dispose of the ODE instead? Should be OK because monic of deg 1; doesn't increase degree as much] 
by the derivative of a nondifferential polynomial $q$; this at most adds the degrees of the dividend and divisor, yielding $d+T(d,n)$. (Recall that $s_q(x'-f(\bx)) -q' = \widetilde{r}$, where $\widetilde{r}$ is the pseudoremainder. The total degree of  $s_q(x'-f(\bx))$ is at most $d+T(d,n)$. In fact, $d+T(d,n)-1$ is a strict upper bound because the degree of the separant actually goes down, but we ignore this. The proper derivatives in $\widetilde{r}$ all come from $q'$, which has total degree at most $T(d,n)$. After replacing the proper derivatives using the relations $\xpef$ and obtaining $r$, the total degree in the transformed version of $q'$ is also at most $d+T(d,n)$.)

Another call to ${\tt Triangulate}$ yields $T(d+T(d,n),n)=R(d,n)_1$ and restarts the process (beginning with pseudodivision in step 3 since the equations are now triangular and have their separants in the set of inequations). After $k$ iterations have eliminated $k$ ODEs, the degree is at most $R(d,n)_k$, whence pseudodivision and another call to ${\tt Triangulate}$ yield $R(d, n)_{k+1}= T(d+R(d,n)_k,n)$.%//d comes from x'-f(x) and R(d,n)_k comes from the nondiff'l q
\end{proof}

\begin{notation}\label{defrtow}
For natural numbers $d,n\geq 1$ and $0\leq k \leq n$ we define 

\[\text{RTower}(d,n)_k:= 2^{{\iddots ^{{\text{Tower}(d,n+k)}}}},
\]
\noindent where the right-hand side is an exponent tower of height $k(n-1)+1$ (i.e., the tower consists of nested exponents with $k(n-1)$ copies of 2 followed by a final exponent $\text{Tower}(d, n+k)$).

We also define

\[
\text{RTower}(d,n):=\text{RTower}(d,n)_n= 2^{{\iddots ^{{\text{Tower}(d,2n)}}}}
\]
\noindent where the right-hand side is an exponent tower of height $n(n-1)+1$. % (i.e., the tower consists of nested exponents with $n(n-1)$ copies of 2 followed by a final exponent $\text{Tower}(d,2n)$).

\end{notation}

\begin{theorem}[{Explicit bounds on degree complexity of $\rgaexp$; Proof in appendix}]\label{expRGA}
Let $R(d,n)$ be the recursive function from Theorem \ref{cxtyRGA} that bounds the degree complexity of the output of $\rgaexp$. The following inequality holds for natural numbers $d,n\geq 1$:

\[
R(d,n)<\text{RTower}(d,n).%= 2^{2^{\iddots ^{ 2^{3d+n+1}}}},
\]

\noindent %where the right-hand side is a tower of powers of 2 of height $n+1$.
(See Notation \ref{defrtow} for the definition of $\text{RTower}(d,n)$.)
\end{theorem}

After expanding the top exponent $\text{Tower}(d,2n)$ of $\text{RTower}(d,n)$, the bound on $\rgaexp$ becomes an exponent tower with $n(n-1)+2n =n^2+n$ copies of 2 followed by a final exponent $3d+2n+1$.

%[TODO: Discussion of complexity of RGA; compare and contrast to GSP doubly exponential]
\phantomsection
\label{trinotbad}
While the nonelementary bounds in Theorem \ref{expRGA} are enormous from a practical perspective, there are important compensating factors. %around $\rgaexp$. 
Other published versions of RGA \cite[pp. 162-3]{BoulierLOP95} \cite[ p. 587]{GolubitskyKMO08} \cite[p. 111]{rga} use analogues of ${\tt Triangulate}$ (i.e., they are based on recursive splitting and pseudoreduction), so they must have nonelementary \cite[pp. 419-23]{AhoHU74} worst-case complexity like we find for $\rgaexp$ (Theorems \ref{cxtyRGA},\ref{expRGA}). The \emph{Wu-Ritt process}, a close relative of ${\tt Triangulate}$, has nonelementary complexity \cite[p. 121]{GalloM90} similar to what we show in Theorem \ref{expTri}.

In spite of RGA's apparently high worst-case complexity, it has nonetheless been used profitably in applications (p. \pageref{rgaapp}). Likewise, our experience (e.g., Sections \ref{lorenzex1}, \ref{lorenzex2}) shows that RGA has potential for studying algebraic invariants due to its systematic nature and universal results.

More theoretically, the problem of getting a differential radical decomposition of an explicit system with nondifferential inequations is close in spirit to the \emph{non}-differential radical ideal membership problem. The \emph{effective Nullstellensatz} \cite{brownawell1987bounds} shows that radical ideal membership in polynomial rings over fields is strictly easier than general ideal membership in those rings. This opens the door for specialized approaches to radical ideal membership that outperform general \Grob basis methods in terms of asymptotic complexity. (M\"{o}ller and Mora \cite{MollerM84} adapted a previous example \cite{mayr1982complexity} to prove doubly exponential lower bounds for \Grob bases. In particular, they demonstrated ideals with generators of degree $d$ in $n$ variables such that \emph{all} \Grob bases using certain monomial orderings
must have an element of degree doubly exponential in $n$ and polynomial in $d$.) \phantomsection \label{szanto} As mentioned in Remark \ref{difftriangrmk}, Sz\'{a}nt\'{o} \cite{{Szanto97},{szanto1999computation}} developed a %\emph{characteristic set} (closely related to regular chains) 
characteristic set method of singly exponential complexity (in $n$, polynomial in $d$) \cite{amzallag2019complexity} for radical ideal membership. (See also \cite{{BurgisserS09},{boulier2006well},{AubryLM99},{Kalkbrener94}}.) Sz\'{a}nt\'{o}'s method reduces radical ideal membership testing to checking whether polynomials pseudoreduce to zero modulo certain characteristic sets. (In fact, the theory is closely related to that of RGA.)
%that Sz\'{a}nt\'{o}'s algorithm solves with singly exponential complexity. 

\phantomsection \label{replacetri} It may be possible to generalize this to the differential case and prove Lemma \ref{difftriang} while replacing $\tri$ with a method of singly exponential complexity. %(perhaps doubly exponential, since it seems hard to avoid iterating the algebraic elimination $n$ times to handle the $n$ derivatives $x_i'$). 
The results of \cite{novikov1999trajectories} and \cite{GhorbalSP14} imply that the invariant yielded by $\rgaexp$ can be obtained by other techniques with doubly exponential complexity (Section \ref{relatedcxty}). This further supports the idea that some modification of ${\tt Triangulate}$ might dispense with the nonelementary bounds. Lastly, Theorem \ref{radinvartest} and the resulting analysis in Section \ref{check} suggest that the closely related problem of \emph{checking} invariance has a very rare worst case and that singly exponential complexity is actually the norm.

%known complexity around RGA (previous works focused on diff'n/order; trivial here); nonelem is bad, but same as others. However, seems close in spirit to algorithms (just couldn't get over the hump for lemma) that have rel good cxty. By GSP and NovYak we know the problem can be solved by other methods with doubly expon cxty, so maybe some modification/heuristics can help a lot. RGA is used in apps in spite of the cxty issues. Moreover, next section suggests that maybe the application usu. doesn't need horrible cxty; so maybe the worst-case is very special here. Certainly our examples were doable/didn't entirely blow up.

\begin{subsection}{Invariance Checking, Totally Real Varieties, and Complexity}\label{check}

%[TODO: purpose of next two: simple initial check that can discard a lot of noninvariants without needing GB]
%[TODO: issues of completeness and soundness; we're willing to tolerate false negatives but not false positives; we think these should be rare, anyway, assuming prevalence of t.r.]

In this section we connect totally real varieties to invariance checking. %(as opposed to invariant generation like in Theorem \ref{radlddecomp}). %for algebraic sets with respect to a polynomial vector field.  
%The necessary background was already covered by the end of Section \ref{lie}, but we preferred to wait until this point for the following insights because %some of these insights are new (to the best of our knowledge) while 
%Section \ref{rev} is dedicated to review. % of known results. 
%Moreover, 
Theorem \ref{radinvartest} below and the subsequent commentary are also relevant to complexity and bounds.%, topics that came into focus in Section \ref{rgabds}.
  
\begin{proposition}\label{checknorad}
Let $A=\{p_1,\ldots,p_m\}\subseteq \rx$, let $\mathbf{F}$ be a polynomial vector field, and let $\vsr{A}$ be  an invariant set with respect to $\mathbf{F}$. Then $\rrad{(A)}$ is an invariant ideal. If in addition $\vsc{A}$ is totally real, then $\sqrt{(A)}$ is an invariant ideal.
\end{proposition}
%7-16-22 good final
\begin{proof}
By the real Nullstellensatz (Theorem \ref{rnss}) and statement 2 of Theorem \ref{fullinvarcrit}, the ideal $(\ldfid{A})\subseteq \rrad{(A)}$. Since $(A)\subseteq (\ldfid{A})\subseteq \rrad{(A)}$, Proposition \ref{radprop} (3a) shows that $\rrad{(\ldfid{A})}=\rrad{(A)}$. The ideal $(\ldfid{A})$ is closed under Lie differentiation and hence is invariant; it follows by Lemma \ref{rradinvar} that $\rrad{(A)}$ is invariant.

If $\vsc{A}$ is totally real, then by Proposition \ref{radprop} (2),(3b) and Proposition \ref{trrad} we have $\sqrt{(A)}=\rrad{(A)}$, which is invariant by the preceding paragraph.
\end{proof}

\begin{proposition}\label{noninvarsuff}
Let $A=\{p_1,\ldots,p_m\}\subseteq \rx$ and let $\mathbf{F}$ be a polynomial vector field. If $\ldf{p_j}\notin \rrad{(A)}$ for some $1\leq j\leq m$, then $\vsr{A}$ is not an invariant set with respect to $\mathbf{F}$. If $\vsc{A}$ is totally real and $\ldf{p_j}\notin \sqrt{(A)}$ for some $1\leq j\leq m$, then $\vsr{A}$ is not invariant.
\end{proposition}
%7-16-22 final good
\begin{proof}
Immediate from Proposition \ref{checknorad} because the hypotheses of Proposition \ref{noninvarsuff} imply that $\rrad{(A)}$ and $\sqrt{(A)}$, respectively, are not invariant ideals.
\end{proof}

The next theorem generalizes to multiple polynomials Lie's criterion for invariance of smooth algebraic manifolds \cite[Fig. 1]{GhorbalSP14}\cite[Thm. 2.8]{olver2000applications}: let $h\in\rx$ and assume that the real variety $\vsr{h}$ has \emph{no} singular points. Lie's criterion then implies that $\vsr{h}$ is invariant with respect to polynomial vector field $\mathbf{F}$ if the first Lie derivative $\ldf{h}$ vanishes everywhere on $\vsr{h}$. 
\begin{theorem}\label{radinvartest}
Let $A=\{p_1,\ldots,p_m\}\subseteq \rx$ and let $\mathbf{F}$ be a polynomial vector field. If $(A)$ is real radical (i.e., $(A)=\rrad{(A)}$), then $\vsr{A}$ is an invariant set with respect to $\mathbf{F}$ if and only if $\ldf{p_j}\in (A)$ for all $1\leq j\leq m$. Equivalently, if $\vsc{A}$ is totally real and $(A)$ is radical (i.e., $(A)=\sqrt{(A)}$), then $\vsr{A}$ is an invariant set with respect to $\mathbf{F}$ if and only if $\ldf{p_j}\in (A)$ for all $1\leq j\leq m$.
\end{theorem}
%7-16-22 final good; use contrapos of preceding prop for forward direction of this thm
\begin{proof}
Immediate from Corollary \ref{membinvar} and Proposition \ref{noninvarsuff}.
\end{proof}

\begin{corollary}\label{invarcheckcpxty}
 Consider all real algebraic sets that are defined by finite collections $A\subseteq \rx$ such that $\vsc{A}$ is totally real and $(A)$ is radical. There is an algorithm for checking invariance of such sets with respect to polynomial vector fields $\mathbf{F}:=\xpf$ that is polynomial in the maximal degree $d$ of elements of $A$ and $\xpf$ and singly exponential in the number of variables $n$. 
\end{corollary}
\begin{proof}
Since $\vsc{A}$ is totally real and $(A)$ is radical, by Theorem \ref{radinvartest} $\vsr{A}$ is invariant with respect to $\mathbf{F}$ if and only if $\ldf{q} \in (A)$ for all $q\in A$. Apply a test like Sz\'{a}nt\'{o}'s \cite{{Szanto97},{szanto1999computation}} for radical ideal membership that is singly exponential in the number of variables and polynomial in the degree.

\end{proof}
Theorem \ref{radinvartest} raises some important questions. How common is it for $\vsc{A}$ to be totally real and $(A)$ to be radical, and can we check those properties efficiently? Regarding the prevalence of totally real varieties, we make the following observations about Theorem \ref{signchange}, which states that the complex variety corresponding to an irreducible real algebraic variety defined by a single polynomial $p\in \rx$ is totally real if and only if $p$ assumes both positive and negative values as $\bx$ varies over $\bbr^n$. 

Fix the degree $d$. Blekherman (\cite[Thm. 1.1]{blekherman2006there}) has given an asymptotic, probabilistic result showing that for large $n=|\mathbf{x}|$, a randomly chosen polynomial almost certainly attains both positive and negative values. (More precisely, \cite{blekherman2006there} puts a probability measure on the space of polynomials and shows that as $n\rightarrow \infty$, the ratio of the measure of the set of nonnegative polynomials and the measure of the set of all polynomials goes to 0. See the first inequality in Theorem 1.1 of \cite{blekherman2006there}, which gives an upper bound $c_2n^{-1/2}$ on this ratio for some constant $c_2$.) Also, if $n\geq 2$ the space of reducible polynomials of at most a given bounded degree has strictly lower dimension than the space of all polynomials with that bound on the degree (see Example \ref{irredex} below). Hence for large $n$ a randomly chosen polynomial $p\in\rx$ will almost certainly define a totally real complex variety by Theorem \ref{signchange} because $p$ will almost certainly be irreducible and take on both positive and negative values. We suspect that a similar result holds for varieties that are defined by multiple, possibly reducible, polynomials (i.e., the generic situation should be that a complex variety defined over $\bbr$ is totally real).

We are not aware of specific results on the prevalence of radical ideals among all ideals, but we expect that a ``typical" ideal (with respect to some reasonable probability distribution) is radical. Consequently, we hypothesize that for almost all (in a precise sense that would require additional work to specify) complex varieties defined over the reals, checking invariance with respect to a polynomial vector field reduces by Theorem \ref{radinvartest} to checking membership of the first Lie derivatives in the ideal generated by the defining polynomials of the variety.

This leads to an interesting situation with regard to the computational complexity of checking for invariance. Given $A\subseteq \rx$, \emph{if} $\vsc{A}$ is totally real \emph{and} $(A)$ is radical, then we can check invariance of $\vsr{A}$ with complexity singly exponential in the number of variables. As explained in the preceding paragraph, we have reason to believe that this is possible for almost all choices of $A$. 

However, we are left with the ``measure zero" cases where $\vsc{A}$ is not totally real or $(A)$ is not radical. Determining whether a given candidate $A$ has these properties appears, in the worst case, to be doubly exponential in the number of variables. For instance, decomposing a variety into irreducible components (which seems unavoidable for determining if the variety is totally real) has doubly exponential lower bounds in $n$ \cite[Thm. 1]{chistov2009double}. (In particular, Chistov shows that for sufficiently large $n,d$ there are polynomials in $n$ variables of degree less than $d$ defining a reducible variety with a component whose corresponding prime ideal has no set of generators with all elements having degree less than a certain doubly exponential bound.) Likewise, checking if a variety is radical is closely related to \emph{primary ideal decomposition} \cite[Sect. 4.8]{clo1_4}, which has doubly exponential complexity (following again from \cite[Thm. 1]{chistov2009double}).%a standard way to check if a variety is radical uses \emph{primary ideal decomposition}, which is also known to have doubly exponential complexity \cite{todo}. //GTZ88 implies that primary decomp isn't the only way to go for perfect fields and that sqfr factorization once you have I\capR (p)-primary for p\in R, R a PID. However, it seems like getting to such a point requires finding GB in a nonzero dim ideal (certainly GTZ88 8.3 uses this to get down to zero-dim).

In summary, then, we suspect that with singly exponential complexity we can almost always correctly guess whether a given $A$ is invariant, but \emph{proving} that we have the correct answer may require doubly exponential complexity.  (This is similar, but not identical, to the hypothesis--also possibly true--that proving invariance of an algebraic set has singly exponential average-case complexity but doubly exponential worst-case complexity.) However, we cannot discard the possibility that the worst case is also singly exponential; settling the matter appears difficult. % Section \ref{relatedcxty} discusses .

\begin{example}\label{irredex}
%7-16-22 final good
As part of the preceding analysis we claimed that the set of reducible polynomials of at most a given bounded degree has smaller dimension than the set of all such polynomials, and hence there is a precise sense in which almost all polynomials are irreducible. If $n=1$, then over $\bbc$ every polynomial of degree $d> 1$ is reducible, so we only care about the case $n\geq 2$. Rather than prove a formal theorem with full details, we illustrate with a simple case that mirrors the general result (the example is based on an argument from \cite{459881}).

 Let $K$ be an algebraically closed field and consider $P_{2,d\leq 2}\setminus \{0\}\subseteq K[x,y]$, where $P_{2,d\leq 2}\setminus \{0\}$ is the set of all nonzero polynomials over $K$ in two variables having degree at most 2. An arbitrary element of $P_{2,d\leq 2}\setminus \{0\}$ has the form $p=a_1 + a_2x + a_3y +a_4x^2 +a_5xy + a_6y^2$, with not all $a_i$ being zero. We can thus identify $p$ with the 6-tuple $(a_1,\ldots, a_6)\in K^6\setminus\{\mathbf{0}\}$. Irreducibility is not affected by nonzero scalar multiples, so it is natural to view $p$ as a point $(a_1:\cdots : a_6)$ in $\mathbb{P}^5(K)$, the five-dimensional \emph{projective space} over $K$ \cite[Sect. 8.2]{clo1_4}. By definition, $\mathbb{P}^5(K)$ is the quotient of $K^6\setminus\{\mathbf{0}\}$ by the equivalence relation $\sim$ such that $\mathbf{a}\sim\mathbf{b}$ if and only if there exists $0\neq \alpha\in K$ such that $\alpha \mathbf{a}=\mathbf{b}$. (Going forward we omit the $K$ in  $\mathbb{P}^5(K)$.) For instance, both $1+2x^2 -3xy$ and $-2-4x^2+6xy$ correspond to $(1:0:0:2:-3:0)\in \mathbb{P}^5$.
  
  Points of $\mathbb{P}^5$ that correspond to reducible polynomials in $K[x,y]$ are contained in the image of the map $f: \mathbb{P}^2 \times \mathbb{P}^2 \rightarrow \mathbb{P}^5$ defined by 
\[f((b_1:b_2:b_3),(c_1:c_2:c_3)) = (b_1c_1:b_1c_2 + b_2c_1: b_1c_3+b_3c_1: b_2c_2:b_2c_3+b_3c_2:b_3c_3).\]

This point represents the factorization $p=qr$, where $q=b_1+b_2x+b_3y$ and $r=c_1+c_2x+c_3y$. The image of $f$ has dimension at most $4=2+2=\text{ dim}(\mathbb{P}^2) + \text{ dim}(\mathbb{P}^2) = \text{ dim}(\mathbb{P}^2\times \mathbb{P}^2)< \text{ dim}(\mathbb{P}^5) =5$. These values are justified by Examples 1.30, 1.33 (p. 67)  and Theorem 1.25 (p. 75) of \cite{shaf3ed}. Hence every polynomial corresponding to a point in the ``large" set $\mathbb{P}^5\setminus f(\mathbb{P}^2\times\mathbb{P}^2)$ (the complement of a lower-dimensional set) is irreducible.
\end{example}
\begin{remark}\label{eqeffrmk}
We note that the prevalence of totally real and/or irreducible varieties is also relevant to the complexity of generating \emph{equations} for algebraic invariants. In particular, Theorem \ref{invarcor} implies that these and other conditions make it easier to interpret the output of $\rgaexp$. For instance, suppose $\rgaexp$ returns a collection of regular differential systems $(A_1,S_1),\ldots, (A_r,S_r)$ and each $(A_i,S_i)$ satisfies the hypotheses of Theorem \ref{invarcor} (2). Then we may ignore the (doubly exponential) complication of finding generators of the saturation ideals $\sat[S_i]{(A_i\cap \rx)}$ because the invariant is simply $\cup_{i=1} ^r\vsr{A_i\cap \rx}$ and the inequations are superfluous. (Compare to Theorem \ref{radlddecomp} (3).) % and apparently calls for new techniques. 
\end{remark}
%[TODO: discuss possibility that worse-case is also singly exponential; we just couldn't prove that]

\end{subsection}

\section{Related work}\label{related}

%[systematicity: once it terminates, you know you have the answer: don't have to go further; a harder problem than just checking; could have verified \ref{lorenzex1} with simple Lie deriv check (it's a first integral)

%we don't require homog forms like GP14 (and so no worries about homog/dehomog like in GP14); we don't care about the witness to ideal membership, that can get us to completeness faster (in gen'l, our mechanism/use of differential ideals ensures that we get an invariant/something whose lie derivs lie in the original ideal without having to use templates or test membership; membership--in the radical--is automatic); no need to solve determinants ("symbolic linear programming" in GP14 tech report); we avoid nonlinear real arith

%Similarly, I think our advantage over Kong17 and Boreale20 is that we don't need a GB (Kong17 is automatically GB, but bec. restricts to Darboux; Boreale20 requires GB).

%Better than NovYak for decidability of (LD(A)) because applies to multiple generators (not clear that one step of equality implies stabilization...actually, yes, if we use the single-poly track for each one. But that's overkill for multiple, which could terminate sooner because of crosstalk?) and has a priori bounds that should be much lower; no GBs needed. So we improve on NovYak, which improves on HBT, which (along w/ GB, etc.) implies decidability. t.r. stuff explains feeling that worst case bounds are usu very excessive.

\subsection{Invariant Generation and Checking}\label{gencheck}
Several recent algebraic algorithms for generating invariants are close in spirit to our use of $\rgaexp$. In particular, all are \emph{template methods} that find invariants having the form of a parameterized polynomial input (see the example in Section \ref{lorenzex2}). However, $\rgaexp$'s basic elimination method (namely, differential pseudodivision) is unique in the class. This leads to a major advantage of our approach: $\rgaexp$ avoids both \Grob bases and real quantifier elimination. In particular, we obtain the largest possible invariant that is consistent with the input system without having to test ideal membership or the vanishing of Lie derivatives. %not so convenient for checking invar, though, bec

Liu, Zhan, and Zhao \cite{LiuZZ11} proved that invariance of semialgebraic sets (i.e., defined by polynomial equations and inequalities over the real numbers) with respect to a polynomial vector field is decidable and gave an algorithm--referred to here as the LZZ algorithm--for the problem of checking invariance. By enforcing an invariance criterion on a template, the algorithm gives a way of generating invariants of a given form. The chief downside of LZZ is that it uses two expensive tools, \Grob bases and real quantifier elimination. (The generation version only uses real quantifier elimination, but the template might need to be large to find a nontrivial invariant.) %both of which $\rgaexp$ avoids.

%TODO: LZZ11  on ideas from, e.g., Matringe \cite{todo} (lin alg; generalized later by GP14)

	Ghorbal and Platzer  \cite{GhorbalP14} contributed a %a flexible algebraic method \cite{GhorbalP14} %*give p.*
%that has informed work on other important algorithms for invariant generation \cite{Boreale20}. For convenience we refer to the 
procedure that here we refer to as DRI (``differential radical invariant"). %The foundation of DRI is a ... 
Given a homogeneous template polynomial $h$ with parameter coefficients, assume an order $N$ and write a symbolic matrix representing membership of the $N$-th Lie derivative of $h$ in the ideal generated by $h$ and the preceding $N-1$ Lie derivatives. % \cite[Thm. 4]{GhorbalP14}.
 DRI seeks an invariant of maximal dimension with this form by minimizing the rank of the symbolic matrix. %Without doing something like this, one may end up with trivial invariants (e.g., all the parameters are zero).

 	One virtue of DRI is completeness: %(*vet; issue of bounded coeffs*)
any algebraic invariant that exists can be produced by choosing $N$ and the degree of $h$ to be sufficiently large  \cite[Cor. 1]{GhorbalP14}. Another is its reliance on linear algebra (similar to \cite{{RebihaMM15},{MatringeMR10}}, which DRI partly generalizes), an area with many efficient computational tools. DRI does not use \Grob bases, but nonlinear real arithmetic/quantifier elimination is generally required to find parameter values that minimize the rank. %Moreover, DRI achieves completeness while avoiding the bottleneck of computing real radical ideals. %Let $I\subseteq \mathbb{R}[\mathbf{x}]$ be an ideal generated by polynomials with real coefficients. The real radical $\sqrt[\mathbb{R}]{I}$ of $I$ is the ideal of all polynomials $p$ such that for some natural number $s$ and polynomials $g_1,g_2,\ldots, g_s$, the sum $p^{2s}+g_1^2 + g_2^2+\cdots +g_s^2$ belongs to $I$. (This is a \emph{sum-of-squares}, which we henceforth abbreviate by \emph{SOS}.) 
DRI has considerable power, as demonstrated by a number of nontrivial case studies from the aerospace domain \cite[Sect. 6]{GhorbalP14}. Nonetheless, it is significantly slower than the following two algorithms \cite[Sect. 6]{Boreale20}.
			
	Kong et al. \cite{KongBSJH17} use a similar idea for invariant generation, though they restrict themselves to the case of invariants $p$ that are \emph{Darboux polynomials}. %\cite{zhang2017integrability}. 
This means that the Lie derivative $\dot{p}=ap$ for some polynomial $a$ (i.e., $p$ divides its own Lie derivative with respect to the vector field in question). This prevents completeness, but also contributes to better empirical performance when it applies \cite[Table 1]{KongBSJH17}. However, the implementation still uses \Grob bases and nonlinear real arithmetic \cite[Remark 3]{KongBSJH17}. %rmk 3 of Kong
	
	Recent work of Boreale \cite{Boreale20} has strong ties to both DRI \cite{GhorbalP14} and \cite{KongBSJH17}. The POST algorithm \cite[p. 9]{Boreale20} builds on the same theory as \cite{GhorbalP14} but focuses on invariants containing a given initial set. %say something about barrier sets? 
Moreover, POST uses \Grob bases and an iterative process involving descending chains of vector spaces and ascending chains of ideals. In some sense, Boreale's work is intermediate between that of \cite{GhorbalP14} and \cite{KongBSJH17}: POST offers completeness guarantees missing from \cite{KongBSJH17} but in general is slower than \cite{KongBSJH17} and faster than DRI. Boreale also uses POST and a generic point of the initial set to study %safety in 
the semialgebraic case. Boreale extends these ideas to PDEs in \cite{boreale2022automatic} and Boreale and Collodi treat the special case of polynomial PDE conservation laws in \cite{boreale2022linear}.  %In Example \ref{rgasemi} we briefly discuss the relation to our approach.% for such sets.

	In \cite{GhorbalSP14} Ghorbal, Sogokon, and Platzer give an invariant checking algorithm that is related to LZZ but is restricted to algebraic sets and improves efficiency in that case. Like LZZ, the method in \cite{GhorbalSP14} uses \Grob bases and nonlinear real arithmetic. The algorithm computes successive Lie derivatives of the input  $A\subseteq \rx$ and checks ideal membership (where \Grob bases come in) until the first-order Lie derivatives of all preceding elements belong to the ideal generated by those elements. This will eventually happen by Hilbert's basis theorem (Theorem \ref{hbt}), so $A$ and the subsequent Lie derivatives computed up to that stage generate $(\ldfid{A})$ (see our comments immediately preceding Theorem \ref{radlddecomp}). Real arithmetic checks that each Lie derivative computed vanishes on all of $\vsr{A}$. This holds if and only if $\vsr{A}$ is invariant with respect to $\mathbf{F}$ (Corollary \ref{membinvar}).

	 The recent ESE invariant checking algorithm of \cite{GhorbalS22} refines and extends both LZZ and \cite{GhorbalSP14}. ESE again uses \Grob bases and real arithmetic, but it structures quantifier elimination subproblems more efficiently than does LZZ and, unlike \cite{GhorbalSP14}, applies to semialgebraic sets. Based on the topological notion of \emph{exit sets} \cite{conley1978isolated} (which concern vector field trajectories exiting the boundary of a set in $\bbr^n$), the ESE algorithm can verify or disprove invariance of a candidate semialgebraic set. ESE is often much faster than an enhanced version of LZZ that is also described in \cite{GhorbalS22}. %contains two main algorithms: ESE and an improved version of LZZ. Experiments show ESE outperforming even the enhanced LZZ.)
	 
	 We lastly mention a new approach of Wang et al.~\cite{WangCXZK21} for generating \emph{invariant barrier certificates} (a concept closely tied to invariants; see \cite[Thm. 4]{WangCXZK21}). The method employs \Grob bases to find a sufficient number of Lie derivatives and uses them to satisfy conditions \cite[Def. 4]{WangCXZK21} that guarantee a template is an invariant barrier certificate. However, \cite{WangCXZK21} then translates the criterion %, which is reminiscent of Theorem \ref{fullinvarcrit}, 
into a numerical optimization problem. Experiments show promising performance, but the method's reliance on local extrema occasionally results in erroneous output \cite[Table 1]{WangCXZK21}. (Wang et al. provide the option of searching for global extrema using a branch-and-bound technique, but this can increase the complexity substantially.) This contrasts with the trade-offs made by fully symbolic methods like $\rgaexp$ for generating or checking invariants.

\subsection{Complexity}\label{relatedcxty}

%LZZ, GSP14, ESE into one quicker package, that should almost always hold, and give bounds much better than NovYak worst case.

As indicated in Section \ref{gencheck}, the order of Lie derivatives needed to check invariants is an important factor in the complexity of the algorithms from \cite{LiuZZ11,GhorbalSP14,Boreale20,WangCXZK21,GhorbalS22}. The best upper bound known is doubly exponential in the number of variables \cite[Thm. 4]{novikov1999trajectories}, but both \cite{GhorbalS22} and \cite{novikov1999trajectories} suggest this is overly conservative. This intuition is consistent with our discussion in Section \ref{check}, where Theorem \ref{radinvartest} collapses invariance checking to a single round of radical ideal membership testing (assuming a radical ideal that defines a totally real variety; we argued that this should almost always hold). Whether or not the worst-case complexity is greater than singly exponential, it is necessarily considerable since generating algebraic invariant sets is NP-hard \cite[p. 289]{GhorbalP14}.) 

In another direction, Section \ref{rgabds} demonstrates that the complexity of RGA is nontrivial to analyze. %While degree bounds for RGA have not previously appeared, %6-27-22 reconfirmed by looking at notes and abstracts (but see Reviewer 1, comment 1 8-25
Upper bounds on the \emph{order} of intermediate and output differential polynomials were worked out in \cite{GolubitskyKMO08,GustavsonOP18}. %They are large and, in the case of multiple derivations (partial differential polynomials), not even primitive recursive \cite{todo}. 
The order bounds on RGA from \cite{GolubitskyKMO08} amount to $n!$ for ODEs of the form $\mathbf{x'} =\mathbf{f}(\mathbf{x})$ in $n$ variables. However, this result is excessive for $\rgaexp$ because we differentiate at most $n$ times. Moreover, in $\rgaexp$ we only handle differential polynomials of order at most 1 since we use the explicit form to substitute nondifferential polynomials for proper derivatives. 

Degree bounds for differential characteristic sets appear in \cite[Thm. 5.48]{simmons_towsner}, but these are again excessive for our situation.

\section{Conclusion}\label{conc}
%//can be shortish; focus on summarizing contributions.

In this paper we have adapted the Rosenfeld-\Grob algorithm to generate algebraic invariants of polynomial differential systems. Finding invariants is a crucial and computationally difficult task that our method tackles with a new and systematic approach using tools from algebraic geometry, differential algebra, and the theory of regular systems. Our algorithm $\rgaexp$ provides an alternative to the traditional \Grob bases and real quantifier elimination while giving a novel representation of the largest algebraic invariant contained in an input algebraic set. We have also highlighted the prevalence and importance of totally real varieties for reducing the computational complexity of both generating and checking invariants of differential equations.

 The unique aspects of $\rgaexp$ provide multiple directions for future work. As discussed in Remark \ref{difftriangrmk} and Section \ref{rgabds}, p. \pageref{replacetri}, one important goal is modifying the subroutine $\tri$ to preserve differential solutions while obtaining a regular algebraic decomposition with singly exponential complexity. Beyond theoretical complexity improvements, our framework would greatly benefit from effective heuristics (e.g., for detecting totally real varieties) and optimizations like pruning superfluous branches in $\tri$. Some rankings are much easier to compute than others \cite{boulier2000efficient}, so translating from one differential ranking to another could be an important strategy. Partial runs of $\rgaexp$ are another interesting possibility. In our experience, the final iterations of the algorithm tend to be drastically more expensive than the early ones while yielding less additional information (e.g., giving relatively trivial, low-dimensional chains when we already know the generic chain). The example in \ref{lorenzex2} demonstrates this; we are most interested in the generic Chain 2 and would like the option of only computing it.
 
The treatment of parameters is also a critical matter. Our example used auxiliary variables with constant derivative, but this is suboptimal. Fakouri et al. \cite{fakouri2018new} describe a modified version of RGA designed to more efficiently identify all possible regular chains corresponding to different values of parameters. %(see also boulier18 \cite{todo}). 
Adapting this algorithm to our problems could be important for a more scalable $\rgaexp$. %Ideally, an implemented and optimized version of $\rgaexp$ could extend the reach of invariant generation tools like Pegasus. %Possibly this could involve sacrificing completeness 

Dong et al. \cite{dong2023differential} take another approach to differential elimination and parameters. There the focus is on isolating input-output relations and structural identifiability rather than invariants.  Like our work, \cite{dong2023differential} focuses heavily on nondifferential ideals and their relation to explicit differential systems. However, the systems treated there have a more special shape (state-space form) than ours. Their algorithms are also probabilistic. Nonetheless, those algorithms yield significant improvements in performance over standard differential elimination and there are tantalizing hints of deeper connections to the approach we have taken in this paper. (For instance, in most cases the representations found in \cite{dong2023differential} are actually characteristic sets.) Elucidating these connections and generalizing the novel techniques from \cite{dong2023differential} could be very fruitful for the elimination approach to invariant sets.

While purely algebraic invariants are important for continuous dynamical systems and enjoy strong decidability properties \cite{PlatzerT20}, many polynomial vector fields have no such invariants that are nontrivial (not isolated points, coordinate axes, etc.). The Van der Pol oscillator is an example that nonetheless has interesting semialgebraic invariants \cite{JonesP19,odani1995limit}. %\cite{GhorbalP14} cites the semialgebraic case as an important target for future work on complete methods for invariant generation. We propose using %elimination-theoretic tools like 
In the literature there are several encodings of inequalities as equations using auxiliary variables (e.g., \cite{PlatzerQR09, anderson1977output}). Our initial investigation suggests that this idea, together with RGA, can yield insights into semialgebraic invariants of nonlinear systems with non-obvious dynamics.

In light of these prospects, we view our results in this paper as laying the groundwork for an enhanced method that produces easily verified invariants, combines optimal theoretical complexity with strong practical performance, and is broadly applicable. Progress on these fronts could lead to inclusion in practical tools for CPS verification such as the theorem prover \KX \cite{FultonMQVP15} with its automatic invariant-generating utility, Pegasus \cite{SogokonMTCP21}.

\paragraph{Acknowledgements}
We are grateful to Fran\c{c}ois Boulier, Khalil Ghorbal, and Yong Kiam Tan for helpful exchanges. We also thank the reviewers for their very thorough and thoughtful suggestions. This work was supported by the AFOSR under grants FA9550-18-1-0120 and FA9550-16-1-0288.

\section{Appendix} \label{appendix}
\setcounter{theorem}{14}
\begin{lemma}
%10-20-21 careful check
Let $K$ be a field. 
\begin{enumerate}
\item If $X\subseteq K^n$ and $p\in K[\mathbf{x}]$, then $p$ vanishes at every point of $X$ if and only if $p$ vanishes at every point of the Zariski closure $\overline{X}^K$ of $X$; i.e., $\vi{K}{X} = \vi{K}{\overline{X}^K}$.
\item Given $X,Y\subseteq K^n$, we have $\overline{X}^K=\overline{Y}^K$ if and only $\vi{K}{X}=\vi{K}{Y}$.
\end{enumerate}
\end{lemma}

\begin{proof}
1. In both cases $\mathbf{V}_{K}(p)$ is a Zariski-closed set containing $X$.

\noindent 2. If $\overline{X}^K=\overline{Y}^K$, then $\vi{K}{X}=\vi{K}{\overline{X}^K}=\vi{K}{\overline{Y}^K}=\vi{K}{Y}$ by 1. Conversely, suppose $\vi{K}{X}=\vi{K}{Y}$. Then similarly $\vi{K}{\overline{X}^K}=\vi{K}{\overline{Y}^K}$ by 1. Since $\overline{X}^K$,$\overline{Y}^K$ are Zariski closed, we have $\overline{X}^K=\vs[K]{\vi{K}{\overline{X}^K}}=\vs[K]{\vi{K}{\overline{Y}^K}}=\overline{Y}^K$ by Lemma \ref{viva}.
\end{proof}

\setcounter{theorem}{16}
\begin{lemma}
%10-20-21 321 careful check
Let $X\subseteq \mathbb{R}^n$. Then the real Zariski closure $\overline{X}^{\bbr}$ equals $\overline{X}^{\mathbb{C}}\cap\mathbb{R}^n$, the restriction of the complex Zariski closure to the reals.
\end{lemma}

\begin{proof}
$\subseteq$: We must show that $\overline{X}^{\mathbb{C}}\cap\mathbb{R}^n$ contains $X$ and is real Zariski closed. Containment of $X$ is clear. For real Zariski closedness, note that  $\overline{X}^{\mathbb{C}}\cap\mathbb{R}^n$ is the zero set of the real and imaginary parts of the defining polynomials of $\overline{X}^{\mathbb{C}}$. (That is, suppose $\overline{X}^{\mathbb{C}}=\vs[\bbc]{A}$ for some $A\subseteq \cx$. For all $p\in A$, distribute over the real and imaginary parts of each monomial's coefficient to write $p=p_1 +ip_2$ where $p_1,p_2\in \rx$. We only care about real solutions $\mathbf{a}\in \bbr^n$ here, so $p(\mathbf{a})=0$ if and only if $p_1(\mathbf{a})=p_2(\mathbf{a})=0$.)

\noindent $\supseteq$: We must show that $\overline{X}^{\mathbb{C}}\cap\mathbb{R}^n$ is contained in every real Zariski closed set $Y$ that contains $X$. Observe that $\vs[\bbc]{\vi{\bbr}{Y}}$ is a complex Zariski closed set containing $Y$ and thus $X$. Then $\overline{X}^{\mathbb{C}}\subseteq \vs[\bbc]{\vi{\bbr}{Y}}$ and so $\overline{X}^{\mathbb{C}}\cap\mathbb{R}^n\subseteq \vs[\bbc]{\vi{\bbr}{Y}}\cap \bbr^n =  \vs[\bbr]{\vi{\bbr}{Y}} =Y$, with the last equality following from Lemma \ref{viva}.  %The other containment is immediate from the definitions.
\end{proof}

\setcounter{theorem}{25}
\begin{lemma}
%andre reviewed and I did, too.
Let $I \,\unlhd \, \mathbb{C}[\mathbf{x}]$ and let $0\notin S\subseteq \bbc[\bx]$ be finite. Then
\begin{enumerate}
\item $\sat{I} =\sat[(\mathrm{\Pi}  S)]{I}$ and 
 \item $\mathbf{V}_{\mathbb{C}}(I:S^{\infty})=\vsc{\sat[(\mathrm{\Pi}  S)]{I}}=\overline{\mathbf{V}_{\mathbb{C}}(I)\setminus \mathbf{V}_{\mathbb{C}}(\mathrm{\Pi}  S)}^{\mathbb{C}}$.
 \end{enumerate}
\end{lemma}
\begin{proof}
\begin{enumerate}

\item $\subseteq$: If $p\in \sat{I}$, then $(s_1^{k_1}s_2^{k_2}\cdots s_m^{k_m})p\in I$ for some $s_j\in S$ and $k_j\geq 0$. Multiplying by appropriate powers of the $s_j$ to equalize the exponents, we see that $(\mathrm{\Pi}  S)^{max\{k_1,\ldots,k_m\}}p\in I$ (recall that ideals are closed under multiplication by anything). Hence $p\in \sat[(\mathrm{\Pi}  S)]{I}$.

\smallskip
\noindent $\supseteq$: This direction holds because $\mathrm{\Pi}  S$ is in the multiplicative set generated by $S$.

\item Part 1 implies the first equality. For the last equality:

$\subseteq$: Let $\mathbf{a}\in \vsc{\sat[(\mathrm{\Pi}  S)]{I}}$. By definition of Zariski closure, it suffices to show that $q(\mathbf{a})=0$ for every $q\in \cx$ that vanishes everywhere on $\vsc{I}\setminus \vsc{\mathrm{\Pi}  S}$.

We claim that $q\in \sqrt{\sat[(\mathrm{\Pi}  S)]{I}}$. This is equivalent to $(\mathrm{\Pi}  S)q \in \sqrt{I}$. (By definition, $q\in \sqrt{\sat[(\mathrm{\Pi}  S)]{I}}$ means that $(\mathrm{\Pi}  S)^Mq^N \in I$ for some natural numbers $M,N$. Multiplying by appropriate powers of $\mathrm{\Pi}  S$ and $q$ as in the proof of part 1, we have $((\mathrm{\Pi}  S)q)^{max\{M,N\}}\in I$, whence  $(\mathrm{\Pi}  S)q \in \sqrt{I}$.) By the Nullstellensatz (Theorem \ref{nss}), we must prove that $(\mathrm{\Pi}  S)q$ vanishes at every point of $\vsc{I}$. But this is true since %$q\in \vi{\bbc}{\vsc{I}\setminus \vsc{\mathrm{\Pi}  S}}$ and thus 
$q$ vanishes everywhere on $\vsc{I}$ except possibly at some points of $\vsc{\mathrm{\Pi}  S}$ (where $\mathrm{\Pi}  S$ vanishes by definition). Hence $q\in \sqrt{\sat[(\mathrm{\Pi}  S)]{I}}$ and so $q(\mathbf{a})=0$.

\smallskip
\noindent $\supseteq$: Since $\mathbf{V}_{\mathbb{C}}(I:(\mathrm{\Pi}  S)^{\infty})$ is $\bbc$-Zariski closed, it already contains $\overline{\mathbf{V}_{\mathbb{C}}(I)\setminus \mathbf{V}_{\mathbb{C}}(\mathrm{\Pi}  S)}^{\mathbb{C}}$ \emph{if} $\mathbf{V}_{\mathbb{C}}(I:(\mathrm{\Pi}  S)^{\infty})$ contains the smaller set $\mathbf{V}_{\mathbb{C}}(I)\setminus \mathbf{V}_{\mathbb{C}}(\mathrm{\Pi}  S)$. Hence  we only need to show that $p(\mathbf{b})=0$ for every $p\in \sat[(\mathrm{\Pi}  S)]{I}$ and $\mathbf{b}\in \vsc{I}\setminus \vsc{\mathrm{\Pi}  S}$. This holds because $(\mathrm{\Pi}  S)^Mp\in I$ for some $M$ and so $((\mathrm{\Pi}  S)^Mp)(\mathbf{b}) =0$; we have $(\mathrm{\Pi}  S)(\mathbf{b})\neq 0$ by assumption, so it must be that $p(\mathbf{b})=0$.
\end{enumerate}
\end{proof}

\setcounter{theorem}{26}
\begin{theorem}[{Hilbert's Nichtnullstellensatz;} {Table \ref{galois}}]
Let $A\subseteq \cx$, and let $0\notin S\subseteq \cx$ be finite. A polynomial $p\in\cx$ vanishes at every complex solution of $(A=0,S\neq 0)$ if and only if $p\in\sqrt{\sat{(A)}}$.
\end{theorem}

\begin{proof}
The $(\subseteq)$ case in the proof of Lemma \ref{satgeom} (2) (in particular, the argument there showing $q\in \sqrt{\sat[(\mathrm{\Pi}  S)]{I}}$) establishes that vanishing at every complex solution of ($A=0,S\neq 0$) implies membership in $\sqrt{\sat{(A)}}$. 

For the other direction, if $p\in\sqrt{\sat{(A)}}$, then $p^N\in \sat{(A)}$ for some natural number $N$. Substitute this $p^N$ for the $p$ that appears in the proof of the $(\supseteq)$ case in Lemma \ref{satgeom} (2). Follow the argument given there to show that $p^N$, and hence $p$, vanishes at every complex solution $\mathbf{b}$ of $(A=0,S\neq 0)$.  % \vsc{A}\setminus \vsc{\mathrm{\Pi}  S}$.
\end{proof}

\setcounter{theorem}{40}
\begin{proposition}
For any $A\subseteq \rx$, there exists a finite set $B\subseteq \rx$ such that  $\mathbf{V}_{\mathbb{C}}(B)$ is totally real and $\mathbf{V}_{\bbr}(A)=\mathbf{V}_{\bbr}(B)$.
\end{proposition}
%10-19-21 pretty careful check
%10-21-21 careful check
%//another way of stating it: BrakeEt16 num alg geom paper; apparently folklore/common knowlege

%//omit if never end up using
\begin{proof}

By Hilbert's basis theorem there is a finite generating set $B$ of $\rrad{(A)}$, the real radical of the ideal generated by $A$. The definition of the real radical implies that $\mathbf{V}_{\bbr}(A)=\mathbf{V}_{\bbr}(\rrad{(A)})= \vsr{B}$. 

To show that $\vsc{B}$ is totally real, we must establish $\cl{\vs[\bbr]{B}}{\bbc}= \vs[\bbc]{B}$.  By Lemma \ref{vanishclos} (2) it suffices to show that any $p\in \cx$ that vanishes everywhere on $\vsr{B}$ also vanishes on $\vsc{B}$ (i.e, $\vi{\bbc}{\vsr{B}}\subseteq \vi{\bbc}{\vsc{B}}$; the reverse containment is immediate).  Splitting $p=p_1+ip_2$ into real and imaginary parts (see the proof of Lemma \ref{restrictclos}), it follows from the real Nullstellensatz (Theorem \ref{rnss}) that both $p_1,p_2$ vanish everywhere on $\vsr{B}$ and so must belong to $\rrad{(B)}=(B)$ (since $(B)=\rrad{(A)}$ is real radical). %\rrad{(A)}$, whence $p\in (\rrad{(A)})_{\bbc}=(B)_{\bbc}$. 
Thus $p_1,p_2$ (and consequently $p$) vanish at each point of $\vsc{B}$.%and the proof is complete.
\end{proof}

\setcounter{theorem}{43}
\begin{lemma}
Let $A,B \subseteq \rx$. If $\vsc{A}$ is totally real, then $\vsc{A}\setminus \vsc{B}$ is a totally real constructible set. \end{lemma}
% 7-5-22 careful mental
%10-21-21 323-4 pretty careful

\begin{proof}
By Lemma \ref{vanishclos} (2) it suffices to show that if $p\in \cx$ vanishes on $\vsr{A}\setminus \vsr{B}$, then $p$ vanishes on $\vsc{A}\setminus \vsc{B}$. Let $\mathbf{a}\in \vsc{A}\setminus \vsc{B}$ with the goal of showing $p(\mathbf{a})=0$. Note that for every $q\in B$ the product $pq$ vanishes on all of $\vsr{A}$ since $p$ vanishes on  $\vsr{A}\setminus \vsr{B}$ and $q$ vanishes on $\vsr{B}$. It follows by Lemma \ref{vanishclos} (1) that $pq$ vanishes on $\vsc{A}$ because $\vsc{A}$ is totally real by assumption and so $\vsr{A}$ is Zariski dense in $\vsc{A}$. Since $\mathbf{a}\in \vsc{A}\setminus \vsc{B}$, we have $\tilde{q}(\mathbf{a})\neq 0$ for some $\tilde{q}\in B$. Thus $p(\mathbf{a})\tilde{q}(\mathbf{a})=0$ and $p(\mathbf{a})=0$ as needed.
\end{proof}

\setcounter{theorem}{47}
\begin{theorem}[{Characterization of algebraic invariants }{\cite[Lemma 5]{Boreale20}}{\cite[Lemma 2.1]{christopher2007inverse}}]
Let $\mathbf{F}$ be a polynomial vector field and let $X\subseteq \mathbb{R}^n$ be real Zariski closed. The following are equivalent:
\begin{enumerate}
\item $X$ is an algebraic invariant set of $\mathbf{F}$.
\item For all $p_1,\ldots,p_m \in \rx$ such that $X=\mathbf{V}_{\bbr}(p_1,\ldots,p_m)$, we have $\mathcal{L}_{\mathbf{F}}^{(k)}(p_i)(\mathbf{a})=0$ for all $k\geq 0$, $1\leq i \leq m$, and $\mathbf{a} \in X$.
\item There exists $I\unlhd \rx$ such that $X=\mathbf{V}_{\bbr}(I)$ and $I$ is an invariant ideal with respect to $\mathbf{F}$ (i.e., $\ldf{I}\subseteq I$).
\end{enumerate}
\end{theorem}
\begin{proof}
\noindent $(1)\Rightarrow (2)$: Part (ii) of \cite[Lemma 2.1]{christopher2007inverse} proves that if $X=\mathbf{V}_{\bbr}(p_1,\ldots,p_m)$ is invariant with respect to $\mathbf{F}$, then $\mathcal{L}_{\mathbf{F}}(p_i)\in \sqrt{(p_1,\ldots, p_m)}$ for $1\leq i \leq m$. Hence $\mathcal{L}_{\mathbf{F}}(p_i)$ vanishes at every point of $X$. Reapply part (ii) of \cite[Lemma 2.1]{christopher2007inverse} to $\mathbf{V}_{\bbr}(p_1,\ldots,p_m,$ $\mathcal{L}_{\mathbf{F}}(p_1),$ $\ldots, \mathcal{L}_{\mathbf{F}}(p_m))$, which is still the same invariant $X$. This shows that second-order Lie derivatives of the $p_i$ vanish on $X$. Continue the process for any order $k$.  (See also \cite[Thm. 1]{GhorbalP14} for a variant of $(1)\Leftrightarrow (2)$.)
\smallskip

\noindent $(2)\Rightarrow (3)$: Assuming (2), if $A=\{p_1, \ldots, p_m\}$ then we have $\mathbf{V}_{\bbr}(p_1,\ldots,p_m)=\mathbf{V}_{\bbr}((\mathcal{L}^*_{\mathbf{F}}(A)))$; as noted above, $(\mathcal{L}^*_{\mathbf{F}}(A))$ is an invariant ideal.
\smallskip

\noindent $(1)\Leftrightarrow (3)$: This is precisely \cite[Lemma 5]{Boreale20}.
\end{proof}

\setcounter{theorem}{48}
\begin{corollary}
Let $\mathbf{F}$ be a polynomial vector field and let $X=\vsr{p_1,\ldots, p_m}$. If $\ldf{p_i}\in (p_1,\ldots,p_m)$ for all $1\leq i\leq m$, then $X$ is an algebraic invariant set of $\mathbf{F}$.
\end{corollary}

\begin{proof}
By the sum and product rules (Lemma \ref{ldderiv}), for $I=(p_1,\ldots, p_m)$ to be an invariant ideal (and hence satisfy statement 3 of Theorem \ref{fullinvarcrit}) it suffices that the Lie derivatives of the generators $p_1,\ldots, p_m$ belong to $I$.
\end{proof}

\setcounter{theorem}{70}
\begin{proposition}
Let ${\tt DiffPseudoDiv}$$(p,q)=r$. Then for some $\widetilde{s}$ a product of factors of $s_q$, $\widetilde{i}$ a product of factors of $i_q$, and  $\widetilde{q}\in [q]$ we have $(\widetilde{s})(\widetilde{i})p-\widetilde{q} = r$. In particular, we have $p\in \sat[\{s_q,i_q\}]{[q]}$ if $r$ is 0. 

In the nondifferential case (or if $p$ contains no proper derivatives of the leader $l_q$ of $q$) we have $(\widetilde{i})p-\widetilde{q} = r$, with $\widetilde{q}\in (q)$ now, and $p\in \sat[\{i_q\}]{(q)}$ if $r$ is 0.  %//there's some ambiguity between (A)\subseteq K[x] and (A)\subseteq K{x}. I don't think it can cause any problems (I checked where pdivsat is used), so leave ambig.
\end{proposition}

\begin{proof}
This follows from the form of operations in ${\tt DiffPseudoDiv}$. If $p$ contains a highest proper derivative $(l_{q})^{(k)}$ of the leader of $q$, then the first intermediate pseudoremainder (from the call  ${\tt DiffPseudoDiv}$$(p,q^{(k)})$ using step 2 is 
 \begin{equation*}\label{singlesteprep} \widetilde{r}:= (s_q/g)p -(c/{g})((l_{q})^{(k)})^{d-1}q^{(k)},\end{equation*}
 
\noindent where $s_q$ is the initial of $q^{(k)}$, $d$ is the highest power of $(l_{q})^{(k)}$ in $p$,  $c$ is the coefficient of $((l_{q})^{(k)})^d$ in $p$,  and $g$ is the GCD of $s_q$ and $c$. Note that $\widetilde{r}$ has the claimed property. (Consequently, Proposition \ref{pdivsat} also applies to a single step of pseudodivision as defined in Remark \ref{premred}.) By induction on $k$ and $d$ we may assume that $(s_q^*)(i_q^*)\widetilde{r} -q^* = r$ for some $s_q^*$ a product of factors of $s_q$, $i_q^*$ a product of factors of $i_q$, and $q^*\in [q]$. % as do subsequent pseudoremainders computed in the course of ${\tt DiffPseudoDiv}$. By
Multiplying $\widetilde{r}$ by $(s_q^*)(i_q^*)$ and subtracting $q^*$, the two preceding equations imply that 

\[(s_q^*)(s_q/g)(i_q^*)p -((s_q^*)(i_q^*)(c/{g})((l_{q})^{(k)})^{d-1}q^{(k)}+q^*)=r,\]

\noindent which has the correct form. 

Analogous arguments hold for the case that $p$ contains no proper derivative of $l_q$ and the nondifferential case. The assertions for $r=0$ follow from the definition of a saturation ideal. (Multiply both sides by appropriate factors of $s_q,i_q$ so that  $\widetilde{s}$ and $\widetilde{i}$ become powers of $s_q,i_q$ and not just arbitrary products of their factors.) %by clearing denominators in $\widetilde{s}$ and $\widetilde{i}$ and using the definition of a saturation ideal.% See \cite[Sect. I.9]{KolchinDAAG} for a more general version. [TODO: check Kolchin ref]
\end{proof}

\setcounter{theorem}{99}
\begin{lemma}
 For all $j\in \bbn$, the $j$-th Fibonacci number $F_j$ satisfies $F_j< 2^j$.
 \end{lemma}
 %7-15-22 good final
 \begin{proof}
 \begin{align*}
F_{j} \,&= \, \frac{\phi_1^{j} - \phi_2^{j}}{\sqrt{5}} \,= \, \frac{\left(\frac{1+\sqrt{5}}{2}\right)^{j} - \left(\frac{1-\sqrt{5}}{2}\right)^{j}}{\sqrt{5}}
&& \hspace{.5cm} \text{(Binet's formula, Proposition \ref{Binet})} \\
&< \frac{2^j +1}{\sqrt{5}} && \hspace{.5cm} \text{(} \phi_1<2, |\phi_2| <1\text{)}\\
&< 2^j. && \hspace{.5cm} \text{(}  \sqrt{5}>2, 1\leq 2^j \text{)}
\end{align*}
\end{proof}

\setcounter{theorem}{102}
\begin{lemma}
%7-15-22 good final
All natural numbers $d,n\geq 1$ satisfy $2^{4\cdot \text{Tower}(d,n)} \leq \text{Tower}(d,n+1)$.
\end{lemma}

\begin{proof}
First let $n=1$. Then 

\[2^{4\cdot \text{Tower}(d,1)} = 2^{4\cdot 2^{3d+1}} 
= 2^{2^{3d+1+2}}=  2^{2^{3d+3}}= \text{Tower}(d,2).\]

\noindent Now let $n>1$. We find

\begin{align*}
 2^{4\cdot \text{Tower}(d,n)} &=2^{4\cdot 2^{2^{\iddots ^{{3d+n+1}}}} }\\
 &= 2^{2^{\left(2^{\iddots ^{{3d+n+1}}}+2\right)}} \\
 &\leq  2^{2^{\left(2\cdot 2^{\iddots ^{{3d+n+1}}}\right)}}.
\end{align*}

\noindent The net effect is to add 1 to the exponent of the third 2 from the bottom of the tower. %Multiplying that exponent by 2 instead leaves the lower 2's unchanged and adds 1 to the next higher exponent. 
Repeat the cycle of ``adding 1 to an exponent is less than multiplying the exponent by 2, which adds 1 to the next exponent up" as long as the exponent is 2. This propagates addition of 1 up the tower to the final exponent, whence

\begin{align*}
 2^{4\cdot \text{Tower}(d,n)} &\leq 2^{ 2^{2^{\iddots ^{{3d+n+2}}}} }=\text{Tower}(d,n+1).
\end{align*}
\end{proof}

\setcounter{theorem}{103}
\begin{theorem}[{Explicit bounds on degree complexity of ${\tt Triangulate}$}]
Let $T(d,n)$ be the recursive function from Theorem \ref{cxtyTri} that bounds the degree complexity of the output of ${\tt Triangulate}$. The following inequality holds for natural numbers $d,n\geq 1$:

\[
T(d,n)<\text{Tower}(d,n).%= 2^{2^{\iddots ^{ 2^{3d+n+1}}}},
\]

\noindent %where the right-hand side is a tower of powers of 2 of height $n+1$.
(See Notation \ref{deftow} for the definition of $\text{Tower}(d,n)$.)
\end{theorem}
%7-15-22 good final
\begin{proof}
We induct on $n$, starting with $n=1$. By definition, $T(d,1)= (F_{2\cdot T(d,0) +1})\cdot T(d,0)=F_{2d+1}d$. We observe the following:

\begin{align*}
F_{2d+1} d  &< 2^{2d+1} d && \hspace{.5cm} \text{(Lemma \ref{fibbd})}\\
&<  2^{2d+1}2^d \, = 2^{3d+1} \,=\text{Tower}(d,1).\\
\end{align*}
\noindent This proves the base case. Now suppose the inequality holds for $n=k$; i.e., $T(d,k)< \text{Tower}(d,k)$. %=2^{2^{\iddots ^{ {3d+k+1}}}},
%\noindent where the right-hand side is a tower of powers of 2 of height $k+1$. 
The inductive case is similar but we now have an additional power of 2, allowing us to use Lemma \ref{doubledouble}. %For notational convenience we define $\text{Tower}(d,0):=3d+1$, the exponent of $2$ in $\text{Tower}(d,1)$. Thus for all $n\geq 1$ we have [TODO: put the right expression here.] 
By definition, $T(d,k+1)= (F_{2\cdot T(d,k) +1})\cdot T(d,k)$. We obtain the following:

\begin{align*}
 (F_{2\cdot T(d,k) +1})\cdot T(d,k)  &< 2^{\left(2\cdot \text{Tower}(d,k)+1\right)} \cdot 2^{\log_2{\text{Tower}(d,k)}} && \hspace{.5cm} \parbox{3cm}{(Lemma \ref{fibbd} and inductive hypothesis)} \\
 &= 2^{\left(2\cdot \text{Tower}(d,k)+1 +\log_2{\text{Tower}(d,k)}\right)}  &&\\
 &<  2^{\left(2\cdot  2^{\log_2{\text{Tower}(d,k)}}+2\log_2{\text{Tower}(d,k)}\right)}  && \hspace{.5cm} \parbox{3cm}{($\log_2{\text{Tower}(d,k)}>2$)}\\
 &< 2^{\left(3\cdot  2^{\log_2{\text{Tower}(d,k)}} \right)} && \hspace{.5cm} \text{(Lemma \ref{doubledouble})}\\
&= 2^{\left(3\cdot \text{Tower}(d,k)\right)}\\
 &< \text{Tower}(d,k+1). &&\hspace{.5cm} \text{(Lemma \ref{3tow})}
\end{align*}

\noindent This completes the proof.
\end{proof}

\setcounter{theorem}{106}
\begin{theorem}[{Explicit bounds on degree complexity of $\rgaexp$}]
Let $R(d,n)$ be the recursive function from Theorem \ref{cxtyRGA} that bounds the degree complexity of the output of $\rgaexp$. The following inequality holds for natural numbers $d,n\geq 1$:

\[
R(d,n)<\text{RTower}(d,n).%= 2^{2^{\iddots ^{ 2^{3d+n+1}}}},
\]

\noindent %where the right-hand side is a tower of powers of 2 of height $n+1$.
(See Notation \ref{defrtow} for the definition of $\text{RTower}(d,n)$.)
\end{theorem}
%7-15-22 good final
%5-27-22 464-5 work

\begin{proof}
We use induction on $k$ for $0\leq k\leq n$ to bound $R(d,n)_k$ (recall that $R(d,n):= R(d,n)_n$. For $k=0$ we have $R(d,n)_0 := T(d,n)<\text{Tower}(d,n)$ by Theorem \ref{expTri}; this equals $\text{RTower}(d,n)_0$, which is a tower with $0=0(n-1)$ copies of 2 followed by exponent $\text{Tower}(d,n+0)$.

Now let $0\leq k<n$ and suppose $R(d,n)_k<\text{RTower}(d,n)_k$. We show that $R(d,n)_{k+1}<\text{RTower}(d,n)_{k+1}$.
\begin{align*}
R(d,n)_{k+1} &= T(d+R(d,n)_k,n) && \\
&< \text{Tower}(d+R(d,n)_k,n) && \hspace{.5cm} \text{(Theorem \ref{expTri})}\\
&<\text{Tower}(d+\text{RTower}(d,n)_k,n) && \hspace{.5cm} \parbox{5cm}{(inductive hypothesis; \text{Tower} is an increasing function of both inputs)}\\
&=2^{{\iddots ^{{3(d+\text{RTower}(d,n)_k)+\alpha}}}} && \hspace{.5cm} \parbox{5cm}{($n$ copies of 2 (hence at least one); $\alpha:=1$ if $n=1$ and $\alpha:=n+1$ if $n>1$)}\\
&=2^{{\iddots ^{{3\cdot \text{RTower}(d,n)_k+3d +\alpha}}}} && \hspace{.5cm}\\ %\parbox{5cm}{}\\
&<2^{{\iddots ^{{4\cdot \text{RTower}(d,n)_k}}}} && \hspace{.5cm}  \parbox{5cm}{(unfold the definitions to see that $\text{RTower}(d,n)_k>3d+\alpha$)}\\ %\parbox{5cm}{(by Lemma \ref{doubledouble}; this applies because at least one 2 appears in $\text{Tower}(d,n+k)$ and $\text{Tower}(d,n+k)$ has final exponent $3d+\alpha$ if $n=1$ and $3d+n+k+1\geq 3d+\alpha$ otherwise)}\\
&= 2^{{\iddots ^{{4\cdot \left(2^{{\iddots ^{{\text{Tower}(d,n+k)}}}}\right)}}}} && \hspace{.5cm} \parbox{5cm} {(where the tower replacing $\text{RTower}(d,n)_k$ has $k(n-1)$ copies of 2 followed by final exponent $\text{Tower}(d,n+k))$}\\
&\leq 2^{{\iddots ^{{\left(2^{{\iddots ^{{4\cdot \text{Tower}(d,n+k)}}}}\right)}}}} && \hspace{.5cm} \parbox{5cm}{(clear since $\text{Tower}(d,n+k)>1$)}\\
&\leq 2^{{\iddots ^{\text{Tower}(d,n+k+1)}}} && \hspace{.5cm} \parbox{5cm}{(Lemma \ref{3tow} applies because the lower tower has at least one copy of 2; this reduces the copies of 2 by one and leaves $n-1+k(n-1)=(k+1)(n-1)$ copies)}\\
&= \text{RTower}(d,n)_{k+1}.
\end{align*}

\end{proof}
\end{document}